\documentclass[12pt]{article}

\usepackage{amsmath,amsthm,amssymb,mathrsfs,bbm}
\usepackage[toc, page]{appendix}
\usepackage{geometry,pdflscape,setspace}
\usepackage{graphicx,subfig}
\usepackage{fancyhdr}
\usepackage{authblk}
\usepackage[toc, page]{appendix}
\usepackage[normalem]{ulem}
\usepackage[format=hang, textfont={md}, labelfont={bf}, justification=raggedright]{caption}
\usepackage{natbib,threeparttable,dcolumn,multicol,lscape,bigstrut,booktabs,multirow,colortbl,array}
\usepackage[colorlinks=true,citecolor=blue,linkcolor=blue,%
    anchorcolor=black,urlcolor=blue,pdfborder={0 0 0},pdfstartview=FitH]{hyperref}

\usepackage{tikz}
\usepackage{circuitikz}
\usetikzlibrary{plotmarks,arrows,calc,patterns}

\usepackage{amsmath}

\graphicspath{{./Figure/}}

\theoremstyle{definition}
\newtheorem{theorem}{Theorem}%[section]
\newtheorem{lemma}{Lemma}%[section]
\newtheorem{corollary}{Corollary}%[section]

\newtheorem{remark}{Remark}%[section]
\newtheorem{example}{Example}%[section]

\newtheorem{assumption}{Assumption}%[section]
\newtheorem{definition}{Definition}%[section]

\newcommand{\indep}{\perp \!\!\! \perp}

\newcommand{\asto}{\overset{a.s.}{\to}}

\begin{document}
\title{Unobserved Heterogeneous Spillover Effects in Instrumental Variable Models \thanks{I am deeply grateful to my advisor, D\'esir\'e K\'edagni, for his invaluable guidance and support. I also sincerely thank my committee members, Andrii Babii, Jacob Kohlhepp, Adam Rosen, and Valentin Verdier, for their constructive comments and continuous support. I am thankful to Muyang Ren, Tianqi Li, and participants in the UNC Econometrics Workshop and Seminars, Duke Microeconometrics Breakfast, and Triangle Econometrics Conference for helpful discussions and suggestions. This research uses data from Add Health. No direct support was received from grant P01 HD31921 for this analysis. All remaining errors are my own.}}
\author{Huan Wu \thanks{Department of Economics, University of North Carolina at Chapel Hill, the United States. Email address: \href{mailto: huan.wu@unc.edu}{huan.wu@unc.edu}}}
%\affil{University of North Carolina at Chapel Hill}
\date{\textbf{Job Market Paper} \\ 
    This version: November 11, 2025 \\
    \vspace{0.2em}
    \href{https://huanwu-econ.github.io/files/JMP_Wu.pdf}{[Link to the latest version]}}
\maketitle
\onehalfspacing
% Can also choose \singlespacing or \doublespacing
%\doublespacing

%%%%%%%%%%%%%%%%%%%%%%%%%%%%%%%%%%%%%%%%%%%%%%%%%%%%%%%%%%%%%%%%%%%%%%%%
%%%% Abstract 
\begin{abstract}

\begin{footnotesize}
    %This paper develops a general framework for identifying causal effects in settings with spillovers, where both outcomes and endogenous treatment decisions are influenced by peers within a known group. It introduces \textit{the generalized local average controlled spillover and direct effects (LACSEs and LACDEs)}, which extend the local average treatment effect framework to settings with spillovers and establish sufficient conditions for their point identification without restricting the cardinality of the instrumental variable support. These conditions clarify the necessity of commonly imposed restrictions to achieve point identification with binary instruments in related studies. The paper then defines \textit{the marginal controlled spillover and direct effects (MCSEs and MCDEs)}, which naturally extend the marginal treatment effect framework to settings with spillovers and are nonparametrically point identified from continuous variation in instruments. These marginal effects serve as building blocks for a broad class of policy-relevant treatment effects, including some casual spillover parameters in the related literature. The paper develops semiparametric estimators and establishes their asymptotic properties, and further proposes a parametric approach for practical implementation. An application using Add Health data illustrates the framework and reveals heterogeneity in education spillovers within best-friend networks.
    This paper develops a general framework for identifying causal effects in settings with spillovers, where both outcomes and endogenous treatment decisions are influenced by peers within a known group. It introduces \textit{the generalized local average controlled spillover and direct effects (LACSEs and LACDEs)}, which extend the local average treatment effect framework to settings with spillovers and establish sufficient conditions for their point identification without restricting the cardinality of the support of instrumental variables. These conditions clarify the necessity of commonly imposed restrictions to achieve point identification with binary instruments in related studies. The paper then defines \textit{the marginal controlled spillover and direct effects (MCSEs and MCDEs)}, which naturally extend the marginal treatment effect framework to settings with spillovers and are nonparametrically point identified from continuous variation in instruments. These marginal effects serve as building blocks for a broad class of policy-relevant treatment effects, including some causal spillover parameters in the related literature. Semiparametric and parametric estimators are developed, and an application using Add Health data reveals heterogeneity in education spillovers within best-friend networks.
\end{footnotesize}
\end{abstract}

\maketitle
{\footnotesize \textbf{Keywords}: Unobserved heterogeneous spillover/direct effect; violation of SUTVA; causal inference; instrumental variable.

%\textbf{JEL subject classification}: C21, C31, C14, C36
}

\newpage

\section{Introduction}
In econometric analyses of treatment effects, the Stable Unit Treatment Value Assumption (SUTVA) is typically imposed, requiring that each individual's potential outcomes depend only on their own treatment assignment and not on the treatment assignments of others. This assumption, however, may fail in environments involving social interactions or group structures, where one unit's treatment can influence another's outcome. In such contexts, the SUTVA is unlikely to hold. In addition, treatment assignment may be endogenous, particularly in observational studies or randomized experiments with imperfect compliance, which further complicates the identification of causal parameters.

This paper develops a new framework for identifying causal effects in environments with within-group spillovers and endogenous treatments, such as education decisions among friends or pricing choices in oligopolistic markets. The framework explicitly accounts for two sources of SUTVA violations: individual outcomes may depend on peers' treatment selection, and individual treatment decisions may be influenced by instruments assigned to other group members.

To begin, the paper introduces \textit{the generalized local average controlled spillover and direct effects (LACSEs and LACDEs)}, which measure treatment effects from peers and from one's own treatment, respectively, for specific subpopulations. These parameters extend the local average treatment effect (LATE) framework \citep{imbens1994identification} by allowing for spillovers in both outcomes and endogenous treatment decisions. This paper is the first to formally establish conditions under which the LACSEs and LACDEs can be point identified for specific subpopulations, and to derive general point identification results that do not rely on the cardinality of the support of instrumental variables. The results characterize the precise variation in instruments required to achieve point identification of local effects in environments with within-group spillovers and endogenous treatments.

When the instrumental variables exhibit continuous variation, the analysis introduces \textit{the marginal controlled spillover and direct effects (MCSEs and MCDEs)}, which measure the effects of peers' and individuals' own treatments conditional on specific values of unobserved characteristics within the group. These parameters naturally extend the marginal treatment effect (MTE) framework \citep{heckman2001policy,heckman2005structural} to settings with spillovers. The paper first establishes the nonparametric point identification of the MCSEs and MCDEs from continuous variation in instruments, without imposing functional form restrictions on the outcome equation or the joint distribution of unobserved characteristics across group members, thereby accommodating flexible forms of spillovers. Similar to the standard MTE, the MCSEs and MCDEs serve as building blocks for identifying a broad class of policy-relevant treatment effects (PRTEs), including the LACSEs and LACDEs, as well as other PRTEs arising from counterfactual policy changes, facilitating policy evaluation in environments with spillovers.

%Additionally, our framework allows the unobserved factors that determine treatment assignment and outcomes to be arbitrarily correlated across group members, capturing the dependence that naturally arises within networks. By relaxing SUTVA and  incorporating both within-group spillovers and correlated unobservables, our approach provides a more credible and flexible foundation for causal inference in group-based settings.

\subsection*{Organization of the Paper}

Section~\ref{sec:disc_treatment} develops the framework with within-group spillovers and endogenous treatments, formally defining, identifying, and analyzing the causal parameters of interest. 
Section~\ref{sec:basic_setting} introduces an outcome model that allows each unit's potential outcome to depend flexibly on the entire vector of treatments within the group, capturing spillover effects from peers' treatment decisions on own outcomes. Treatment selection follows a single-index threshold-crossing structure, in which an individual receives treatment if her unobserved characteristic falls below a threshold function determined by her own and her peers' instrumental variables. This structure, which can be interpreted as the equilibrium behavior of a simultaneous incomplete-information game \citep{aradillas2010semiparametric}, captures how peers' instruments can influence individual treatment decisions. Importantly, this framework imposes no parametric restrictions on the outcome equation or the threshold function and allows for arbitrary dependence among unobserved characteristics across group members, accommodating a broad range of environments in which spillovers may be present.

Section~\ref{sec:local_average_effects} defines and identifies two causal parameters, the generalized local average controlled spillover and direct effects. The term ``generalized local" indicates that these effects are defined for specific subpopulations of groups, while ``spillover" and ``direct" refer to the sources of treatment variation from peers and from the individual herself, respectively. The term ``controlled" highlights that one treatment dimension, either own or peer, is held fixed when measuring the effect of the other. This section establishes general conditions under which the LACSEs and LACDEs are point identified, without requiring the instrumental variables to be either discrete or continuous, as long as they generate the required variation for identification. These identification conditions also clarify the rationale for additional restrictions, such as one-sided noncompliance, which are often imposed to achieve point identification with binary instruments (e.g., \cite{vazquez2023causal}).

Section~\ref{sec:id_mcse_msde} establishes that when instrumental variables exhibit continuous variation, the marginal spillover effect and marginal direct effect are nonparametrically point identified without requiring functional form assumptions on the outcome equation. In addition, the joint distribution of unobserved characteristics across group members is nonparametrically identified over the support of the observed treatment probabilities, without imposing parametric restrictions on how these unobserved factors are distributed. The MCSEs and MCDEs are defined analogously to the LACSEs and LACDEs but condition on a specific realization of unobserved characteristics within each group. By conditioning on the latent characteristics of all group members, these parameters flexibly capture heterogeneity in both direct and spillover effects that arise from variation in unobserved factors. Section~\ref{sec:prte} further demonstrates that the marginal controlled effects form the basis for identifying a broad class of policy-relevant treatment parameters. By integrating the MCSEs and MCDEs over appropriate regions of the unobserved heterogeneity distribution, one can recover the LACSEs, LACDEs, and other treatment effect parameters associated with counterfactual policy interventions.

Section~\ref{sec:compare_mte} formally compares the MCSEs and MCDEs with the standard MTE and shows that, in the presence of spillovers, the conventional MTE may lose its causal interpretation, whereas in the absence of spillovers, the MCSEs and MCDEs coincide with the standard MTE. These results demonstrate that the MCSE-MCDE framework provides a natural generalization of the MTE framework to accommodate environments with spillovers. Section~\ref{sec:test_implication} derives testable implications implied by the model structure and the identification assumptions.

Section~\ref{sec:est_inference} develops a semiparametric estimation procedure for the MCSEs and MCDEs, extending the framework of \citet{carneiro2009estimating} to accommodate within-group spillovers. The proposed approach mitigates the curse of dimensionality associated with covariates while maintaining the model's nonparametric flexibility, as the key structural components other than the covariate adjustment are left unrestricted. The section also establishes the asymptotic properties of the semiparametric estimators. Because these estimators converge at nonparametric rates, their finite sample precision may be limited in small samples or when groups include a large number of members. To address this concern, a complementary parametric framework is introduced, relying on intuitive assumptions that enable straightforward implementation and facilitate valid inference through nonparametric bootstrap methods.

Section~\ref{sec:simulation_application} presents both parametric simulation results and an empirical application.  Section~\ref{sec:para_simulation} presents Monte Carlo simulation results for the parametric estimation procedure, demonstrating the strong finite-sample performance of the proposed parametric methods. %both in estimation accuracy and in the coverage of confidence intervals. 
Section~\ref{subsec:para_application} implements the proposed framework empirically using the parametric procedure. The analysis examines how education attainment affects long-term earnings within best-friend groups, drawing on data from the National Longitudinal Study of Adolescent to Adult Health (Add Health). The results indicate positive dependence between friends' unobserved characteristics and reveal systematic heterogeneity in the marginal controlled direct and spillover effects. The estimated MCDEs of completing 16 years of education are significantly positive when the best friend has also attained this level of education across most values of the latent characteristics, but become statistically insignificant when the friend has not. Similarly, the estimated MCSEs are significantly positive for individuals who completed 16 years of education across most values of the latent characteristics, whereas for those who did not, the spillover effects are insignificant and even negative for certain ranges of unobserved heterogeneity. These results provide empirical evidence of heterogeneous spillover effects of education on long-term earnings within friendship networks, highlighting how their magnitude and direction depend on both individuals' and peers' education attainment.

%On the one hand, the MCDEs of completing 16 years of education are significantly positive when the best friend has also completed 16 years of education, across most selected values of latent characteristics. However, when the best friend has not completed 16 years, the MCDEs become statistically insignificant across all evaluated levels of unobservables. On the other hand, for those who attained 16 years of education, the MCSEs are significantly positive across most values of latent characteristics. In contrast, for individuals who did not complete 16 years, the spillover effects become statistically insignificant, and even negative at certain values of latent characteristics. These findings provide evidence of spillover effects of education attainment on long-term earnings within best-friend networks and underscore that both the direction and magnitude of these effects vary systematically with the education status of individuals and their best friends.

The framework can be extended to accommodate additional settings. Section~\ref{sec:extension} generalizes the analysis to cases where outcomes depend on an exposure mapping, which is a known function of group members' treatment statuses, rather than the full treatment vector. This extension is particularly relevant when group sizes vary or are large, making the full treatment representation impractical. The section formally defines the MCSEs and MCDEs under this extended setting and establishes their nonparametric point identification using continuous instrumental variables.

\subsection*{Related literature}

Recent research has devoted increasing attention to the identification and estimation of treatment effects in the presence of spillovers. This paper contributes to several key strands within this growing body of work.

A common strategy for addressing interference has been to impose parametric structures on social interactions. For instance, \citet{manski1993identification} discussed the linear-in-means model, formulated as a system of linear simultaneous equations to capture endogenous, exogenous, and correlated peer effects. Building on this result, subsequent work, such as \citet{bramoulle2009identification} and \citet{blume2015linear}, extended the framework to more complex forms of interaction within linear models and derived conditions under which social effects can be identified. However, they fundamentally rely on correct parametric assumptions regarding the structure of social interactions. Such assumptions may lead to model misspecification, particularly in the presence of nonlinear spillovers or heterogeneity across individuals. The framework developed in this paper departs from such reliance on parametric restrictions by studying identification under a nonparametric structure in both the outcome equation and the treatment selection mechanism. This design accommodates flexible and potentially complex spillover mechanisms and provides a robust framework for causal analysis in environments with within-group interactions.

In the main setting considered, an individual's outcome depends on the full vector of treatments within the group, consistent with the treatment response framework of \citet{manski2013identification}. Within the context of randomized controlled trials (RCTs), \citet{hudgens2008toward} and \citet{aronow2017estimating}, along with related studies, formalized design-based frameworks for analyzing interference. The framework in this paper extends this line of research  by allowing for noncompliance, so that individuals may not adhere to their assigned treatments. This feature is important in observational studies, where treatment status is not fully controlled by the researcher, and in experimental settings where imperfect compliance may occur. The analysis adopts a large-sample framework rather than a design-based approach to study causal identification under endogenous treatment selection.

\citet{vazquez2023causal} employed a potential outcomes framework to analyze similar settings with spillovers operating through both outcomes and treatment selection, using a binary instrumental variable for identification. His approach classifies individuals into discrete compliance types according to how their treatment choices respond to changes in instruments and focuses on identifying local average spillover and direct effects for each specific type, which are closely related to the LACSEs and LACDEs introduced in this paper. He achieved point identification by excluding certain subpopulations under one-sided noncompliance, a restriction also used in related work such as \citet{ditraglia2023identifying}. This paper establishes general identification conditions for the LACSEs and LACDEs, which clarify why one-sided noncompliance is required for point identification when the instrumental variable is binary. 
This paper further employs a continuously distributed instrumental variable to point identify the marginal controlled direct and spillover effects, defined conditional on continuous realizations of latent characteristics within groups. This approach connects the analysis to the marginal treatment effect literature and establishes the marginal effects as fundamental components for identifying a wide class of policy-relevant treatment effects. In particular, aggregating these marginal effects recovers the local average direct and spillover effects in \citet{vazquez2023causal}, as well as other causal parameters under counterfactual policy interventions.

Recent studies, including \citet{balat2023multiple} and \citet{hoshino2023treatment}, use instrumental variable methods to identify spillover effects in settings with direct strategic interactions among agents, where each individual's treatment choice directly depends on the treatment decisions of other group members. Frameworks with direct strategic interactions assume that an individual's treatment decision does not directly depend on the instruments assigned to other group members, thereby ruling out spillovers from peers' instruments in the treatment selection process.
In contrast, the framework developed here does not model explicit strategic interactions in treatment choices but allows each individual's treatment decision to depend on instruments assigned to other group members. This structure can be interpreted as the equilibrium outcome of a simultaneous incomplete-information game, following \citet{aradillas2010semiparametric}, and thus provides a complementary perspective to models that incorporate direct strategic interaction.

The spillover framework developed in this paper and the multivalued treatment framework are not nested. When the group is treated as a single decision-making unit, the group treatment vector can be reformulated as a multivalued group-level treatment. This links the setting to multivalued MTEs such as \citet{lee2018identifying}. When applied to spillover contexts, identification in \citet{lee2018identifying} relies on an exclusion restriction that an individual's treatment does not depend on peers' instruments, whereas the framework developed here allows and models such spillovers from peers' instruments. 

\section{Model} \label{sec:disc_treatment}

\subsection{Setting} \label{sec:basic_setting}
I consider a sample of $G$ independent and identically distributed (i.i.d.) groups, indexed by $g = \{1, \cdots, G\}$. Each group consists of the same number of units, denoted by $n \geq 2$. For example, a group may correspond to a market with several competing firms or to a household with multiple members. Within each group, units are indexed by $i = \{0, \cdots, n-1\}$. Throughout, I assume that spillover effects operate only within groups and do not extend across groups.

Researchers are often interested in how a treatment affects an outcome. Let $Y_{ig}$ denote the outcome of interest for unit $i$ in group $g$, and let $\mathcal{Y}$ denote its support. In some settings, the outcome $Y_{ig}$ may depend not only on unit $i$'s own treatment status but also on the treatment choices of other units within the same group. For example, a firm's market share is influenced both by its own pricing decisions and by those of its competitors. Within a friendship network, an individual's labor earnings may be influenced by her best friend's education attainment, not only through direct support or access to resources, but also through information-sharing or social learning mechanisms that facilitate the transmission of knowledge about opportunities, norms, and strategies. In such contexts, the Stable Unit Treatment Value Assumption (SUTVA) may be violated, which motivates researchers to develop models that explicitly allow for spillover effects in outcomes.

The binary treatment decision of unit $i$ in group $g$ is denoted by $D_{ig} \in \{0,1\}$, where $D_{ig} = 1$ indicates that unit $i$ adopts the treatment and $D_{ig} = 0$ otherwise. In many applications, treatment decisions are not randomly assigned but instead depend on unobserved characteristics that also influence outcomes, giving rise to endogeneity concerns. For instance, a firm's pricing decision or an individual's education choice may both be endogenously determined. In group settings, units may make their decisions simultaneously, taking into account private information as well as expectations about the behavior of other group members. Each unit's decision depends on its own private information, denoted by $V_{ig}$, as well as on its expectations about the probability that other members of the group will adopt the treatment. Units form their expectations on the basis of publicly observed variables $(Z_{ig}, Z_{-ig})$, where $Z_{ig}$ denotes the random assignment received by unit $i$ in group $g$, and $Z_{-ig}$ denotes the assignments of the remaining group members. The vector $(Z_{ig}, Z_{-ig})$ thus serves as the set of instrumental variables for addressing endogeneity in treatment decisions. Since each unit's treatment choice may respond to the assignments received by other group members, these instruments can also induce spillover effects in treatment selection.

Building on the setting described above, consider the following model for unit $i$ in group $g$, where the peer of unit $i$ is denoted by $-i$. For clarity of exposition, I focus on the case in which each group $g$ consists of two units, indexed by $i = \{0,1\}$, while noting that the identification and estimation results extend straightforwardly to groups with more than two members: 

\begin{equation} \label{eq:basic_model_disc_trt}
    \left\{\begin{array}{l}
        Y_{ig} = m_i(D_{ig}, D_{-ig}, U_{ig}, U_{-ig}) \\
        D_{ig} = \mathbbm{1}\left\{V_{ig} \leq h_{ig}(Z_{ig}, Z_{-ig})\right\}.
        \end{array}\right.
\end{equation}

The first line of Equation \eqref{eq:basic_model_disc_trt} specifies the outcome equation. In this framework, unit $i$'s outcome $Y_{ig}$ depends on her own treatment $D_{ig}$ and on her group member's treatment, $D_{-ig}$, which explicitly models spillover effects. Importantly, I also allow $Y_{ig}$ to depend on both unit $i$'s own unobservables $U_{ig}$ and the unobservables of her group member, $U_{-ig}$. For example, in the oligopoly market, this specification captures the possibility that firm $i$'s market share $Y_{ig}$ is influenced not only by its own unobserved product characteristics $U_{ig}$ but also by the unobserved product characteristics of its rival, $U_{-ig}$.

Throughout the paper, I define the potential outcome for unit $i$ when her treatment is set to $d$ and her group member's treatment to $d'$ as $Y_{ig}(d, d') \equiv m_i(d, d', U_{ig}, U_{-ig})$. The observed outcome $Y_{ig}$ is determined according to the following equation,
\begin{equation*}
    \begin{aligned}
        Y_{ig} =& \big(Y_{ig}(1, 1) D_{-ig} + Y_{ig}(1, 0) (1 - D_{-ig})\big) D_{ig} \\
        &+ \big(Y_{ig}(0, 1) D_{-ig} + Y_{ig}(0, 0) (1 - D_{-ig})\big) (1 - D_{ig}).
    \end{aligned}
\end{equation*}
This paper focuses on identifying reduced-form causal effects arising from a unit's own treatment and peers' treatments, rather than the underlying structural parameters specified in structural equations. A detailed comparison with a system of structural equations is provided in the Appendix \ref{app:model_endo_effect}.

The framework imposes no functional form restrictions on the outcome equation $m_i$, and the subscript $i$ indicates that each group member, $i \in \{0,1\}$, may have a distinct outcome equation, meaning their functional forms are not required to be identical. It also places no restrictions on the dimension of the unobserved components $(U_{ig}, U_{-ig})$. Consequently, the influence of the peer's treatment $D_{-ig}$ on unit $i$'s outcome $Y_{ig}$ remains fully unrestricted. This generality provides a flexible structure that accommodates complex and heterogeneous spillover patterns in outcomes.

The second line of Equation \eqref{eq:basic_model_disc_trt} characterizes the treatment selection mechanism. The treatment decision of unit $i$, $D_{ig}$, may be endogenous because it is determined by a continuous unobserved factor $V_{ig}$ that can also influence the outcome. The treatment selection $D_{ig}$ depends only on the unit's own unobservable $V_{ig}$ and not directly on her group member's unobservable $V_{-ig}$. This restriction is plausible in many applications. For example, in the oligopoly market discussed above, a firm's pricing decision is driven by its own private demand shock, while the competitor's demand shock is unobserved and therefore cannot directly affect the firm's decision rule. In the returns to education example, an individual's education decision is determined solely by her own education costs. The best friend's education costs, which are unobserved to the individual, do not directly influence her schooling decision.

Crucially, it is empirically reasonable to allow the unobserved factors $V_{ig}$ and $V_{-ig}$ to be arbitrarily dependent, since group members often share related unobserved characteristics or are exposed to common shocks. This dependence further complicates identification, and the framework accommodates it without imposing parametric restrictions on the joint distribution of unobservables. This flexibility accommodates a wide range of empirically relevant correlations. In the oligopoly setting, correlation across firms' idiosyncratic shocks arises naturally. For instance, a market-wide change in consumer tastes or a new advertising regulation may simultaneously affect how all products are perceived by consumers, thereby inducing correlation between the demand shocks $V_{ig}$ and $V_{-ig}$.

I model the treatment selection mechanism using a single threshold crossing rule: unit $i$'s chooses to take the treatment, $D_{ig} = 1$, if the unobserved factor $V_{ig} \in \mathbb{R}$ does not exceed a threshold $h_i(Z_{ig}, Z_{-ig})$, where $ h_i: \mathbb{R}^{k_i} \times \mathbb{R}^{k_{-i}} \mapsto \mathbb{R} $ is an unspecified function. I do not impose a parametric form on the threshold function $h_i$, and the subscript $i$ emphasizes that its functional form may differ across units $i \in \{0,1\}$ within the same group. In contrast to complete information games, where unit $i$'s treatment $D_{ig}$ directly depends on the treatment decision of her peer $D_{-ig}$, my framework is consistent with a simultaneous-move game with incomplete information, as studied by \cite{aradillas2010semiparametric} and related papers. In this setting, $D_{ig}$ corresponds to player $i$'s action, and $(Z_{ig}, Z_{-ig})$ represent publicly observed signals that serve as instrumental variables. Each unit observes $(Z_{ig}, Z_{-ig})$ and forms beliefs about the joint treatment choices within the group, specifically the probability $\mathbb{P}(D_{ig} = 1, D_{-ig} = 1 \mid Z_{ig}, Z_{-ig})$, and then chooses her treatment based on these beliefs. Thus, treatment decisions are interdependent through expectations rather than through observing others' realized treatment choices. As shown in \citet{aradillas2010semiparametric}, the optimal decision rule in such simultaneous-move incomplete information game is consistent with a single-index threshold-crossing structure of the form in Equation~\eqref{eq:basic_model_disc_trt}. Appendix \ref{app:incomplete_info_game} provides a detailed discussion.

For identification, which I discuss in detail later, the instrumental variables $(Z_{ig}, Z_{-ig})$ must be independent of the unobserved heterogeneity $(U_{ig}, U_{-ig}, V_{ig}, V_{-ig})$ in the group and must not directly affect the outcomes $(Y_{ig}, Y_{-ig})$. In studies that focus on complete information setting, such as those analyzed by \cite{balat2023multiple} and \cite{hoshino2023treatment}, strategic interactions between $D_{ig}$ and $D_{-ig}$ are modeled explicitly, but unit $i$'s treatment is not allowed to depend on her group member's instrument $Z_{-ig}$. A distinguishing feature of this framework is that it does not rely on any additional exclusion restrictions on instruments: I allow the treatment $D_{ig}$ to depend on both the unit's own instrument $Z_{ig}$ and the peer's instrument $Z_{-ig}$, thereby accommodating potential spillovers from instruments into treatment decisions. Moreover, the instrumental variables may take the form of common public signals observed by all group members, so that the same variable $Z$ serves as the instrument for each member, or unit-specific instruments that vary across members, $Z_{ig} \neq Z_{-ig}$. This framework accommodates both shared and individual sources of exogenous variation.

Figure \ref{fig:causal_relation} presents a directed acyclic graph (DAG) that illustrates the causal relationships among the key variables within group $g$. The red arrows represent spillover channels: unit $i$'s outcome $Y_{ig}$ may depend on her peer's treatment $D_{-ig}$, and her treatment $D_{ig}$ may depend on her peer's instrument $Z_{-ig}$. Direct interaction between treatments $D_{ig}$ and $D_{-ig}$, however, is ruled out. The unobserved heterogeneity $V_{ig}$ and $V_{-ig}$ introduce endogeneity, as they may simultaneously affect both treatments and outcomes. Those are represented by the black dashed arrows. The blue dashed arrow reflects potential dependence between $V_{ig}$ and $V_{-ig}$, for which I do not impose any functional restrictions.

\begin{figure}[!ht]
    \centering
    \resizebox{0.3\textwidth}{!}{%
    \begin{circuitikz}
    \tikzstyle{every node}=[font=\LARGE]
    \draw [ line width=0.6pt ] (5.75,13.25) circle (0cm);
    \draw [ line width=0.6pt ] (6.25,13) circle (0.75cm) node {\LARGE $Y_i$} ;
    \draw [ line width=0.6pt ] (6.25,10.5) circle (0.75cm) node {\LARGE $D_i$} ;
    \draw [ line width=0.6pt ] (6.25,8) circle (0.75cm) node {\LARGE $Z_i$} ;
    \draw [ line width=0.6pt ] (6.25,15.5) circle (0.75cm) node {\LARGE $V_i$} ;
    \draw [ line width=0.6pt ] (10,15.5) circle (0.75cm) node {\LARGE $V_{-i}$} ;
    \draw [ line width=0.6pt ] (10,13) circle (0.75cm) node {\LARGE $Y_{-i}$} ;
    \draw [ line width=0.6pt ] (10,10.5) circle (0.75cm) node {\LARGE $D_{-i}$} ;
    \draw [ line width=0.6pt ] (10,8) circle (0.75cm) node {\LARGE $Z_{-i}$} ;
    \draw [line width=0.6pt, ->, >=Stealth] (6.25,8.75) -- (6.25,9.75);
    \draw [line width=0.6pt, ->, >=Stealth] (6.25,11.25) -- (6.25,12.25);
    \draw [line width=0.6pt, ->, >=Stealth] (10,11.25) -- (10,12.25);
    \draw [line width=0.6pt, ->, >=Stealth] (10,8.75) -- (10,9.75);
    \draw [ color={rgb,255:red,255; green,38; blue,0}, line width=0.8pt, ->, >=Stealth] (7,10.5) -- (9.5,12.5);
    \draw [ color={rgb,255:red,255; green,38; blue,0}, line width=0.8pt, ->, >=Stealth] (9.25,10.5) -- (6.75,12.5);
    \draw [ color={rgb,255:red,255; green,38; blue,0}, line width=0.8pt, ->, >=Stealth] (9.25,8) -- (6.75,10);
    \draw [ color={rgb,255:red,255; green,38; blue,0}, line width=0.8pt, ->, >=Stealth] (7,8) -- (9.5,10);
    \draw [line width=0.6pt, <->, >=Stealth] (7,8) -- (9.25,8);
    \node [font=\LARGE] at (8,5.5) {\textit{Group g = 1, ..., G i.i.d.}};
    \node [font=\LARGE] at (8,5.5) {};
    \draw [ color={rgb,255:red,4; green,51; blue,255}, line width=1pt, <->, >=Stealth, dashed] (7,15.5) -- (9.25,15.5);
    \draw [line width=0.8pt, ->, >=Stealth, dashed] (6.25,14.75) -- (6.25,13.75);
    \draw [line width=0.8pt, ->, >=Stealth, dashed] (10,14.75) -- (10,13.75);
    \draw [line width=0.8pt, ->, >=Stealth, dashed] (5.5,15.25) .. controls (4.5,13) and (4.25,13) .. (5.5,10.75);
    \draw [line width=0.8pt, ->, >=Stealth, dashed] (10.75,15.25) .. controls (11.75,13) and (11.75,13) .. (10.75,10.75);
    \draw [line width=0.8pt, ->, >=Stealth, dashed] (6.75,15) -- (9.25,13.25);
    \draw [line width=0.8pt, ->, >=Stealth, dashed] (9.5,15) -- (7,13.25);
    \end{circuitikz}
    }%
    \caption{Causal relations under spillover setting}
    \label{fig:causal_relation}
\end{figure}
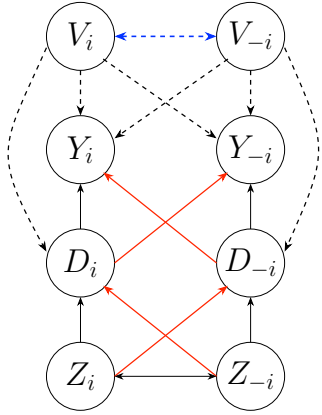

Assumptions \ref{as:RA}-\ref{as:Vdist} set out the maintained restrictions on the key variables that are imposed throughout the paper.

\begin{assumption}(Random assignment) \label{as:RA}
    The instrumental variables assigned to all members within a group are jointly independent of the group's unobserved heterogeneity:
    \begin{equation*}
        \big(Z_{ig}, Z_{-ig}\big) \indep \big(V_{ig}, V_{-ig}, U_{ig}, U_{-ig}\big)
    \end{equation*}
    for $i, -i \in \{0, 1\}$.
\end{assumption}

Assumption \ref{as:RA} requires that instruments are randomly assigned at the group level, implying that the group-level instrument vector $(Z_{ig}, Z_{-ig})$ is independent of the unobserved heterogeneity of all units in the group. This assumption places no restrictions on the dependence structure between $Z_{ig}$ and $Z_{-ig}$ within a group. The instruments may be arbitrarily correlated across units within a group, as long as they remain jointly independent of the unobserved heterogeneity $(V_{ig}, V_{-ig}, U_{ig}, U_{-ig})$.

\begin{assumption}(Exclusion restriction) \label{as:er}
    Given $d_0, d_1$ and $u_0, u_1$, the instrumental variables $(Z_{ig}, Z_{-ig})$ do not directly affect the outcome $Y_{ig}$:
    \begin{equation*}
        m_{i}\big(d_0, d_1, z_0, z_1, u_0, u_1\big) = m_{i}\big(d_0, d_1, z_0', z_1', u_0, u_1\big)
    \end{equation*}
    for any $z_0 \neq z'_0$ and $z_1 \neq z'_1$.
\end{assumption}

Assumption \ref{as:er} requires that the instruments affect the outcome only through their influence on treatment take-up, without exerting any direct effect on the outcome. This condition corresponds to the standard exclusion restriction commonly imposed in instrumental variable analyses.

\begin{assumption}(Distribution of $V_{ig}$) \label{as:Vdist}
    The unobserved variable $V_{ig}$ is continuously distributed.
\end{assumption}

Assumption \ref{as:Vdist} requires that the unobserved heterogeneity $V_{ig}$ has a continuous distribution, which is a common condition in the literature. Under this assumption, $V_{ig}$ can be normalized to follow a uniform distribution on the interval $(0,1)$.

Another implicit restriction embedded in the treatment selection equation is a monotonicity structure. Specifically, consider the case in which instruments $Z_{ig}$ and $Z_{-ig}$ are binary, taking values in $\{z_0, z_1\}$. If the threshold function satisfies the following ordering condition:
\begin{equation*}
    h_i(z_0, z_0) \leq h_i(z_0, z_1) \leq h_i(z_1, z_0) \leq h_i(z_1, z_1),
\end{equation*}
for all group members $i$ and $-i$, then the treatment selection equation implies a corresponding monotonicity property for treatment take-up, consistent with the condition studied in the literature (e.g., \citealp{vazquez2023causal}):
\begin{equation*}
    D_{ig}(z_0, z_0) \leq D_{ig}(z_0, z_1) \leq D_{ig}(z_1, z_0) \leq D_{ig}(z_1, z_1),
\end{equation*}
for each $i$ and $-i$ within group $g$.

The proposed framework applies to a broad class of empirical settings where spillovers operate through both outcomes and endogenous treatment decisions. Illustrative examples include oligopoly markets, where firms' pricing decisions may influence competitors' market shares, and education contexts, where an individual's labor market outcomes depend on the best friend's schooling decision. More broadly, the framework can be extended to settings such as households, where behaviors involving risky activities generate spillover effects on the health outcomes of other members.

\begin{example}(Duopoly market: pricing decisions)
    To illustrate, consider an oligopoly market with two competing firms, Costco and Sam's Club. Each firm decides whether to raise the price of its membership card and is interested in how this decision affects its market share. A firm's market share depends not only on its own pricing decision but also on its competitor's pricing strategy, giving rise to spillover effects from one firm's decision to the other's outcome.

    Assume that pricing decisions are made simultaneously and that neither firm observes the other's choice at the decision stage. Costco's decision, denoted by $D_{ig}$, depends on a private demand shock $V_{ig}$, such as an idiosyncratic change in reputation or advertising effectiveness, that is unobserved by Sam's Club. This unobserved factor affects both Costco's incentive to raise its membership price and its resulting market share $Y_{ig}$, thereby introducing endogeneity. Although each firm does not directly observe its competitor's pricing decision, both form expectations about rival behavior based on publicly observed market signals $(Z_{ig}, Z_{-ig})$, such as industry-wide cost shocks like tariffs, which are plausibly exogenous and can serve as valid instruments. Moreover, a tariff shock affecting Sam's Club may also influence Costco's pricing decision, generating instrumental spillovers from one firm's assignment to the other's endogenous treatment.
\end{example}

\begin{example}(Friendship Network: education decision)
    Consider a friendship network consisting of two best friends who decide whether to pursue higher education. Each individual's education choice may affect not only her own future earnings but also her friend's, generating spillover effects through information sharing, social learning, or mutual support mechanisms.

    Assume that education decisions are made simultaneously and that neither friend observes the other's choice at the decision stage. Each individual's decision depends on an idiosyncratic unobserved factor, such as intrinsic academic motivation or costs, that influences both the probability of attending college and future earnings, thereby creating endogeneity. While friends do not observe each other's choices, their education decisions may be jointly influenced by shared public characteristics, such as the average family background of classmates in their school cohort. This shared characteristic can serve as a plausible instrumental variable, as it is typically exogenous to individual-specific unobserved ability but may affect education choices through peer effects (see \cite{bifulco2011effect, bifulco2014high, cools2019girls}). The individual is also exposed to public characteristics associated with her best friend, generating instrumental spillovers from the friend's assignment to the individual's endogenous education decision.
\end{example}

\subsection{Local Average Effects and Identification Strategy} \label{sec:local_average_effects}

%The case where spillovers affect both treatment take-up and outcomes has received growing attention in the recent literature, including \citet{vazquez2023causal}, \citet{ditraglia2023identifying}, and related studies. However, without imposing additional restrictions, such as the one-sided noncompliance condition, these frameworks generally fail to achieve point identification of their causal parameters of interest. As acknowledged by \citet{vazquez2023causal}, \textit{“\dots in general, the simultaneous presence of spillovers on treatment status and outcomes can impede identification of causally interpretable parameters even when the instruments are randomly assigned.”} The one-sided noncompliance assumption requires that units can only take the treatment when assigned to it, a condition that is testable in the data but sometimes unrealistic in practice.

A central insight from the treatment effect literature is that, when treatment assignment is endogenous, causal effects can often be identified for specific subpopulations defined by the instrument, for example, the local average treatment effect (LATE) in \cite{imbens1994identification} and related studies. However, in the presence of spillovers, individuals' outcomes depend not only on their own treatment but also on the treatments received by others in their group. The spillovers complicate the interpretation of conventional LATE parameters, as variation in peers' treatments introduces additional causal channels. To disentangle these channels, this section extends the LATE framework to define local average effects that separately capture the causal effect of peers' treatments on an individual's outcome and the direct effect of the individual's own treatment. These parameters retain the causal interpretability of LATE while accommodating the presence of endogenous treatment and spillovers within groups. The identification of these local average effects further motivates the development of a framework based on marginal treatment effects, which explicitly accounts for spillovers operating through both treatment selection and outcomes, as formalized in Section \ref{sec:id_mcse_msde}. Since the identification analysis is conducted at the level of a super-population of groups, I suppress the group subscript $g$ throughout this section to simplify notation.

Building on the framework introduced in Section \ref{sec:basic_setting}, the expected potential outcome $Y_i(d, d')$ generally depends on both the individual's and her peers' unobserved characteristics, $(V_i, V_{-i})$. Following the terminology in the literature, I refer to the conditional expectations $\mathbb{E}[Y_i(d, d') \mid (V_i, V_{-i}) \in P]$, where $P$ denotes a subset of the support of $(V_i, V_{-i})$, as local average potential outcomes. These parameters capture the average potential outcomes for subpopulations defined by specific values of the group-level unobservables, which reflect the underlying unobserved heterogeneity in the population. Taking appropriate differences between local average potential outcomes under different treatment combinations $(d, d')$ yields the average spillover effects from peers' treatments and the direct effects from a unit's own treatment. Definition \ref{def:general_local_effects} provides formal definitions of these parameters.

\begin{definition}\label{def:general_local_effects}(Generalized local average controlled effects)
    Consider the model in Equation (\ref{eq:basic_model_disc_trt}).
    \begin{enumerate}
        \item Fix the treatment of unit $i$ at $D_i = d$, $d \in \{0,1\}$. The generalized local average controlled spillover effect (LACSE), conditional on the group-level unobserved heterogeneity satisfying $(V_i, V_{-i}) \in P$ for some subset $P \subset (0,1)^2$, is defined as as
        \begin{equation*}
            \operatorname{LACSE}_{i}^{(d)}(P) \equiv \mathbb{E}[Y_{i}(d, 1) - Y_{i}(d, 0) \mid (V_i, V_{-i}) \in P].
        \end{equation*}
        \item For unit $i$, fix the peer's treatment at $D_{-i} = d$, where $d \in {0,1}$. The generalized local average controlled direct effect (LACDE), conditional on the group-level unobserved heterogeneity satisfying $(V_i, V_{-i}) \in P$ for some subset $P \subset (0,1)^2$, is defined as as
        \begin{equation*}
            \operatorname{LACDE}_{i}^{(d)}(P) \equiv \mathbb{E}[Y_{i}(1, d) - Y_{i}(0, d) \mid (V_i, V_{-i}) \in P].
        \end{equation*}
    \end{enumerate}
\end{definition}

The generalized LACSEs capture counterfactual spillover effects by exogenously fixing a unit's own treatment, whereas the generalized LACDEs capture counterfactual direct effects by exogenously fixing peers' treatments. Both effects are defined conditional on a subpopulation characterized by $(V_i, V_{-i}) \in P$, in the same spirit as the LATE widely studied in the literature. These parameters possess clear causal interpretations, as they disentangle the distinct influence channels of a unit's own treatment and peers' treatments, while allowing for unobserved treatment effect heterogeneity through conditioning on group-level unobservables. I now establish the identification of the generalized LACSEs and LACDEs using instrumental variables under Assumptions \ref{as:RA}-\ref{as:Vdist}. It is worth noting that this identification result accommodates both discrete instruments with multiple support points and continuously distributed instruments.

I define the propensity score function for unit $i$ as the probability of treatment conditional on the group-level instrument vector, $P_{i}(Z_i, Z_{-i}) \equiv \mathbb{P}(D_i = 1 \mid Z_i, Z_{-i}) $. I denote this function simply by $P_i$ and define the support of the propensity scores for all group members as $\mathcal{P} \equiv \text{Supp}(P_i, P_{-i})$. The propensity score function $P_i$ identifies the threshold function $h_i$ in the treatment selection equation for each group member $i$, as demonstrated in the following derivation:
\begin{equation} \label{eq:ps_idea}
    \begin{aligned}
        &\mathbb{P}\left(D_{i} = 1 \mid Z_{i} = z_0, Z_{-i} = z_1\right) \\
        =& \mathbb{P}\left(V_{i} \leq h_i(z_0, z_1) \mid Z_{i} = z_0, Z_{-i} = z_1\right) \\
        =& \mathbb{P}\left(V_{i} \leq h_i(z_0, z_1)\right)= h_i(z_0, z_1). 
    \end{aligned}
\end{equation}
Additionally, the propensity scores $(P_i, P_{-i})$ are independent of all group members' unobserved heterogeneity, since they are functions of the instruments $(Z_i, Z_{-i})$.

Building on the identification of the propensity score, Theorem \ref{thm:id_lace} establishes the identification of the generalized LACSEs and LACDEs.

\begin{theorem}(Identifying generalized LACSEs and LACDEs) \label{thm:id_lace}
    Suppose that Assumptions \ref{as:RA}-\ref{as:Vdist} hold and let $d \in \{0,1\}$. Under the following conditions, the generalized LACSEs and LACDEs defined in Definition \ref{def:general_local_effects} can be identified.
    \begin{enumerate}
        \item Suppose there exist $(p_0, p_1), (p_0, p_1') \in \mathcal{P}$, $p_1' > p_1$. \textit{The local average controlled spillover effect}, $\text{LACSE}_i^{(1)}(P)$ for $P =\{V_i \leq p_0, p_1 < V_{-i} \leq p_1'\}$, can be identified as 
        \begin{equation*}
            \text{LACSE}_i^{(1)}(V_i \leq p_0, p_1 < V_{-i} \leq p_1') = \frac{\mu_{i,i}^{(1)}(p_0, p_1') - \mu_{i,i}^{(1)}(p_0, p_1)}{C(p_0, p_1') - C(p_0, p_1)}.
        \end{equation*}
        \textit{The local average controlled spillover effect}, $\text{LACSE}_i^{(0)}(P)$ for $P =\{V_i > p_0, p_1 < V_{-i} \leq p_1'\}$, can be identified as 
        \begin{equation*}
            \text{LACSE}_i^{(0)}(V_i > p_0, p_1 < V_{-i} \leq p_1') = \frac{\mu_{i,i}^{(0)}(p_0, p_1') - \mu_{i,i}^{(0)}(p_0, p_1)}{\big(p_1' - p_1\big) - \big[C(p_0, p_1') - C(p_0, p_1)\big]},
        \end{equation*}
        where $\mu_{i,i}^{(d)}(p_0, p_1) \equiv \mathbb{E}[Y_i \mathbbm{1}\{D_i = d\} \mid P_i = p_0, P_{-i} = p_1]$, and $C(p_0, p_1) \equiv \mathbb{E}[D_i D_{-i} \mid P_i = p_0, P_{-i} = p_1]$.
        \item Suppose there exist $(p_0, p_1), (p_0', p_1) \in \mathcal{P}$, $p_0' > p_0$. \textit{The local average controlled direct effect}, $\text{LACDE}_i^{(1)}(P)$ for $P =\{p_0 < V_i \leq p_0', V_{-i} \leq p_1\}$, can be identified as 
        \begin{equation*}
            \text{LACDE}_i^{(1)}(p_0 < V_i \leq p_0', V_{-i} \leq p_1) = \frac{\mu_{i,-i}^{(1)}(p_0', p_1) - \mu_{i,-i}^{(1)}(p_0, p_1)}{C(p_0', p_1) - C(p_0, p_1)}.
        \end{equation*}
        \textit{The local average controlled direct effect}, $\text{LACDE}_i^{(0)}(P)$ for $P =\{p_0 < V_i \leq p_0', V_{-i} > p_1\}$, can be identified as 
        \begin{equation*}
            \text{LACDE}_i^{(0)}(p_0 < V_i \leq p_0', V_{-i} > p_1) = \frac{\mu_{i,-i}^{(0)}(p_0', p_1) - \mu_{i,-i}^{(0)}(p_0, p_1)}{\big(p_0' - p_0\big) - \big[C(p_0', p_1) - C(p_0, p_1)\big]},
        \end{equation*}
        where $\mu_{i,-i}^{(d)}(p_0, p_1) \equiv \mathbb{E}[Y_i \mathbbm{1}\{D_{-i} = d\} \mid P_i = p_0, P_{-i} = p_1]$.
        \item Suppose there exist $(p_0, p_1)$, $(p_0, p_1')$, $(p_0', p_1)$, $(p_0', p_1') \in \mathcal{P}$. \textit{The local average controlled spillover effect}, $\text{LACSE}_i^{(d)}(P)$ for $P =\{p_0 < V_i \leq p_0', p_1 < V_{-i} \leq p_1'\}$, can be identified as
        \begin{equation*}
            \begin{aligned}
                &\text{LACSE}_i^{(d)}(p_0 < V_i \leq p_0', p_1 < V_{-i} \leq p_1') \\
                =& \text{sgn}(2d - 1) \cdot \frac{\big[\mu_{i, i}^{(d)}(p_0', p_1') - \mu_{i, i}^{(d)}(p_0', p_1)\big] - \big[\mu_{i, i}^{(d)}(p_0, p_1') - \mu_{i, i}^{(d)}(p_0, p_1)\big]}{\big[C(p_0', p_1') - C(p_0', p_1)\big] - \big[C(p_0, p_1') - C(p_0, p_1)\big]}.
            \end{aligned}
        \end{equation*}
        \textit{The local average controlled direct effect}, $\text{LACDE}_i^{(d)}(P)$ for $P =\{p_0 < V_i \leq p_0', p_1 < V_{-i} \leq p_1'\}$, can be identified as
        \begin{equation*}
            \begin{aligned}
                &\text{LACDE}_i^{(d)}(p_0 < V_i \leq p_0', p_1 < V_{-i} \leq p_1') \\
                =& \text{sgn}(2d - 1) \cdot \frac{\big[\mu_{i, -i}^{(d)}(p_0', p_1') - \mu_{i, -i}^{(d)}(p_0', p_1)\big] - \big[\mu_{i, -i}^{(d)}(p_0, p_1') - \mu_{i, -i}^{(d)}(p_0, p_1)\big]}{\big[C(p_0', p_1') - C(p_0', p_1)\big] - \big[C(p_0, p_1') - C(p_0, p_1)\big]}.
            \end{aligned}
        \end{equation*}
    \end{enumerate}
\end{theorem}

\begin{proof}
    See Appendix \ref{app:id_lacse_lacde}.
\end{proof}

The identification argument proceeds as follows, with the formal proof provided in Appendix \ref{app:id_lacse_lacde}. Given a pair of observed propensity scores $(P_i, P_{-i}) = (p_0, p_1)$, and noting that the propensity score $P_i$ identifies the threshold function $h_i$, the joint treatment realizations $(D_i, D_{-i})$ partition the space of unobserved heterogeneity $(V_i, V_{-i})$ into four mutually exclusive subpopulations, separated by the thresholds $(p_0, p_1)$. The relationships are summarized as
\begin{equation} \label{eq:subpopulation}
    \begin{array}{lll}
&(D_i, D_{-i}) = (1,1) \Longleftrightarrow V_i \leq p_0, V_{-i} \leq p_1, & (D_i, D_{-i}) = (0,1) \Longleftrightarrow V_i > p_0, V_{-i} \leq p_1, \\
&(D_i, D_{-i}) = (1,0) \Longleftrightarrow V_i \leq p_0, V_{-i} > p_1, & (D_i, D_{-i}) = (0,0) \Longleftrightarrow V_i > p_0, V_{-i} > p_1.
\end{array}
\end{equation}
For instance, the probability of observing $\{D_i = 1, D_{-i} = 1\}$ conditional on $(P_i, P_{-i}) = (p_0, p_1)$ identifies the share of the subpopulation with $\{V_i \leq p_0, V_{-i} \leq p_1\}$, that is,
\begin{equation*}
    \mathbb{P}(D_i = 1, D_{-i} = 1 \mid P_i = p_0, P_{-i} = p_1) = \mathbb{P}(V_i \leq p_0, V_{-i} \leq p_1).
\end{equation*}
The top left panel of Figure \ref{fig:id_local_average_effect} illustrates how these four subpopulations correspond to distinct realizations of $(D_i, D_{-i})$ given observed propensity scores $(p_0, p_1)$.

\begin{figure}[!htt]
        \centering
        \caption{Identifying Local Average Controlled Effects via Propensity Score Variation}
        \includegraphics[width = \textwidth]{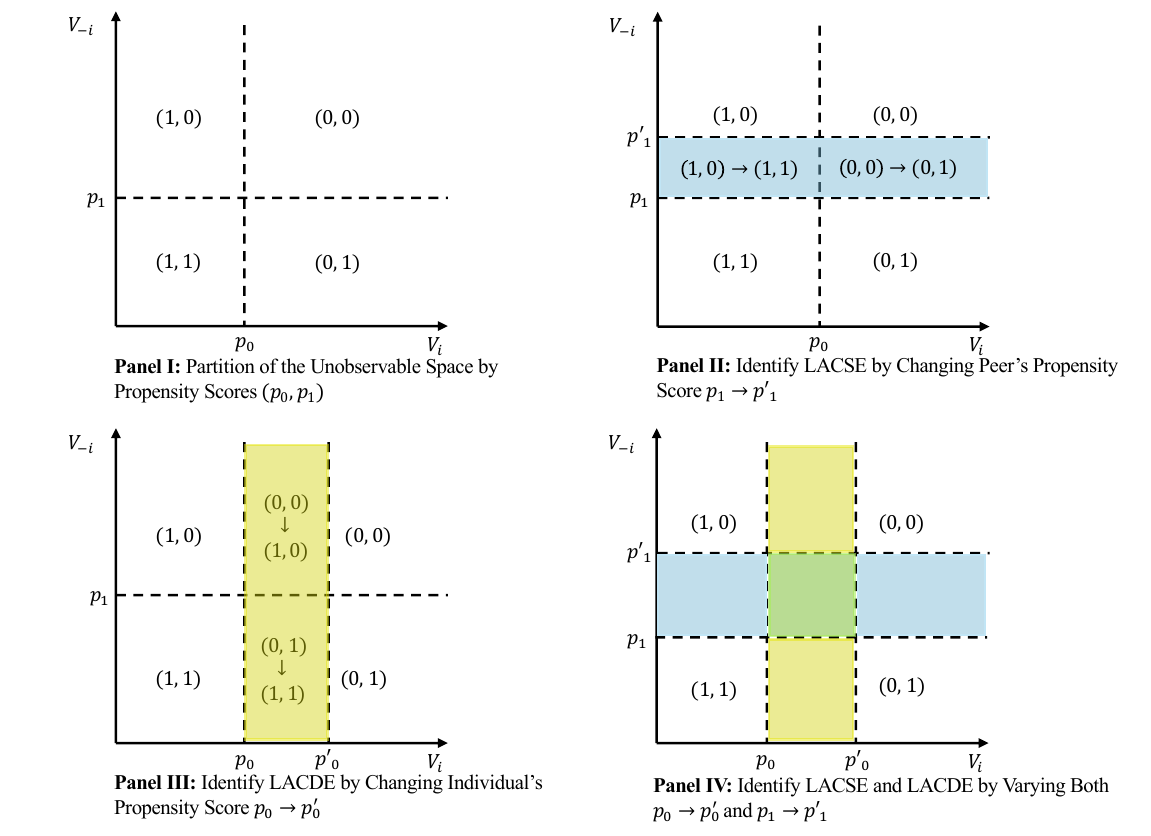}
        \label{fig:id_local_average_effect}
\end{figure}

Given the data, the conditional expectation $\mathbb{E}[Y_i \mathbbm{1}\{D_i = d, D_{-i} = d'\} \mid P_i = p_0, P_{-i} = p_1]$ can be directly recovered from observables. Under the model framework and Assumptions \ref{as:RA}-\ref{as:Vdist}, these observed moments identify the average potential outcomes $Y_i(d, d')$ for the subpopulations associated with the treatment realization $\{D_i = d, D_{-i} = d'\}$. These quantities form the foundation for the identification strategy of Theorem \ref{thm:id_lace}. For instance, when $(D_i, D_{-i}) = (1,1)$,
\begin{equation*}
    \mathbb{E}[Y_i D_i D_{-i} \mid P_i = p_0, P_{-i} = p_1] = \mathbb{E}[Y_i(1, 1) \mathbbm{1}\{V_i \leq p_0, V_{-i} \leq p_1\}],
\end{equation*}
which identifies the average potential outcome $Y_i(1,1)$ for the subpopulation with unobserved characteristics satisfying $\{V_i \leq p_0, V_{-i} \leq p_1\}$.

The identification of the generalized LACSEs and LACDEs exploits exogenous variation in the propensity score values. Suppose there exists another pair of observed propensity scores $(p_0, p_1')$, $p_1' > p_1$. By shifting the propensity scores from $(p_0, p_1)$ to $(p_0, p_1')$ and applying the relationships established in Equation~\eqref{eq:subpopulation}, the subpopulation with unobserved characteristics in the region $\{V_i \leq p_0, p_1 < V_{-i} \leq p_1'\}$ changes its treatment status from $(D_i, D_{-i}) = (1, 0)$ to $(1, 1)$. Likewise, the subpopulation in $\{V_i > p_0, p_1 < V_{-i} \leq p_1'\}$ changes from $(D_i, D_{-i}) = (0, 0)$ to $(0, 1)$. These two groups correspond to the blue-shaded areas in the top-right panel of Figure \ref{fig:id_local_average_effect}.

In both cases, only the peer $-i$ changes her treatment status $D_{-i}$, providing variation that identifies the average spillover effect. Taking the difference between the conditional expectations, $\mathbb{E}[Y_i \mathbbm{1}\{D_i = d, D_{-i} = 1\} \mid \cdot, \cdot] - \mathbb{E}[Y_i \mathbbm{1}\{D_i = d, D_{-i} = 0\} \mid \cdot, \cdot]$, evaluated at $(p_0, p_1)$ and $(p_0, p_1')$, isolates the variation in outcomes attributable to the subpopulations that experience a change in peer treatment status. This difference identifies the local average controlled spillover effects, $\operatorname{LACSE}_i^{(1)}$ and $\operatorname{LACSE}_i^{(0)}$, for the subpopulations corresponding to the two blue-shaded regions in the top-right panel of Figure \ref{fig:id_local_average_effect}. This variation yields the identification result stated in Item 1 of Theorem \ref{thm:id_lace}.

Analogously, when another pair of propensity scores $(p_0', p_1)$ with $p_0' > p_0$ is observed, shifting from $(p_0, p_1)$ to $(p_0', p_1)$ and using the relationships in Equation~\eqref{eq:subpopulation} induces changes in treatment status for unit $i$ only. Specifically, the subpopulations defined by $\{p_0 < V_i \leq p_0', V_{-i} \leq p_1\}$ and $\{p_0 < V_i \leq p_0', V_{-i} > p_1\}$ change their treatment realizations from $\{D_i = 0, D_{-i} = 1\}$ to $\{D_i = 1, D_{-i} = 1\}$ and from $\{D_i = 0, D_{-i} = 0\}$ to $\{D_i = 1, D_{-i} = 0\}$, respectively, as illustrated by the two yellow-shaded regions in the bottom-left panel of Figure \ref{fig:id_local_average_effect}. Because only unit $i$ changes treatment status, the resulting variation identifies the local average controlled direct effect. This variation corresponds to the identification result presented in Item 2 of Theorem \ref{thm:id_lace}.

Suppose the support of the propensity scores exhibits sufficient variation such that four distinct pairs, $(p_0, p_1)$, $(p_0, p_1')$, $(p_0', p_1)$, and $(p_0', p_1')$, are observed with $p_0' > p_0$ and $p_1' > p_1$. Applying the same logic as before, shifting the peer's propensity score from $p_1$ to $p_1'$ while fixing unit $i$'s score at $p_0'$ and, conversely, shifting unit $i$'s score from $p_0$ to $p_0'$ while fixing the peer's score at $p_1'$, identifies the corresponding LACSEs and LACDEs for subpopulations defined by these regions of $(V_i, V_{-i})$. Next, taking cross-differences of the local average effects across the four propensity-score pairs, $(p_0, p_1)$, $(p_0, p_1')$, $(p_0', p_1)$, and $(p_0', p_1')$, isolates the LACSEs and LACDEs for the subpopulation with $(V_i, V_{-i})$ lying in the rectangle $\{p_0 < V_i \leq p_0', p_1 < V_{-i} \leq p_1'\}$, illustrated by the green-shaded area in the bottom-right panel of Figure \ref{fig:id_local_average_effect}.

For example, the difference between $\operatorname{LACSE}_i^{(1)}$ identified for the regions $\{V_i \leq p_0, p_1 < V_{-i} \leq p_1'\}$ and $\{V_i \leq p_0', p_1 < V_{-i} \leq p_1'\}$ yields $\operatorname{LACSE}_i^{(1)}$ for the subpopulation $\{p_0 < V_i \leq p_0', p_1 < V_{-i} \leq p_1'\}$. Similarly, analogous differences yield $\operatorname{LACSE}_i^{(0)}$ and $\operatorname{LACDE}_i^{(d)}$, $d \in \{0, 1\}$, for the same region. This result is formally presented in Item 3 of Theorem \ref{thm:id_lace}, which requires the support of $(P_i, P_{-i})$ to contain four distinct points forming \textit{the “vertices” of a rectangle} in the $(p_0, p_1)$ space.

\begin{remark}(Identification With Binary Instrument) \label{remark:id_binary_iv}
    The conclusions in Theorem \ref{thm:id_lace} apply to both discrete and continuous instrumental variables, provided that they generate sufficient variation in the propensity scores to satisfy the identification conditions. To illustrate, consider the case where a binary instrumental variable is assigned to both units $i$ and $-i$ within each group. In this case, four distinct pairs of propensity scores $(P_i, P_{-i})$ can be observed, corresponding to the four possible combinations of instrument assignments in each group:
    \begin{equation*}
        \begin{array}{c|cc} 
& Z_{-i}=0 & Z_{-i}=1 \\
\hline Z_i=0 & \left(P_i(0,0), P_{-i}(0,0)\right) & \left(P_i(0,1), P_{-i}(1,0)\right) \\
Z_i=1 & \left(P_i(1,0), P_{-i}(0,1)\right) & \left(P_i(1,1), P_{-i}(1,1)\right)
\end{array}
    \end{equation*}
    
    According to Theorem \ref{thm:id_lace}, identifying the LACSEs requires at least two pairs of propensity scores where unit $i$'s score remains fixed while the peer $-i$'s score varies, and identifying the LACDEs instead requires variation in unit $i$'s propensity score while holding the peer $-i$'s propensity score fixed. When the instruments are binary, the support of $(P_i, P_{-i})$ consists of only four points. In this case, identification relies on specific equalities among these propensity scores: identification of LACSEs requires that any two of $P_i(0,0)$, $P_i(0,1)$, $P_i(1,0)$, or $P_i(1,1)$ be equal, while identification of LACDEs requires that any two of $P_{-i}(0,0)$, $P_{-i}(0,1)$, $P_{-i}(1,0)$, or $P_{-i}(1,1)$ be equal. These conditions include the one-sided noncompliance restriction as a special example, a commonly imposed assumption in the literature to achieve identification with binary instruments \citep{kang2016peer, vazquez2023causal, ditraglia2023identifying}.

    The one-sided noncompliance assumption requires that a unit cannot take the treatment unless it receives the instrument assignment, that is, $\mathbb{P}(D_i = 1 \mid Z_i = 0) = 0$ for all $i$. This restriction is equivalent to
    \begin{equation*}
        \begin{aligned}
            P_i(0, 0) = P_i(0, 1) = 0, \quad P_{-i}(0, 0) = P_{-i}(0, 1) = 0,
        \end{aligned}
    \end{equation*}
which satisfies the conditions of Items 1 and 2 in Theorem \ref{thm:id_lace}. Under this structure, the LACSEs and LACDEs identified by Theorem \ref{thm:id_lace} coincide with the local average spillover and direct effects studied in the existing literature \citep{vazquez2023causal}.
    
    As shown in Appendix \ref{app:id_binary_iv}, when the instrument is binary, additional restrictions such as one-sided noncompliance are therefore necessary to point identify local average effects. Interested readers are referred to Appendix \ref{app:id_binary_iv} for a detailed discussion.
\end{remark}

Theorem~\ref{thm:id_lace} demonstrates that local average effects can be identified when the instrumental variables exhibit sufficient variation to induce the necessary differences in propensity scores. As discussed in Remark~\ref{remark:id_binary_iv}, when instruments take only binary values, additional restrictions are required to achieve point identification of certain local average spillover or direct effects. Together, these results highlight that adequate variation in the instruments is crucial for identifying causally interpretable parameters in settings where spillovers influence both outcomes and treatment selection.

When the instrumental variables exhibit continuous variation, they induce continuous variation in the propensity scores. In this case, one can extend the identification strategy in Theorem~\ref{thm:id_lace} by taking limits as $p_0' \to p_0$ and $p_1' \to p_1$, thereby identifying the average controlled spillover and direct effects conditional on $(V_i, V_{-i})$ evaluated at a specific point $(p_0, p_1)$ within the interior of the propensity score support. The next section formalizes this idea by introducing the marginal controlled spillover and marginal controlled direct effects. These parameters serve as building blocks for identifying not only the local average controlled spillover and direct effects discussed above, but also a broader class of policy-relevant treatment effects of interest to researchers.

\subsection{Marginal Effects and Identification Results} \label{sec:id_mcse_msde}

Definition~\ref{def:MCE} formally defines the causal spillover and direct effects evaluated at specific values of the unobserved characteristics $(V_i, V_{-i})$.

\begin{definition}(Marginal controlled spillover effects (MCSE) and marginal controlled direct effects (MCDE)) \label{def:MCE}
    Consider the model in Equation (\ref{eq:basic_model_disc_trt}).
    \begin{enumerate}
        \item Fix the treatment of unit $i$ at $D_i = d$, $d \in \{0,1\}$. \textit{The marginal controlled spillover effect (MCSE)}, given $V_i = p_0$ and $V_{-i} = p_1$, is defined as
        \begin{equation*}
            \operatorname{MCSE}_{i}^{(d)}(p_0, p_1) \equiv \mathbb{E}[Y_{i}(d, 1) - Y_{i}(d, 0) \mid V_{i} = p_0, V_{-i} = p_1], (p_0, p_1) \in (0, 1)^2.
        \end{equation*}
        \item For unit $i$, fix the peer's treatment at $D_{-i} = d$, $d \in \{0,1\}$. \textit{The marginal controlled direct effect (MCDE)}, given $V_i = p_0$ and $V_{-i} = p_1$, is defined as
        \begin{equation*}
            \operatorname{MCDE}_{i}^{(d)}(p_0, p_1) \equiv \mathbb{E}[Y_{i}(1, d) - Y_{i}(0, d) \mid V_{i} = p_0, V_{-i} = p_1], (p_0, p_1) \in (0, 1)^2 
        \end{equation*}
    \end{enumerate}
\end{definition}

The marginal controlled spillover effect captures the impact of changing the peer's treatment on a unit's potential outcome while holding the unit's own treatment status fixed, conditional on the unobserved characteristics $(V_i, V_{-i})$ within the group. Similarly, the marginal controlled direct effect measures the impact of changing a unit's own treatment on her potential outcome while holding the peer's treatment constant, again conditional on $(V_i, V_{-i})$. Because both effects are defined relative to the group-level unobserved heterogeneity, they capture treatment effect heterogeneity arising from latent factors within the group.

The marginal controlled spillover and direct effects are defined analogously to the standard marginal treatment effect (MTE), conditioning on continuous unobserved heterogeneity within the support of the latent variables. Unlike the conventional MTE framework, which rules out interference across units, the marginal controlled effects explicitly incorporate spillovers arising from peers' treatment decisions. As such, they extend the MTE concept to environments where spillovers exist in both outcomes and treatment selection. Section \ref{sec:compare_mte} formally characterizes the connection and distinction between these effects and the standard MTE. By conditioning on the continuous unobservables, the marginal controlled effects provide the building blocks for a wide class of policy-relevant parameters. In particular, the generalized local average controlled effects introduced in Definition \ref{def:general_local_effects} represent a specific class of policy-relevant parameters that can be obtained by integrating the marginal controlled effects over selected regions of the latent heterogeneity space. These connections will be discussed in detail in Section \ref{sec:prte}. The policy-relevant effects play a central role in evaluating counterfactual policies in settings with endogenous treatment and spillovers.

As discussed in the setting, the unobserved heterogeneities $V_i$ and $V_{-i}$ within a group may be correlated. Identification of the parameters of interest requires recovering the joint density of $(V_i, V_{-i})$. Because the marginal distributions of $V_i$ and $V_{-i}$ are normalized to be uniform on the interval $(0,1)$, their joint distribution is characterized by their copula. Formally, the copula is defined as
\begin{equation*}
    C_{V_i, V_{-i}}(p_0, p_1) \equiv \mathbb{P}\big(V_i \leq p_0, V_{-i} \leq p_1 \big), (p_0, p_1) \in (0,1)^2.
\end{equation*}
Lemma \ref{lemma:id_copula} provides identification of this copula on the support of the propensity scores $(P_i, P_{-i})$ without imposing any functional form assumptions, where $P_i$ denotes unit $i$'s propensity score as defined in Equation \eqref{eq:ps_idea}.

\begin{lemma}(Copula of $(V_{i}, V_{-i})$) \label{lemma:id_copula}
    Under Assumptions \ref{as:RA}-\ref{as:Vdist}, the copula between $V_i$ and $V_{-i}$ is identified as
    \begin{equation*}
        C_{V_i, V_{-i}}(p_0, p_1) = \mathbb{P}\big(D_i=1, D_{-i}=1 \mid P_i=p_0, P_{-i}=p_1\big),
    \end{equation*}
    for $(p_0, p_1) \in \mathcal{P}$, where $\mathcal{P}$ denotes the support of $(P_i, P_{-i})$.
\end{lemma}

\begin{proof}
    See Appendix \ref{appendix:id_copula}.
\end{proof}

Let $c_{V_i, V_{-i}}(\cdot, \cdot)$ denote the copula density of $(V_i, V_{-i})$. Since Lemma \ref{lemma:id_copula} establishes identification of the copula, the copula density can be obtained provided that the conditional probability $\mathbb{P}(D_i = 1, D_{-i} = 1 \mid P_i, P_{-i})$ is twice differentiable. This differentiability condition requires that $P_i$ and $P_{-i}$ exhibit continuous variation, which in turn implies that at least some components of the instrument vector $(Z_i, Z_{-i})$ must be continuously distributed. Assumption \ref{as:cts_iv} introduces this continuity requirement.

\begin{assumption}(Continuous instruments) \label{as:cts_iv}
    At least one component of the instrumental variables $(Z_i, Z_{-i})$ is continuously distributed.
\end{assumption}

It then follows that the copula density of $(V_i, V_{-i})$ is identified, as stated in Corollary \ref{corr:id_copula_density}.

\begin{corollary} (Copula density of $(V_i, V_{-i})$) \label{corr:id_copula_density}
    Suppose that Assumptions \ref{as:RA}-\ref{as:cts_iv} hold. If $\mathbb{E}[D_i D_{-i} \mid P_i=p_0, P_{-i}=p_0]$ is twice differentiable at $(p_0, p_1) \in \mathcal{P}$, the copula density $c_{V_i, V_{-i}}(p_0, p_1)$ is identified as 
    \begin{equation*}
        c_{V_i, V_{-i}}(p_0, p_1) = \frac{\partial^2 \mathbb{E}\big[D_i D_{-i} \mid P_i=p_0, P_{-i}=p_1\big]}{\partial p_0 \partial p_1}.
    \end{equation*}
\end{corollary}

\begin{proof}
    See Appendix \ref{appendix:id_copula}.
\end{proof}

Following the literature, the conditional expectation of the potential outcome, given the values of the group-level unobserved characteristics $(V_i, V_{-i})$,
\begin{equation*}
    m_{i}^{(d_0, d_1)}(p_0, p_1) \equiv \mathbb{E}\big[Y_{i}(d_0, d_1) \mid V_{i}=p_0, V_{-i}=p_1\big], (p_0, p_1) \in (0, 1)^2,
\end{equation*}
is defined as the marginal treatment response (MTR) function. The marginal controlled spillover effects (MCSEs) and marginal controlled direct effects (MCDEs) introduced in Definition \ref{def:MCE} are obtained as differences of the corresponding MTR functions. Hence, identification of the MCSEs and MCDEs requires identifying the underlying MTR functions. Theorem \ref{thm:id_mtr_disc_trt} provides the identification of the parameters of interest, the MCSEs and MCDEs, while the detailed process for identifying MTR functions is presented in Appendix \ref{appendix:id_mtr}.

\begin{theorem}(Identifying MCSEs and MCDEs) \label{thm:id_mtr_disc_trt}
    Suppose that Assumptions \ref{as:RA}-\ref{as:cts_iv} hold. For $d_0, d_1 \in \{0,1\}$ and $(p_0, p_1)$ being an interior point of $\mathcal{P}$, the following additional regularity conditions are imposed: (i) $\mathbb{E}[D_i D_{-i} \mid P_i = p_0, P_{-i} = p_1]$ and $\mathbb{E}[Y_i \mathbbm{1}\{D_i = d_0\} \mathbbm{1}\{D_{-i} = d_1\} \mid P_i = p_0, P_{-i} = p_1]$ are twice differentiable; (ii) the marginal treatment response functions $m_i^{(d_0, d_1)}(p_0, p_1)$ are continuous; and (iii) the copula density $c_{V_i, V_{-i}}(p_0, p_1)$ is bounded from above and away from zero. 
    
    Then, the marginal controlled spillover effects (MCSEs), $\operatorname{MCSE}_i^{(d)}(p_0, p_1)$, are identified as
    \begin{equation*}
        \begin{aligned}
            \text{sgn}(2d - 1) \cdot \frac{\partial^2 \mathbb{E}\big[Y_i \mathbbm{1}\{D_i=d\} \mid P_i=p_0, P_{-i}=p_1\big]}{\partial p_0 \partial p_1} \Big/ \frac{\partial^2 \mathbb{E}\left[D_i D_{-i} \mid P_i=p_0, P_{-i}=p_1\right]}{\partial p_0 \partial p_1},
        \end{aligned}
    \end{equation*}
    and the marginal controlled direct effects (MCDEs), $\operatorname{MCDE}_i^{(d)}(p_0, p_1)$, are identified as 
    \begin{equation*}
        \begin{aligned}
            \text{sgn}(2d - 1) \cdot\frac{\partial^2 \mathbb{E}\big[Y_i \mathbbm{1}\{D_{-i}=d\} \mid P_i=p_0, P_{-i}=p_1\big]}{\partial p_0 \partial p_1} \Big/ \frac{\partial^2 \mathbb{E}\left[D_i D_{-i} \mid P_i=p_0, P_{-i}=p_1\right]}{\partial p_0 \partial p_1},
        \end{aligned}
    \end{equation*}
    for $d \in \{0,1\}$ and $(p_0, p_1)$ being an interior point of $\mathcal{P}$, and the function $\text{sgn(x)}$ denotes the sign of $x$.
\end{theorem}

\begin{proof}
    See Appendix \ref{appendix:id_mtr}.
\end{proof}

The identification of the MCSEs and MCDEs builds on the same logic underlying the identification of the LACSEs and LACDEs in part 3 of Theorem \ref{thm:id_lace}, by taking the limits $p_1' \to p_1$ and $p_0' \to p_0$. This argument is illustrated by the green-shaded region in the bottom-right panel of Figure~\ref{fig:id_local_average_effect}, which represents the limiting case where $p_1' \to p_1$ and $p_0' \to p_0$. The validity of this limiting argument requires the propensity scores to exhibit continuous variation in a neighborhood of $(p_0, p_1)$. Moreover, the identification framework naturally extends to settings with exogenous covariates, with the corresponding results provided in Appendix \ref{appendix:id_covariates}.

The assumption of continuously distributed instruments in Assumption~\ref{as:cts_iv} is sufficient but not necessary for identifying the MCSEs and MCDEs. When instruments have limited or discrete variation, the parametric specifications in Assumption~\ref{as:para_assump} can be used to extrapolate the identification of marginal controlled effects beyond the observed support of the propensity scores. Alternative extrapolation approaches, such as those proposed by \citet{mogstad2018using}, may also be applied. However, point identification may no longer hold, and the parameters would instead be partially identified. A formal treatment of this extension is left for future research.

\begin{remark} (Groups with multiple individuals) 
    The identification strategy naturally extends to settings with more than two individuals per group. Consider a group of size $n < \infty$, indexed by $i \in \{1, \ldots, n\}$. In this case, the threshold function $h_i$ in Equation \eqref{eq:basic_model_disc_trt} depends on the full vector of instrument assignments $(Z_1, \ldots, Z_n)$. The propensity score $\mathbb{P}(D_i = 1 \mid Z_1, \ldots, Z_n)$ identifies the threshold $h_i$. The joint distribution of the unobserved heterogeneities within the group is then recovered from the conditional probability $\mathbb{P}(D_1 = 1, \ldots, D_n = 1 \mid P_1, \ldots, P_n)$, where $P_i$ denotes the propensity score of individual $i$. Once this joint distribution is identified, the marginal treatment response functions can be obtained by differentiating
    \begin{equation*}
        \mathbb{E}\left[Y_{i} \mathbbm{1}\left\{D_{1} = d_1\right\} \cdots \mathbbm{1}\left\{D_{n} = d_n\right\}\mid P_{1}, \cdots, P_{n}\right]
    \end{equation*}
with respect to $(P_1, \ldots, P_n)$, under suitable smoothness conditions. Finally, differences between the resulting MTR functions yield the MCSEs and MCDEs.
\end{remark}

\begin{remark} (Testing the spillover structure)
    The identification of marginal treatment response functions makes it possible to test additional structural assumptions about the nature of spillovers. For example, in addition to Assumptions \ref{as:RA}-\ref{as:Vdist}, suppose that each unit's outcome depends not on the entire treatment vector $\boldsymbol{D} \equiv (D_1, \cdots, D_n)$, but instead on a lower-dimensional function $H(\cdot)$ of this vector. In the literature, $H(\cdot)$ is commonly referred to as the exposure mapping. A standard specification is the average treatment level within the group, $H(\boldsymbol{D}) = \sum_{i=1}^n D_i / n$. Under this structure, given unit $i$'s own treatment $D_i = d$, any two treatment vectors $\boldsymbol{d} = (d_1, \cdots, d_n)$ and $\boldsymbol{\tilde{d}} = (\tilde{d}_1, \cdots, \tilde{d}_n)$ that generate the same exposure level, $H(\boldsymbol{d}) = H(\boldsymbol{\tilde{d}})$, should yield identical marginal treatment response functions: 
\begin{equation*}
    \mathbb{E}\big[Y_i(d, \boldsymbol{d}) \mid V_1 = p_1, \cdots, V_n = p_n\big] = \mathbb{E}\big[Y_i(d, \boldsymbol{\tilde{d}}) \mid V_1 = p_1, \cdots, V_n = p_n\big],
\end{equation*}
for all $(p_1, \cdots, p_n)$ in the support of the propensity score functions. This equality provides a testable implication of the assumed spillover structure $H(\cdot)$, thereby linking identification of MTR functions to specification testing of exposure mappings.
\end{remark}

\subsection{Policy Relevant Treatment Effects} \label{sec:prte}
The MCSEs and MCDEs not only capture heterogeneous spillover and direct effects but also serve as fundamental building blocks for deriving a wide range of causal parameters commonly examined in the literature. This section illustrates several examples demonstrating how the MCSEs and MCDEs can be used to recover other treatment effect parameters of policy relevance.

\subsubsection*{Average Controlled Spillover and Direct Effects}

Researchers are often interested in summarizing heterogeneous spillover and direct effects across individuals by aggregating them into population-level parameters (see, for example, \cite{vazquez2023identification} and related studies). Within this framework, the average controlled spillover effect (ACSE) can be defined as $\operatorname{ACSE}_i(d) \equiv \mathbb{E}[Y_i(d, 1) - Y_i(d, 0)]$, which measures the expected change in unit $i$'s outcome when the peer's treatment status changes exogenously from 0 to 1, holding the unit's own treatment fixed at $D_i = d$. Similarly, the average controlled direct effect (ACDE) can be defined as $\operatorname{ACDE}_i(d) \equiv \mathbb{E}[Y_i(1, d) - Y_i(0, d)]$, which reflects the expected change in unit $i$'s outcome when her own treatment status changes exogenously from zero to one, holding her peers' treatment status fixed at $D_{-i} = d$.

Under the assumption that the propensity scores have full support, i.e., $\mathcal{P} = (0,1)^2$, the ACSEs and ACDEs are point identified using the MCSEs, the MCDEs, and the copula density of $(V_i, V_{-i})$ identified in Section~\ref{sec:id_mcse_msde}, by integrating the MCSEs or MCDEs weighted by the corresponding copula density:
\begin{equation*}
    \begin{aligned}
        &\operatorname{ACSE}_i(d) = \int_{0}^{1} \int_{0}^{1} \operatorname{MCSE}_{i}(d; p_0, p_1) c_{V_i, V_{-i}}(p_0, p_1) dp_0 dp_1, \\
        &\operatorname{ACDE}_i(d) = \int_{0}^{1} \int_{0}^{1} \operatorname{MCDE}_{i}(d; p_0, p_1) c_{V_i, V_{-i}}(p_0, p_1) dp_0 dp_1.
    \end{aligned}
\end{equation*}

When the propensity scores lack full support, the average controlled spillover and direct effects, as well as other policy-relevant treatment effects, remain only partially identifiable. Under the additional assumption that potential outcomes are almost surely bounded, $|Y_i(d, d')| \leq B < \infty$, the MCSEs and MCDEs are confined within the range $[-2B, 2B]$ at points outside the observed support of the propensity scores. To extend identification beyond this region, one may impose the parametric structure in Assumption~\ref{as:para_assump} or adopt an extrapolation approach similar to that proposed by \citet{mogstad2018using}. A formal development of these extensions is left for future research.
%{\color{red} Add the partially identified results and the extrapolation methods.}

\subsubsection*{Local Average Controlled Spillover and Direct Effects}

Once the MCSEs and MCDEs are point identified, they can be used to recover the LACSEs and LACDEs defined in Section~\ref{sec:local_average_effects}. The following discussion illustrates how these marginal effects can be employed to obtain the local average spillover and direct effects examined in the existing literature.

Suppose there exist two values of the instrumental variable, $z_0, z_1$, such that the associated propensity scores $P_i(z, z')$, for $z, z' \in {z_0, z_1}$, can be consistently ordered across all individuals and groups. Without loss of generality, assume that $P_i(z_0, z_0) \leq P_i(z_0, z_1) \leq P_i(z_1, z_0) \leq P_i(z_1, z_1)$. Under this ordering, the treatment selection mechanism in Equation~\eqref{eq:basic_model_disc_trt} implies the following monotonicity condition:
\begin{equation*}
    D_i(z_0, z_0) \leq D_i(z_0, z_1) \leq D_i(z_1, z_0) \leq D_i(z_1, z_1)
\end{equation*}
almost surely, where $D_i(z, z')$ denotes the potential treatment received by unit $i$ when the instrument assignments are fixed exogenously at $(Z_i, Z_{-i}) = (z, z')$.

\cite{vazquez2023causal} partitions the population into a finite number of compliance types based on the values of the potential treatment vector ${D_i(z, z')}_{z, z' \in \{z_0, z_1\}}$. This framework identifies the local average spillover effect $\mathbb{E}[Y_i(0,1) - Y_i(0,0) \mid T_{-i} = c]$, and the local average direct effect $\mathbb{E}[Y_i(1,0) - Y_i(0,0) \mid T_i = c]$, where $T_i = c$ denotes the complier subgroup, defined as the set of units whose unobserved heterogeneity $V_i$ lies between the two thresholds $P_i(z_0, z_1)$ and $P_i(z_1, z_0)$. Intuitively, these are individuals who would not take the treatment under $(z_0, z_1)$ but would take it under $(z_1, z_0)$. This paper also considers the setting in which a one-sided noncompliance condition holds, meaning that individuals cannot receive the treatment when assigned the instrument value $z_0$. Under one-sided noncompliance, the propensity scores satisfy $0 = P_i(z_0, z_0) = P_i(z_0, z_1) \leq P_i(z_1, z_0) \leq P_i(z_1, z_1)$, which corresponds to a special case of the sufficient identification condition stated in Part 1 of Theorem~\ref{thm:id_lace}.

Figure~\ref{fig:app_id_binary_iv} in Appendix~\ref{app:id_binary_iv} illustrates the regions in the $(V_i, V_{-i})$ space corresponding to the subpopulations ${T_i = c}$ and ${T_{-i} = c}$. Integrating the identified MCSEs and MCDEs over these regions, using the copula density of $(V_i, V_{-i})$ as weights, recovers the local average spillover and direct effects analyzed by \cite{vazquez2023causal}:
\begin{equation*}
    \begin{aligned}
        &\mathbb{E}\left[Y_{i}(0, 1) - Y_i(0, 0) \mid T_{-i} = c\right] \\
        =& \frac{1}{\mathbb{P}(T_{-i} = c)} \int_{0}^{P_{-i}(z_1, z_0)} \int_{0}^{1} \operatorname{MCSE}_i\left(0; v_0, v_1\right) c_{V_{i}, V_{-i}}(v_0, v_1) dv_0 dv_1, \\
        &\mathbb{E}\left[Y_{i}(1, 0) - Y_i(0, 0) \mid T_{i} = c\right] \\
        =& \frac{1}{\mathbb{P}(T_{i} = c)} \int_{0}^{1} \int_{0}^{P_i(z_1, z_0)} \operatorname{MCDE}_i\left(0; v_0, v_1\right) c_{V_{i}, V_{-i}}(v_0, v_1) dv_0 dv_1, \\
        &\mathbb{P}(T_{-i} = c) = \int_{0}^{P_{-i}(z_1, z_0)} \int_{0}^{1} c_{V_{i}, V_{-i}}(v_0, v_1) dv_0 dv_1, \\
        &\mathbb{P}(T_{i} = c) = \int_{0}^{1} \int_{0}^{P_i(z_1, z_0)} c_{V_{i}, V_{-i}}(v_0, v_1) dv_0 dv_1,
    \end{aligned}
\end{equation*}
where the copula density is identified in Corollary \ref{corr:id_copula_density}, and Theorem \ref{thm:id_mtr_disc_trt} provides identification of $\operatorname{MCSE}_i\left(0; v_0, v_1\right)$ and $\operatorname{MCDE}_i\left(0; v_0, v_1\right)$.

\subsubsection*{Other Policy Relevant Treatment Effect}

In addition, the identified MCSEs and MCDEs can be used to recover the policy-relevant treatment effect (PRTE), which quantifies the expected change in outcomes induced by a policy-driven shift in the treatment selection mechanism. The PRTE aggregates the underlying marginal controlled effects across the distribution of unobserved heterogeneity, weighted by how the policy changes the propensity of treatment participation, following the interpretation of \citet{heckman2005structural} and related work. This parameter provides a meaningful measure of the impact of counterfactual policy interventions in the presence of heterogeneous treatment effects, allowing for the evaluation of counterfactual interventions that modify the selection into treatment.

To formalize the link between marginal controlled effects and policy-relevant treatment effects, consider how the MCSEs and MCDEs characterize PRTEs arising from exogenous policy changes. Let $\mathcal{A}$ denote a set of feasible policies. For any policy $a \in \mathcal{A}$, I use superscript $a$ to denote the corresponding potential variables under that policy. For example, $D_i^a$ denotes the treatment status of unit $i$ that would be realized if policy $a$ were implemented. Thus, changing the policy from a to $a'$ induces changes in the distribution of instrumental variables, treatment selection, and outcomes, such as $Z_i^a \to Z_i^{a'}$, $D_i^a \to D_i^{a'}$, and $Y_i^a \to Y_i^{a'}$. These counterfactual changes form the basis for evaluating PRTEs.

Under policy $a$, the treatment decision is given by
\begin{equation*}
    D_i^a = \mathbbm{1}\{V_i^a \leq P_i^a(Z_i^a, Z_{-i}^a)\},
\end{equation*}
where $P_i^a(Z_i^a, Z_{-i}^a) = \mathbb{P}(D_i^a = 1 \mid Z_i^a, Z_{-i}^a)$ denotes the policy-specific propensity score. Following the standard policy invariance assumption (as described in Assumption~\ref{as:policy_invariant}), the introduction of a new policy is assumed to affect only the treatment selection mechanism through changes in the propensity score, without altering the joint distribution of unobserved characteristics.

\begin{assumption}(Policy invariance)\label{as:policy_invariant}
    The distribution of 
    \begin{equation*}
        \big(U^a_i, U^a_{-i}, V^a_i, V^a_{-i}\big)
    \end{equation*}
    is invariant to any policy $a \in \mathcal{A}$.
\end{assumption}

The policy-relevant treatment effect (PRTE) measures the average change in outcomes induced by a policy intervention that modifies the treatment assignment mechanism, which is defined as
\begin{equation*}
    \text{PRTE}(a, a') \equiv \frac{\mathbb{E}[Y^{a'}_i - Y^{a}_i]}{\Delta P},
\end{equation*}
where $\Delta P$ denotes the proportion of groups in which at least one member changes treatment status as a result of the policy shift from $a$ to $a'$. 

Given that each pair of propensity scores $(P_i^a, P_{-i}^a)$ partitions the support of $(V_i, V_{-i})$ into four regions associated with distinct treatment realizations $(D_i, D_{-i})$, the expected outcome under policy $a$ can, under Assumption~\ref{as:policy_invariant}, be expressed as follows:
\begin{equation*}
    \begin{aligned}
        \mathbb{E}\big[Y^a_i\big] =& \int_{0}^{1} \int_{0}^{1} \Big\{ m_i^{(1, 1)}(p_0, p_1) \mathbb{P}\big(P^a_i \geq p_0, P^a_{-i} \geq p_1\big)  \\
        &+ m_i^{(1, 0)}(p_0, p_1) \mathbb{P}\big(P^a_i \geq p_0, P^a_{-i} < p_1\big)  \\
        &+ m_i^{(0, 1)}(p_0, p_1) \mathbb{P}\big(P^a_i < p_0, P^a_{-i} \geq p_1\big)  \\
        & + m_i^{(0,0)}(p_0, p_1) \mathbb{P}\big(P^a_i < p_0, P^a_{-i} < p_1\big) \Big\} c_{V_i, V_{-i}}(p_0, p_1) dp_0 dp_1.
    \end{aligned}
\end{equation*}
Because $\mathbb{E}[Y_i^a]$ is represented as a weighted average of the marginal treatment response functions, and the PRTE is defined as the difference in $\mathbb{E}[Y_i^a]$ across alternative policies, the PRTE naturally admits an interpretation as a weighted average of the MCDEs and MCSEs over particular regions of the latent heterogeneity space $(V_i, V_{-i})$. Hence, the identified MCDEs and MCSEs provide the key building blocks for constructing PRTEs associated with a broad class of counterfactual policy interventions. Appendix~\ref{app:prte} presents explicit expressions for the PRTE under several empirically relevant types of policy changes.

\subsection{Comparing With Marginal Treatment Effect} \label{sec:compare_mte}
This section introduces the connection between the marginal controlled effects and the standard marginal treatment effect (MTE) framework. The MCDEs and MCSEs are defined in a manner analogous to the MTE, capturing how potential outcomes vary with continuous unobserved heterogeneity. However, unlike the standard MTE that rules out spillovers, the MCDEs and MCSEs explicitly account for spillovers arising from peers' treatments as well as endogeneity in both own and peer treatment decisions. This section formally examines the relationship between the marginal controlled effects and the conventional MTE, demonstrating that the MCSEs and MCDEs naturally extend the MTE framework to settings with spillovers in outcomes and treatment selection within groups.

If spillover effects exist but are ignored and the standard MTE identification strategy is applied, the conventional estimand
\begin{equation*}
    \begin{aligned}
        \frac{\partial}{\partial p_0} \mathbb{E}\big[Y_i \mid P_i(Z_i)=p_0\big], \text{ where } P_i(z_0) \equiv \mathbb{P}\big(D_i = 1 \mid Z_i = z_0\big),
    \end{aligned}
\end{equation*}
fails to identify the true MTE, $\mathbb{E}[Y_i(1) - Y_i(0) \mid V_i = p_0]$.

In the presence of the spillover structure specified in Equation~\eqref{eq:basic_model_disc_trt}, the conventional propensity score can be written as
\begin{equation*}
    \begin{aligned}
        &\mathbb{P}\left(D_i = 1 \mid Z_i = z_0\right) \\
        =& \mathbb{P}\left(V_i \leq h_i(z_0, Z_{-i}) \mid Z_i = z_0\right) \\
        =& \int_{z_1 \in \mathcal{Z}} \mathbb{P}\left(V_i \leq h_i(z_0, z_1) \mid Z_i = z_0, Z_{-i} = z_1\right) f_{Z_{-i} \mid Z_i=z_0}\left(z_1\right) d z_1 \\
        =& \int_{z_1 \in \mathcal{Z}} h_i\left(z_0, z_1\right) f_{Z_{-i} \mid Z_i=z_0}\left(z_1\right) d z_1,
    \end{aligned}
\end{equation*}
where $\mathcal{Z}$ denotes the support of the peer's instrumental variable $Z_{-i}$, the second equality follows from the law of iterated expectations, and the last equality uses the independence assumption (Assumption~\ref{as:RA}) and the distributional normalization in Assumption~\ref{as:Vdist}. This expression shows that the conventional propensity score is a weighted average of the unit's threshold function $h_i(z_0, z_1)$ over the peer's instrument $Z_{-i}$, conditional on $Z_i$.

Furthermore, Corollary~\ref{cor:breakdown_mte} demonstrates that, when spillovers are present, the conventional MTE identifier,  $\partial \mathbb{E}[Y_i \mid P_i\left(Z_i\right)=p_0] / \partial p_0$ captures a weighted average of the MCDEs, augmented by a bias term arising from the dependence of the unit's treatment on the peer's instrument and from the correlation between $Z_i$ and $Z_{-i}$. 

\begin{corollary}(Breakdown of MTE causal interpretation) \label{cor:breakdown_mte}
    Consider the model with spillovers specified in Equation~\eqref{eq:basic_model_disc_trt}, and suppose that the conditions in Theorem~\ref{thm:id_mtr_disc_trt} are satisfied. Under these assumptions, the conventional MTE identifier identifies 
    \begin{equation*}
        \begin{aligned}
           \frac{\partial \mathbb{E}[Y_i \mid P_i\left(Z_i\right)=p_0]}{\partial p_0} = &\int_{0}^{1} \int_{0}^{p_1} \operatorname{MCDE}_i(1; p_0, v_1) c_{V_i, V_{-i}}(p_0, v_1) f_{P_{-i} \mid P_i = p_0} (p_1) dv_1 dp_1 \\
            +&\int_{0}^{1} \int_{p_1}^{1} \operatorname{MCDE}_i(1; p_0, v_1) c_{V_i, V_{-i}}(p_0, v_1) f_{P_{-i} \mid P_i = p_0} (p_1) dv_1 dp_1 \\
            &+ \mathcal{R}.
        \end{aligned}
    \end{equation*}
    The bias term $\mathcal{R}$ is generally nonzero. It vanishes under two sufficient conditions: (i) the unit's treatment decision is unaffected by the peer's instrument, and (ii) the instruments $Z_i$ and $Z_{-i}$ are mutually independent within each group. The explicit form of $\mathcal{R}$ is provided in Appendix~\ref{app:mte_breakdown}.
    \begin{proof}
        See Appendix \ref{app:mte_breakdown}.
    \end{proof}
\end{corollary}

If spillover effects are absent from both the outcome and the treatment selection processes, the identification results collapse to the standard MTE framework, as shown in Corollary \ref{cor:consistent_mte_sutva}.

\begin{corollary} \label{cor:consistent_mte_sutva}
    Suppose that no spillover effects are present, and the SUTVA holds. Specifically, assume that $V_i \indep V_{-i}$, $Y_i(D_i, d) = Y_i(D_i, d') \equiv Y_i(D_i)$, and $h_i(Z_i, z) = h_i(Z_i, z') \equiv h_i(Z_i)$. Under the conditions of Theorem~\ref{thm:id_mtr_disc_trt}, the identification results for the MCSEs and MCDEs collapse to the standard MTE framework.
    In particular:
    \begin{enumerate}
        \item The propensity score identifies the standard threshold function $h_i(Z_i)$,
        \begin{equation*}
        \mathbb{P}\big(D_i = 1 \mid Z_i = z_0, Z_{-i} = z_1) = h_i(z_0), 
        \end{equation*}
        indicating that the treatment decision of unit $i$ is unaffected by the peer's instrument.
        \item The copula density simplifies to independence, $c_{V_i, V_{-i}}(p_0, p_1) = 1$, suggesting that $V_i \indep V_{-i}$.
        \item The marginal controlled spillover effect is zero for all $(p_0, p_1)$ in the interior of $\mathcal{P}$, $\operatorname{MCSE}_i^{(d)}(p_0, p_1) = 0$, implying that the outcome does not depend on the peer's treatment.
        \item The marginal controlled direct effect reduces to the MTE, $\operatorname{MCDE}_i^{(d)}(p_0, p_1) = \mathbb{E}[Y_i(1) - Y_i(0) \mid V_i = p_0]$, for all $(p_0, p_1)$ in the interior of $\mathcal{P}$.
    \end{enumerate}
\end{corollary}
\begin{proof}
    See Appendix \ref{app:mce_sutva}.
\end{proof}

To conclude, the standard MTE may lose its causal interpretation when spillovers are present, whereas the MCSEs and MCDEs coincide with the MTE under SUTVA. Hence, the framework developed in this paper provides a natural extension of the standard MTE framework, generalizing it to settings with spillovers in both outcomes and treatment selection.

\begin{remark}(Connection to the multivalued treatment model)
    If the group is treated as a single decision-making unit and the treatment vector $(D_{ig}, D_{-ig}) \in \{0,1\}^2$ can be reformulated as a multivalued group-level treatment $D_g \in \{0, 1, 2, 3\}$, the framework can be viewed as a multivalued MTE models, such as \citet{lee2018identifying}. Translating their setup into the spillover context, their identification relies on an exclusion restriction requiring that unit $i$'s threshold function does not depend on peers' instruments. In contrast, the treatment selection structure developed in this paper enables point identification of the threshold functions without imposing such exclusion restrictions and accommodates settings where spillovers arise from peers' instruments. Hence, the two frameworks are not nested. A more detailed comparison with the multivalued treatment framework is provided in Appendix~\ref{app:multi_trt}.
\end{remark}

\subsection{Testable Implications of Identifying Assumptions} \label{sec:test_implication}

The imposed spillover model structure and assumptions yield two sets of testable implications.

The first set arises from the fact that the cross-partial derivatives of the observed conditional expectations,
\begin{equation*}
    \frac{\partial^2}{\partial p_1 \partial p_0} \mathbb{E}\big[\mathbbm{1}\left\{Y_i \in A_1, Y_{-i} \in A_2\right\} \mathbbm{1}\left\{D_i=d, D_{-i}=d'\right\} \mid P_i=p_0, P_{-i}=p_1\big],
\end{equation*}
identify the joint distribution of potential outcomes weighted by the copula density of unobservables, where $A_1, A_2$ denote arbitrary Borel sets in the outcome support. Since both the conditional probabilities and the copula density are nonnegative, these derivatives must be weakly positive, generating a set of \textit{nesting inequalities} that serve as testable restrictions.

The second set of implications follows from the index sufficiency property: the marginal treatment response functions depend only on the values of the propensity scores, not directly on the realizations of the instrumental variables. Hence, for any two instrument pairs $(z_0, z_1)$ and $(\tilde{z}_0, \tilde{z}_1)$ that yield identical propensity scores $(P_i, P_{-i})$, the conditional expectation
\begin{equation*}
    \frac{\partial^2}{\partial p_1 \partial p_0} \mathbb{E}\big[\mathbbm{1}\left\{Y_i \in A_1, Y_{-i} \in A_2\right\} \mathbbm{1}\left\{D_i=d, D_{-i}=d'\right\} \mid Z_i, Z_{-i}\big]
\end{equation*}
should remain invariant across the two sets of instruments, as they correspond to the same marginal treatment response values.

Corollary \ref{corr:test_implication} formally states the nesting inequality and index sufficiency conditions implied by the model.

\begin{corollary} (Testable implications) \label{corr:test_implication}
    The following conditions constitute the testable implications of Assumptions \ref{as:RA}-\ref{as:cts_iv}, given the model structure specified in Equation \eqref{eq:basic_model_disc_trt}.
    \begin{enumerate}
        \item (Nesting inequalities) For any Borel set $A_1, A_2 \subseteq \mathcal{Y}$, $d \in \{0,1\}$, and $(p_0, p_1) \in \mathcal{P}$, the following inequality should hold:
        \begin{equation*}
            \begin{aligned}
                &\frac{\partial^2}{\partial p_1 \partial p_0} \mathbb{E}\Big[\mathbbm{1} \big\{Y_i \in A_1, Y_{-i} \in A_2 \big\} \mathbbm{1}\big\{D_{i} = d, D_{-i} = d \big\} \mid P_{i} = p_0, P_{-i}= p_1 \Big] \geq 0, \\
                &-\frac{\partial^2}{\partial p_1 \partial p_0} \mathbb{E}\Big[\mathbbm{1} \big\{Y_i \in A_1, Y_{-i} \in A_2 \big\} \mathbbm{1} \big\{D_{i} = d, D_{-i} = 1-d \big\} \mid P_{i} = p_0, P_{-i} = p_1\Big] \geq 0.
            \end{aligned}
    \end{equation*}
    \item (Index sufficiency) For any Borel set $A_1, A_2 \subseteq \mathcal{Y}$, $d_0, d_1 \in \{0,1\}$, and different values of instruments, $(z_0, z_1) \neq (\tilde{z}_0, \tilde{z}_1)$, that yield the same propensity score values such that $P_i(z_0, z_1) = P_i(\tilde{z}_0, \tilde{z}_1) = p_0$ and $P_{-i}(z_0, z_1) = P_{-i}(\tilde{z}_0, \tilde{z}_1) = p_1$, the following equalities should hold: 
    \begin{equation*}
        \begin{aligned}
            & \mathbb{E}\big[\mathbbm{1}\big\{Y_i \in A_1, Y_{-i} \in A_2 \big\} \mathbbm{1}\{D_i = d_0, D_{-i} = d_1\} \mid Z_i = z_0, Z_{-i} = z_1\big] \\
            =& \mathbb{E}\big[\mathbbm{1}\big\{Y_i \in A_1, Y_{-i} \in A_2 \big\} \mathbbm{1}\{D_i = d_0, D_{-i} = d_1\} \mid Z_i = \tilde{z}_0, Z_{-i} = \tilde{z}_1\big].
        \end{aligned}
    \end{equation*}
    \end{enumerate}
\end{corollary}
\begin{proof}
    See Appendix \ref{app:testable_implication} for the proof of the validity of these testable implications.
\end{proof}

Existing literature, such as \cite{carr2021testing}, has developed methods that could be implemented to test the conditions in Proposition \ref{corr:test_implication}. A formal application of these testing procedures to the present framework is left for future research.

\section{Estimation and Inference} \label{sec:est_inference}

\subsection{Semiparametric estimation and inference} \label{subsec:semiparametric}

\subsubsection{Estimation procedure} \label{subsec:estimation}

Consider a random sample of $G$ groups, each consisting of n units. For each group $g = 1, \ldots, G$, the observed data
\begin{equation*}
    \{Y_{0 g}, \ldots, Y_{(n-1) g}, D_{0 g}, \ldots, D_{(n-1) g}, Z_{0 g}, \ldots, Z_{(n-1) g}, X_{0 g}, \ldots, X_{(n-1) g}\}_{g = 1, \ldots, G}
\end{equation*} 
are independently and identically distributed across groups. %according to the probability measure $\mathbb{P}$.  

This section develops a semiparametric estimation procedure that extends the framework of \cite{carneiro2009estimating} to settings with spillover effects in both treatment and outcome equations. The estimation section considers the covariate-augmented setting introduced in Appendix~\ref{appendix:id_covariates}. The model incorporating covariates can be expressed as
\begin{equation*}
    \left\{\begin{array}{l}
        \begin{aligned}
            %Y_{ig} = & \big[Y_{ig} (X_{ig}, X_{(1-i)g}, 1, 1) D_{(1-i)g} + Y_{ig} (X_{ig}, X_{(1-i)g}, 1, 0)(1 - D_{(1-i)g})\big] D_{ig} \\
            %& + \big[Y_{ig} (X_{ig}, X_{(1-i)g}, 0, 1) D_{(1-i)g} + Y_{ig}(X_{ig}, X_{(1-i)g}, 0, 0)(1 - D_{(1-i)g})\big] (1 - D_{ig}),
            %& \quad \\
            Y_{ig} = & m_i(X_{ig}, X_{-ig}, D_{ig}, D_{-ig}, U_{ig}, U_{-ig}), \\
            D_{ig} = & \mathbbm{1}\big\{V_{ig} \leq h_i(X_{ig}, X_{-ig}, D_{ig}, D_{-ig})\big\},
        \end{aligned} \\
        \end{array}\right.
\end{equation*}
where outcomes depend on both own and peer covariates and treatments, and treatment decisions follow a threshold-crossing rule. It is assumed that covariates and instruments are randomly assigned at the group level,
\begin{equation*}
    \big(W_{ig}, W_{-ig}\big) \indep \bigg\{\big(V_{ig}, V_{-ig}, U_{ig}, U_{-ig}\big)\bigg\},
\end{equation*}
where $W_{ig} \equiv (X_{ig}, Z_{ig})$.

Assumption~\ref{as:est_assumption} is maintained throughout the estimation section.

\begin{assumption}(Estimation assumptions) \label{as:est_assumption}
    Assume that (i) $\mathbb{E}|Y_{ig}(d, d')| < \infty, d, d' \in \{0,1\}$. (ii) Propensity scores, $P_{ig}, i \in \{0,1\}$ are nondegenerate continuous random variable. (iii) The conditional expectations $\mathbb{E}[Y_{i d d' g} \mid X_g=\mathbf{x}, P_{ig}=p_0, P_{-ig}=p_1]$, $X_g \equiv (X_{ig}, X_{-ig})$, and $\mathbb{E}[D_{ig} D_{-ig} \mid P_{ig}=p_0, P_{-ig}=p_1]$ are assumed to be twice continuously differentiable with respect to $(p_0, p_1)$. (iv) $\partial^2 \mathbb{E}[D_{ig} D_{-ig} \mid P_{ig}=p_0, P_{-ig}=p_1] / \partial p_0 \partial p_1$ is bounded from above and away from zero.
\end{assumption}

To illustrate the estimation procedure, this section focuses on a simple case where each group consists of two units, i.e., $n = 2$ and $i \in \{0, 1\}$. The extension to group sizes $n > 2$ follows analogously. Building on the identification results with exogenous covariates presented in Appendix~\ref{appendix:id_covariates}, this section aims to estimate the marginal treatment response (MTR) functions, 
\begin{equation} \label{eq:estimand}
    \begin{aligned}
        & m_{ig}^{\left(\mathbf{x}, d, d^{\prime}\right)}\left(p_0, p_1\right)=\frac{\partial^2 \mathbb{E}\left[Y_{i d d' g} \mid X_{g} = \mathbf{x}, P_{0g}=p_0, P_{1g}=p_1\right]}{\partial p_0 \partial p_1} \Big/ \frac{\partial^2 \mathbb{E}\left[D_{0g} D_{1g} \mid P_{0g}=p_0, P_{1g}=p_1\right]}{\partial p_0 \partial p_1}, \\
        & Y_{i d d' g} \equiv Y_{ig} \mathbbm{1}\left\{D_{0g}=d, D_{1g}=d^{\prime}\right\} \big\{2\operatorname{sgn}\big(1 - |d-d^{\prime}|\big) - 1 \big\}, \\
        & i, d, d' \in\{0,1\}, g \in \{1, \cdots, G\},
    \end{aligned}
\end{equation}
where $\operatorname{sgn}(x)$ denotes the sign function that indicates the sign of a scalar $x$. The term ${2\operatorname{sgn}(1 - |d - d'|) - 1}$ evaluates to $1$ when $d = d'$ and to $-1$ when $d \neq d'$. The marginal controlled effects are then obtained by taking the difference between the estimated MTR functions.

The propensity score functions, $P_{0g} \equiv P_0(W_g)$ and $P_{1g} \equiv P_1(W_g)$, defined as $\mathbb{P}(D_{ig} = 1 \mid W_g)$ with $W_g \equiv (W_{0g}, W_{1g})$, are not directly observed and must be estimated from the data. The first stage of involves estimating the propensity score functions using a series regression approach. To alleviate the curse of dimensionality, a partially linear additive specification is employed:
\begin{equation} \label{eq:ps_approximate}
    \mathbb{P}\left(D_{i g}=1 \mid W_{g} = w\right) = \varphi_1\left(w_1^{cts}\right)+\cdots+\varphi_{\ell_1}\left(w_{\ell_1}^{cts}\right)+w_{1}^{disc} \vartheta_1 + \cdots + w_{\ell_2}^{disc} \vartheta_{\ell_2}.
\end{equation}
The covariate vector $w$ includes both continuous and discrete components, denoted by $w = (w^{cts}, w^{disc})$, where $w^{cts} = (w^{cts}_1, \cdots, w^{cts}_{\ell_1})$ is an $\ell_1-$dimensional vector of continuous random variables, and $w^{disc} = (w^{disc}_1, \cdots, w^{disc}_{\ell_2})$ is an $\ell_2-$dimensional vector of discrete random variables. 

To preserve flexibility, no parametric restrictions are imposed on the unknown smooth functions $\varphi_1, \ldots, \varphi_{\ell_1}$ associated with the continuous covariates, while the coefficients $\vartheta_1, \ldots, \vartheta_{\ell_2}$ on the discrete covariates remain to be estimated.

Series estimation relies on constructing a basis for smooth functions defined on $\mathbb{R}$, denoted as $\{p_k : k = 1, 2, \dots \}$, such that each continuous function $\varphi_{\ell}$, for $\ell = 1, \dots, \ell_1$, can be approximated arbitrarily well by a linear combination of these basis functions as $k \rightarrow \infty$. Commonly used basis functions include polynomial basis functions, splines, and wavelets. Given a positive integer $\kappa$, define the regressor vector
\begin{equation*}
    P_{\kappa}(w)=\left[p_1\left(w_1^{cts}\right), \ldots, p_\kappa\left(w_1^{cts}\right), \ldots, p_1\left(w_{\ell_1}^{cts}\right), \ldots, p_\kappa\left(w_{\ell_1}^{cts}\right), w_{1}^{disc}, \cdots, w_{\ell_2}^{disc}\right]^{\prime},
\end{equation*}
where the first $\kappa \times \ell_1$ components correspond to basis function expansions of the continuous covariates, and the remaining $\ell_2$ components include the discrete covariates in their original form.

The series estimator of the conditional probability $\mathbb{P}(D_{ig} = 1 \mid W_g)$, for $i \in \{0, 1\}$, is given by
\begin{equation*}
\tilde{P}_i\left(W_g\right)=P_\kappa\left(W_g\right)^{\prime} \hat{\theta}_{\kappa}^i, 
\end{equation*}
where $\hat{\theta}_{\kappa}^i$ is obtained from the least squares optimization problem
\begin{equation*}
    \hat{\theta}_{\kappa}^i = \arg \min_{\theta_\kappa^i \in \mathbb{R}^{\widetilde{\kappa}}} \frac{1}{G} \sum_{g = 1}^{G} \bigg(D_{ig} - P_\kappa(W_g)' \theta_\kappa^i \bigg)^2 ,
\end{equation*}
and $\tilde{\kappa} = \kappa \ell_1 + \ell_2$ denotes the total dimension of the regressor vector $P_\kappa(W_g)$.

\begin{remark}(Series with Lasso regression)
    To increase flexibility in selecting relevant basis terms, one may combine nonparametric series estimation with Lasso regression, which enables automatic selection and regularization of basis terms. The estimation errors of such estimators have been studied in \cite{bickel2009simultaneous} and other related works cited therein.
    Select a positive integer $\kappa$ such that $\kappa \gg g$. Then, the series estimator, $\tilde{P}_i\left(W_g\right)$, with $l_1$-penalization is derived as 
    \begin{equation*}
        \begin{aligned}
            &\hat{\theta}_{\kappa}^i = \arg \min_{\theta_\kappa^i \in \mathbb{R}^{\widetilde{\kappa}}} \frac{1}{G} \sum_{g = 1}^{G} \bigg(D_{ig} - P_\kappa(W_g)' \theta_\kappa^i \bigg)^2 + 2 \lambda \frac{1}{\widetilde{\kappa}} \sum_{j=1}^{\widetilde{\kappa}} \lVert p_j \rVert _G \lvert \theta_j^i \rvert, \\
            & \tilde{P}_i\left(W_g\right)=P_\kappa\left(W_g\right)^{\prime} \hat{\theta}_{\kappa}^i,
        \end{aligned}
    \end{equation*}
    where $\lambda > 0$ is the tuning constant, $p_j(\cdot)$ denotes the $j$-th component of the basis expansion $P_\kappa(\cdot)$, and $\lVert \cdot \rVert _G$ stands for the empirical norm, $\lVert p_j \rVert _G = \sqrt{1 / G \sum_{g=1}^G p_j^2(W_g)}$. In the estimation procedure, the tuning parameter $\lambda$ is selected using cross-validation.
\end{remark}

A finite-sample concern is that the estimated series approximation $\tilde{P}_i(W_g)$ may take values outside the admissible unit interval $[0,1]$. To address this, a standard trimming adjustment can be applied. The trimmed estimator is defined as
\begin{equation*}
    \begin{aligned}
        \hat{P}_i\left(W_g\right)= & \tilde{P}_i\left(W_g\right)+\left(1-\delta-\tilde{P}_i\left(W_g\right)\right) \mathbbm{1}\left\{\tilde{P}_i\left(W_g\right)>1\right\} \\
        & +\left(\delta-\tilde{P}_i\left(W_g\right)\right) \mathbbm{1}\left\{\tilde{P}_i\left(W_g\right)<0\right\},
    \end{aligned}
\end{equation*}
where $\delta > 0$ is a small constant chosen by the researcher. The resulting $\hat{P}_i(W_g)$ thus provides a feasible and bounded estimator of the propensity score $\mathbb{P}(D_{ig} = 1 \mid W_g)$, which is denoted compactly as $\hat{P}_{ig}$ in subsequent analysis.

%{\color{red} (An alternative way could be to apply the machine learning method to obtain the first-stage estimator and the corresponding error bounds.)}

The next step involves estimating the cross-partial derivatives appearing in both the numerator and denominator of the estimand in Equation~\eqref{eq:estimand}. The procedure begins with estimating the denominator of the estimand, namely the cross-partial derivative
\begin{equation*}
    \frac{\partial^2 \mathbb{E}[D_{0g} D_{1g} \mid P_{0g}=p_0, P_{1g}=p_1]}{\partial p_0 \partial p_1}.
\end{equation*}
Local polynomial regression provides a flexible and well-established approach for estimating such derivatives of conditional expectations. Following \citet{fan1996local}, the polynomial order is set to $p = d + 1$, where $d$ denotes the derivative order. Since the object of interest is a second-order derivative, a local cubic regression ($p = 3$) is employed:
\begin{equation} \label{eq:est_copula}
    \begin{aligned}
            \min _{b_0, \cdots, b_9} \sum_{g=1}^G & {\left[D_{0 g} D_{1 g}-b_0-b_1\left(\hat{P}_{0 g}-p_0\right)-\cdots b_4\left(\hat{P}_{0 g}-p_0\right)\left(\hat{P}_{1 g}-p_1\right)\right.} \\
            & \left.-\cdots-b_9\left(\hat{P}_{1 g}-p_1\right)^3\right]^2 K_{h_{G 1}}\left(\hat{P}_g-p\right) \\
        & K_{h_{G 1}}\left(\hat{P}_g-p\right) = K\left(\frac{\hat{P}_{0 g}-p_0}{h_{G 1}}\right) \times K\left(\frac{\hat{P}_{1 g}-p_1}{h_{G 1}}\right),
    \end{aligned}
\end{equation}
where $K(\cdot)$ denotes the kernel function and $h_{G1}$ is the chosen bandwidth parameter. Bandwidths for kernel-based regressions can be selected using $K$-fold cross-validation. The estimated coefficient $\hat{b}_4(p_0, p_1)$ from Equation \eqref{eq:est_copula} serves as an estimator of the cross-partial derivative $\partial^2 \mathbb{E}[D_{0g} D_{1g} \mid P_{0g}=p_0, P_{1g}=p_1] / \partial p_0 \partial p_1$.

The subsequent stage focuses on estimating the cross-partial derivative that appears in the numerator of the estimand, 
\begin{equation*}
    \frac{\partial^2 \mathbb{E}[Y_{i d d' g} \mid X_{g} = \mathbf{x},  P_{0g} = p_0, P_{1g} = p_1]}{\partial p_0 \partial p_1}.
\end{equation*}
Since the covariate vector $X_g$ may be multidimensional, the estimation adopts the semiparametric framework to mitigate the curse of dimensionality.
\begin{assumption}(Partial linear outcomes) \label{as:est_partial_linear}
    Potential outcomes satisfy a partially linear structure of the form
    \begin{equation*}
        Y_{i g}(\mathbf{x}, d, d')= \mathbf{x}'\beta_{idd'} + U_{i g}(d, d'),
    \end{equation*}
    where $\beta_{idd'}$ is a finite-dimensional parameter vector that may vary across units $i \in \{0, 1\}$ and treatment states $(d,d') \in \{0, 1\}^2$, and $U_{ig}(d,d')$ is an unrestricted nonparametric component capturing the remaining heterogeneity. 
    
    The potential outcome is generated according to
    \begin{equation*}
        Y_{i g}(\mathbf{x}, d, d') \equiv m_i(\mathbf{x}, d, d', U_{ig}, U_{-ig})
    \end{equation*}
    where the covariates are fixed at $X_g = \mathbf{x}$ and the treatment assignments are fixed at $(D_{ig}, D_{-ig}) = (d, d')$ exogenously.
\end{assumption}

Under this specification, the conditional expectation, and consequently the marginal treatment response (MTR) function, can be expressed as a semiparametric function of $(\mathbf{x}, p_0, p_1)$, separating the parametric effect of covariates from the nonparametric dependence on the propensity scores:
\begin{equation} \label{eq:semip_estimand}
    \begin{aligned}
        &\frac{\partial^2 \mathbb{E}\left[Y_{i d d' g} \mid X_{g} = \mathbf{x}, P_{0g}=p_0, P_{1g}=p_1\right]}{\partial p_0 \partial p_1} \bigg/ \frac{\partial^2 \mathbb{E}\left[D_{0 g} D_{1 g} \mid P_{0 g}=p_0, P_{1 g}=p_1\right]}{\partial p_0 \partial p_1} \\
        =& \mathbf{x}'\beta_{idd'} + \frac{\partial^2 \mathbb{E}\left[U_{i d d' g} \mid P_{0g}=p_0, P_{1g}=p_1\right]}{\partial p_0 \partial p_1} \bigg/ \frac{\partial^2 \mathbb{E}\left[D_{0 g} D_{1 g} \mid P_{0 g}=p_0, P_{1 g}=p_1\right]}{\partial p_0 \partial p_1},\\
        & U_{i d d' g} \equiv U_{ig}(d, d^{\prime}) \mathbbm{1}\left\{D_{0g}=d, D_{1g}=d^{\prime}\right\} \big\{2\operatorname{sgn}\big(1 - |d-d^{\prime}|\big) - 1 \big\},
    \end{aligned}
\end{equation}

The conditional expectation $\mathbb{E}\left[Y_{ig} \mid D_{0g} = d, D_{1g} = d', P_{0g}, P_{1g}\right]$ can be expressed as
\begin{equation} \label{eq:semip_coefficient}
    \begin{aligned}
        &\mathbb{E}[Y_{i g} \mid D_{0g} = d, D_{1g} = d', P_{0g}, P_{1g}] = \big(\mathbb{E}[X_g \mid D_{0g} = d, D_{1g} = d', P_{0g}, P_{1g}]\big)' \beta_{idd'} \\
        & \qquad \qquad \qquad \qquad \qquad \qquad \qquad \qquad+ \mathbb{E}[U_{i g} \mid D_{0g} = d, D_{1g} = d', P_{0g}, P_{1g}].
    \end{aligned}
\end{equation}
Therefore, conditional on the subsample with $\{D_{0g} = d, D_{1g} = d'\}$, the coefficient vector $\beta_{idd'}$ can be estimated by the least squares regression as
\begin{equation} \label{eq:semip_least_square}
    \begin{aligned}
        &\widehat{\beta}_{idd'}={\left[\sum_{g=1}^G \widetilde{X_{g}}\widetilde{X_{g}}^{\prime}\right]^{-1} } \times \left[\sum_{g=1}^G \widetilde{X_{g}}\Big\{Y_{i g}-\hat{E}_h\left[Y_{i g} \mid \hat{P}_0\left(W_g\right), \hat{P}_1\left(W_{g}\right)\right]\Big\}\right], \\
        & \widetilde{X_{g}} = X_{g} -\hat{E}_h\left[X_{g} \mid \hat{P}_0\left(W_{g}\right), \hat{P}_1\left(W_{g}\right)\right]
    \end{aligned}
\end{equation}
where $\hat{E}_h[\cdot \mid \cdot]$ represents a kernel regression estimator with selected bandwidth $h$. 

The residual then follows as
\begin{equation*}
    \widehat{U}_{i d d' g} = Y_{i d d' g} - X_{g}' \widehat{\beta}_{idd'} \big\{2\operatorname{sgn}\big(1 - |d-d^{\prime}|\big) - 1 \big\}
\end{equation*}
which serves as an estimator of the unobserved component $U_{i d d' g}$. 

In the last step, use the sample 
\begin{equation*}
    \left\{\left(\hat{U}_{i d d' g}, \hat{P}_0(W_g), \hat{P}_0(W_g)\right): g=1, \ldots, G\right\}
\end{equation*}
to estimate the cross-partial derivative $\partial^2 \mathbb{E}[U_{idd'g} \mid P_{0g}=p_0, P_{1g}=p_1] / \partial p_0 \partial p_1
$ through a local polynomial regression of order three,
\begin{equation*}
    \begin{aligned}
        \min _{c_0, \cdots, c_9} \sum_{g=1}^G & {\left[\hat{U}_{i d d g^{\prime}}-c_0-c_1\left(\hat{P}_{0 g}-p_0\right)-\cdots c_4\left(\hat{P}_{0 g}-p_0\right)\left(\hat{P}_{1 g}-p_1\right)\right.} \\
        & \left.-\cdots-c_9\left(\hat{P}_{1 g}-p_1\right)^3\right]^2 K_{h_{G 2}}\left(\hat{P}_g-p\right), \\
        & K_{h_{G 2}}\left(\hat{P}_g-p\right)=K\left(\frac{\hat{P}_{0 g}-p_0}{h_{G 2}}\right) \times K\left(\frac{\hat{P}_{1 g}-p_1}{h_{G 2}}\right).
    \end{aligned}
\end{equation*}
The resulting coefficien $\hat{c}_4(d, d'; p_0, p_1)$ consistently estimates $\partial^2 \mathbb{E}[U_{idd'g} \mid P_{0g}=p_0, P_{1g}=p_1] / \partial p_0 \partial p_1 $.

Finally, the marginal treatment response functions $m_{i g}^{(\mathbf{x}, d, d')}(p_0, p_1)$ are estimated as 
\begin{equation*} 
    \widehat{m}_{i g}^{(\mathbf{x}, d, d^{\prime})}(p_0, p_1) = \mathbf{x}' \hat{\beta}_{id d'} + \frac{\hat{c}_4(d, d'; p_0, p_1)}{\hat{b}_4(p_0, p_1)},
\end{equation*}
where $\widehat{c}_4(d, d'; p_0, p_1)$ and $\widehat{b}_4(p_0, p_1)$ are the local polynomial estimators of the cross-partial derivatives of $\mathbb{E}[U_{idd'g} \mid P_{0g}, P_{1g}]$ and $\mathbb{E}[D_{0g} D_{1g} \mid P_{0g}, P_{1g}]$, respectively.

The next section derives the asymptotic distribution of the estimated marginal treatment response functions, abstracting from the role of covariates to focus on the sampling behavior of the nonparametric components, 
\begin{equation} \label{eq:semip_estimator}
    \frac{\hat{c}_4(d, d'; p_0, p_1)}{\hat{b}_4(p_0, p_1)}.
\end{equation}

\subsubsection{Asymptotic Properties}

\subsubsection*{Uniform consistent rate of propensity score} \label{subsec:ps_consistent}
Cubic spline basis functions $\{p_k: k = 1, 2, \cdots\}$ are employed to approximate the nonparametric components of the propensity score functions. The following assumptions, adapted from \cite{belloni2015some}, provide the regularity conditions required to establish the uniform convergence rate of the series estimators for the propensity score functions.

\begin{assumption}(Series estimation) \label{as:series_est}
    (i) The eigenvalues of $\mathbb{E}[P_\kappa(w_g)P_\kappa(w_g)']$ are bounded above and away from zero uniformly over $G$. (ii) Each function $\varphi_{i} \in \mathcal{G}$ in Equation \eqref{eq:ps_approximate}, where $\mathcal{G}$ is a set of functions $f$ in Hölder classes with exponent $s$, $\Sigma_s(\mathcal{W})$, such that $\|f\|_s$ is bounded from above uniformly over $\mathcal{G}$. (iii) The support of $W^{cts}$ is known and is a Cartesian product of compact connected intervals on which $W^{cts}$ has a probability density function that is bounded away from zero.
\end{assumption}

\begin{lemma}(Uniform rate of propensity score)
    Under Assumptions \ref{as:est_assumption}-\ref{as:series_est}, we have 
    \begin{equation*}
        \max _{g=1, \ldots, G}\left|\hat{P}_i\left(W_g\right)-P_i\left(W_g\right)\right|=O_p\left[\sqrt{\frac{\kappa \log \kappa}{G}}+\kappa^{-s}\right],
    \end{equation*}
    where $\kappa \rightarrow \infty$ as $G \rightarrow \infty$, $\kappa^{m/(m-2)}\operatorname{log}\kappa / G = O(1)$ for any $m > 2$, and $\kappa^{2-2 s} / G = O(1)$.
\end{lemma}

%\subsubsection*{Asymptotics for the parametric part of the outcome equation}

\subsubsection*{Asymptotic properties of cross-derivative estimators} \label{subsec:asymp_cross_deriv}
As shown in Equation~\eqref{eq:semip_estimator}, the proposed estimator is expressed as the ratio of two estimated cross-partial derivatives of conditional mean functions. This subsection derives the asymptotic properties these two cross-derivative estimators under the semiparametric estimation procedure described in Section~\ref{subsec:estimation}.

To derive the asymptotic properties of the cross-partial derivative estimators, impose the following assumption:
\begin{assumption}(Local polynomial regression) \label{as:local_polynomial}
    (i) $\mathbb{E}[D_{0 g} D_{1 g} \mid P_{0 g}=p_0, P_{1 g}=p_1]$ and $\mathbb{E}[U_{i d d^{\prime} g} \mid P_{0 g}=p_0, P_{1 g}=p_1]$ are $(p+1)$-time continuously differentiable, $p \geq 3$. (ii) The conditional distributions of $D_{0 g} D_{1 g} \mid P_{0 g}, P_{1 g}$ and $U_{i d d^{\prime} g} \mid P_{0 g}, P_{1 g}$ are continuous at the point $(p_0, p_1)$. (iii) The kernel $K \in L_1$ is bounded with compact support, and $\|u\|^{4 p} K(u) \in L_1$, $\|u\|^{4 p + 2} K(u) \rightarrow 0$ as $\|u\| \rightarrow \infty$. 
\end{assumption}
This assumption ensures sufficient smoothness of the underlying conditional mean functions and regularity of the kernel function, which together guarantee the validity of local polynomial approximations.

\begin{lemma}(Convergence rates of cross-derivative estimators)
    Under Assumptions \ref{as:est_assumption}-\ref{as:local_polynomial}, the convergence rates of estimators $\hat{b}_4(p_0, p_1)$ and $\hat{c}_4(d, d'; p_0, p_1)$, as calculated in Section \ref{subsec:estimation}, can be derived as 
    \begin{equation*}
        \begin{aligned}
            &\hat{b}_4(p_0, p_1) - \frac{\partial^2}{\partial p_0 \partial p_1} \mathbb{E}\big[D_{0g}D_{1g} \mid P_{0g} = p_0, P_{1g} = p_1\big] \\
            =& O_P\Big[(G h_{G1}^6)^{-1/2} + \max _{g: 1 \leq g \leq G}|\hat{P}_{0g}-P_{0g}| + \max _{g: 1 \leq g \leq G}|\hat{P}_{1g}-P_{1g}| + h_{G1}^4\Big], \\
            &\hat{c}_4(d, d'; p_0, p_1) - \frac{\partial^2}{\partial p_0 \partial p_1} \mathbb{E}\big[U_{i d d' g}  \mid P_{0g} = p_0, P_{1g} = p_1\big] \\
            =& O_P\Big[(G h_{G2}^6)^{-1/2} + \max _{g: 1 \leq g \leq G}|\hat{P}_{0g}-P_{0g}| + \max _{g: 1 \leq g \leq G}|\hat{P}_{1g}-P_{1g}| + h_{G2}^4\Big],
        \end{aligned}
    \end{equation*}
    where $h_{G1}, h_{G2}$ are the bandwidths selected to estimate $\hat{b}_4(p_0, p_1)$ and $\hat{c}_4(d, d'; p_0, p_1)$.
\end{lemma}

\begin{proof}
    See Appendix \ref{app:cvg_rate_cross_deriv}.
\end{proof}

\subsubsection*{Asymptotic distribution of the marginal treatment response}
This section characterizes the asymptotic distribution of the marginal treatment response functions, abstracting from covariate effects. The estimator, defined in Equation~\eqref{eq:semip_estimator}, is constructed as the ratio of two estimated cross-partial derivatives of conditional mean functions. To establish the asymptotic properties of this estimator, the following assumptions are imposed.

\begin{assumption}(Asymptotic distribution) \label{as:asym_dist_ratio}
    (i) $\max _{g=1, \ldots, G}\big|\hat{P}_i(Z_{0 g}, Z_{1 g})-P_i(Z_{0 g}, Z_{1 g})\big|=o_p\big[(G h_{G 1}^6)^{-1 / 2}\big]$. (ii) $h_{G 1}, h_{G 2} \rightarrow 0, G h_{G 1}^6, G h_{G 2}^6 \rightarrow \infty$ as $G \rightarrow \infty$, $h_{G 2}=o(h_{G 1})$, $h_{G 1}, h_{G 2}=o(G^{-1 / 10})$.
\end{assumption}

\begin{theorem}(Asymptotic distributions of the ratio estimator) \label{thm:asymp_dist_ratio}
    Under Assumptions \ref{as:est_assumption}-\ref{as:asym_dist_ratio}, the asymptotic distribution of $\hat{c}_4(d, d^{\prime} ; p_0, p_1) / \hat{c}_4(d, d^{\prime} ; p_0, p_1)$, $d, d' \in \{0,1\}$, can be characterized as 
    \begin{equation*}
        \begin{aligned}
            & \big(G h_{G2}^6\big)^{1/2}\Bigg\{\frac{\hat{c}_4(d, d'; p_0, p_1)}{\hat{b}_4(p_0,p_1)} - \frac{c_4(d, d'; p_0, p_1)}{b_4(p_0,p_1)} \Bigg\} \\
            \xrightarrow{d}& N\Bigg(0, \frac{\sigma^2(p_0, p_1)}{\big(b_4(p_0, p_1)\big)^2 f(p_0, p_1)} \big(M^{-1}\Gamma M^{-1}\big)_{5,5}\Bigg),
        \end{aligned}
    \end{equation*}
    where $\sigma^2(d, d'; p_0, p_1) = \text{Var}(U_{i d d' g} \mid P_{0g} = p_0, P_{1g} = p_1)$, $f(p_0, p_1)$ denotes the density of $(P_{0g}, P_{1g})$ evaluated at the point $(p_0, p_1)$, and $A_{5,5}$ denotes the element located in the fifth row and fifth column of a matrix $A$. The definitions of the matrices $M$ and $\Gamma$ are presented in the Appendix \ref{app:asymp_dist_ratio}.
\end{theorem}

\begin{proof}
    See Appendix \ref{app:asymp_dist_ratio}.
\end{proof}

Finally, the asymptotic distributions of the MCSEs and MCDEs are derived. Their estimators are constructed using the estimated marginal treatment response functions:
\begin{equation*}
    \begin{aligned}
        &\widehat{\text{MCSE}}_i(\mathbf{x}, d; p_0, p_1) = \widehat{m}_{i g}^{(\mathbf{x}, d, 1)}(p_0, p_1) - \widehat{m}_{i g}^{(\mathbf{x}, d, 0)}(p_0, p_1), d \in \{0,1\}, \\
        &\widehat{\text{MCDE}}_i(\mathbf{x}, d; p_0, p_1) = \widehat{m}_{i g}^{(\mathbf{x}, 1, d)}(p_0, p_1) - \widehat{m}_{i g}^{(\mathbf{x}, 0, d)}(p_0, p_1), d \in \{0,1\}. 
    \end{aligned}
\end{equation*}

Assuming that the differences $\hat{c}_4(d, d'; p_0, p_1) / \hat{b}_4(p_0,p_1) - c_4(d, d'; p_0, p_1) / b_4(p_0,p_1)$ are asymptotically independent across different values of $d, d' \in \{0,1\}$, the asymptotic distributions of $\widehat{\text{MCSE}}_i(\mathbf{x}, d; p_0, p_1)$ and $\widehat{\text{MCDE}}_i(\mathbf{x}, d; p_0, p_1)$ follow in Theorem \ref{thm:asymp_dist_ratio}. 

\begin{corollary}
    Suppose that $\big(\hat{c}_4(d, d'; p_0, p_1) / \hat{b}_4(p_0,p_1) - c_4(d, d'; p_0, p_1) / b_4(p_0,p_1) \big)$ are asymptotically independent across different values of $d, d' \in \{0,1\}$, and that Assumptions \ref{as:est_assumption}–\ref{as:asym_dist_ratio} are satisfied. Then, the asymptotic distributions of $\widehat{\text{MCSE}}_i(\mathbf{x}, d; p_0, p_1)$ and $\widehat{\text{MCDE}}_i(\mathbf{x}, d; p_0, p_1)$ can be characterized as  
    \begin{equation*}
        \begin{aligned}
            &\big(G h_{G2}^6\big)^{1/2}\Bigg\{\widehat{\text{MCSE}}(\mathbf{x}, d; p_0, p_1) - \text{MCSE}(\mathbf{x}, d; p_0, p_1)\Bigg\} \\
        &\xrightarrow{d} N\Bigg(0, \frac{\sigma^2(1, d; p_0, p_1) + \sigma^2(0, d; p_0, p_1)}{\big(b_4(p_0, p_1)\big)^2 f(p_0, p_1)} \big(M^{-1}\Gamma M^{-1}\big)_{5,5}\Bigg), \\
        &\big(G h_{G2}^6\big)^{1/2}\Bigg\{\widehat{\text{MCDE}}(\mathbf{x}, d; p_0, p_1) - \text{MCDE}(\mathbf{x}, d; p_0, p_1)\Bigg\} \\
        &\xrightarrow{d} N\Bigg(0, \frac{\sigma^2(d, 1; p_0, p_1) + \sigma^2(d, 0; p_0, p_1)}{\big(b_4(p_0, p_1)\big)^2 f(p_0, p_1)} \big(M^{-1}\Gamma M^{-1}\big)_{5,5}\Bigg)
        \end{aligned}
    \end{equation*}
\end{corollary}

\subsection{Parametric procedure} \label{sec:para_procedure}

\subsubsection{Parametric estimation}
The estimators introduced in Section~\ref{subsec:semiparametric} converge at nonparametric rates, requiring sufficiently large sample sizes to yield reliable estimates. Moreover, when the group size $n>2$, the conditional expectations in the estimand involve higher-dimensional conditioning, further slowing the rate of convergence. Consequently, in settings with limited sample sizes or large groups, it may be preferable to impose parametric assumptions to estimate parameters of interest. This section develops a parametric approach for estimation and inference procedure.

The estimation continues to rely on Assumption~\ref{as:est_assumption}, while introducing the following additional parametric assumptions.

\begin{assumption} \label{as:para_assump} The following specifications are imposed in the parametric setting.
    \begin{enumerate}
        \item (Propensity score) For the treatment assignment equation $D_{ig}=\mathbbm{1}\{\widetilde{V}_{ig} \leq h_i(W_g)\}$ of individual $i$ in group $g$, $i \in \{0, 1\}$, assume that $h_i(\cdot)$ is a $K_1$-th order polynomial function of $W_g = (W_{g,1}, \cdots, W_{g,\ell})' \in \mathbb{R}^\ell$:
    \begin{equation*}
        h_i(W_g) = \sum_{\boldsymbol{k} \in \mathcal{K}_{\ell,K_1}} \theta_{i\boldsymbol{k}} \cdot \prod_{j=1}^\ell W_{g,j}^{k_j},
    \end{equation*}
    where $\boldsymbol{k} = (k_1, \ldots, k_\ell) \in \mathbb{N}_0^\ell$ is a multi-index, $\mathcal{K}_{\ell,K_1} = \{ \boldsymbol{k} \in \mathbb{N}_0^\ell : \sum_{j=1}^\ell k_j \leq K_1 \}$, and $(\theta_{i\boldsymbol{k}})_{\boldsymbol{k} \in \mathcal{K}_{\ell,K_1}} \equiv \theta_{i}$ are polynomial coefficients. Additionally, assume that the unobserved heterogeneity $\widetilde{V}_{ig}$ follows a standard normal distribution: $\widetilde{V}_{ig} \sim N(0, 1)$.
    \item (Copula) Assume that the joint dependence structure of the unobserved heterogeneities $V_{0g}$ and $V_{1g}$ is characterized by a Gaussian copula with correlation parameter $\rho \in [-\varepsilon, \varepsilon]$, where $\varepsilon$ is a constant such that $0< \varepsilon <1$. Specifically, let $V_{ig} = \Phi(\widetilde{V}_{ig})$ for $i \in \{0,1\}$, where $\Phi(\cdot)$ denotes the standard normal cumulative distribution function. The copula of $(V_{0g}, V_{1g})$, denoted by $C_{V_{0g},V_{1g}}(\cdot,\cdot)$, is then given by the Gaussian copula with correlation $\rho$:
    \begin{equation*}
        C_{V_{0g}, V_{1g}}(v_0, v_1) = \Phi_\rho\big(\Phi^{-1}(v_0), \Phi^{-1}(v_1)\big), \forall (v_0, v_1) \in (0, 1)^2,
    \end{equation*}
    where $\Phi_{\rho}$ is the bivariate normal CDF with zero means, unit variances, and correlation $\rho$, $\Phi^{-1}$ denotes the inverse of the standard normal CDF, and $\rho$ is an unknown parameter.
    \item (Marginal treatment response) It is assumed that the potential outcome follows a partially linear specification in the covariates, $Y_{i g}(\mathbf{x}, d, d')=\mathbf{x}' \beta_{i d d'}+U_{i g}(d, d')$, where $U_{i g}(d, d')$ satisfies the condition stated below:
    \begin{equation*}
        \begin{aligned}
            \mathbb{E}[U_{i g}(d, d') \mid V_{0g} = v_0, V_{1g} = v_1] =& \alpha_{i d d', 0} + \alpha_{i d d',1} \Phi^{-1}(v_0) \\
            &+ \alpha_{i d d', 2} \Phi^{-1}(v_1) + \alpha_{i d d', 3} \Phi^{-1}(v_0) \Phi^{-1}(v_1),
        \end{aligned}
    \end{equation*}
    for all $(v_0, v_1) \in (0, 1)^2$, and $\alpha_{i d d'} \equiv (\alpha_{i d d', 0}, \alpha_{i d d',1}, \alpha_{i d d', 2}, \alpha_{i d d', 3})'$ denotes the vector of unknown coefficients that may be heterogeneous across individuals $i$ and treatment states $(d,d')$.
    \end{enumerate}
\end{assumption}

The imposed parametric assumptions are standard in the marginal treatment effect (MTE) literature and provide a tractable yet flexible framework for estimation and inference. Modeling the selection rule as \(D_{ig}=\mathbbm{1}\{\widetilde V_{ig}\le h_i(W_g)\}\) with a parametric index $h_i(W_g)$ and a standard normal unobserved term corresponds to the probit-type latent index widely used in practical implementations of the MTE framework (see, e.g., \cite{carneiro2011estimating}; \cite{kline2016evaluating}). The polynomial specification of $h_i(\cdot)$ provides sufficient flexibility to capture nonlinear relationships between instruments and covariates.

The Gaussian copula structure for $(V_{0g},V_{1g})$ is also a common parametric choice that facilitate likelihood-based estimation and allow dependence in unobserved heterogeneity across group members. Such assumptions have been adopted in the interference literature, including \cite{hoshino2023treatment}, to capture correlated unobservables within the group.

The parametric specification of the MTR function is consistent with the functional-form assumptions commonly employed in the MTE literature to achieve point identification when the available instruments provide limited variation. Under SUTVA, \cite{brinch2017beyond} show that imposing a parametric structure on the MTR function allows for the identification of heterogeneous treatment effects even with discrete instruments. Analogously, in the presence of spillovers, a similar approach can be applied by specifying the MTR function $\mathbb{E}[U_{ig}(d, d') \mid V_{ig} = v_0, V_{-ig} = v_1]$ as a polynomial expansion in the unobserved heterogeneities, $\Phi^{-1}(v_0)$ and $\Phi^{-1}(v_1)$. This formulation accommodates spillover effects from peers' treatments $d'$ and captures potential dependence between group members through involving $\Phi^{-1}(v_1)$. Moreover, the parametric formulation enables extrapolation beyond the observed support of the propensity scores, thereby allowing for the identification of policy-relevant treatment effects (PRTEs) even when instrumental variables exhibit limited or discrete variation 
\citep{brinch2017beyond}.

The objective is to estimate the marginal treatment response function, $m_{i g}^{(\mathbf{x}, d, d')}(v_0, v_1) = \mathbb{E}[Y_{ig}(\mathbf{x}, d, d') \mid V_{0g} = v_0, V_{1g} = v_1]$, for any $\mathbf{x} \in \mathcal{X}$, $d, d' \in \{0, 1\}$, and $(v_0, v_1) \in (0, 1)^2$.

As in the semiparametric case, the first step involves estimating the propensity score functions $P_{0}(W_g), P_{1}(W_g)$, where $W_{i g}=(Z_{i g}, X_{i g}), W_g=(W_{0 g}, W_{1 g}) \in \mathbb{R}^\ell$. Under the first specification in Assumption \ref{as:para_assump}, and assuming that the instruments and covariates are independent of the unobserved heterogeneity $V_{ig}$, we can express the propensity score function as 
\begin{equation*}
    P_{i g} \equiv \mathbb{P}(D_{ig} = 1 \mid W_g)  = \Phi\Bigg(\sum_{\boldsymbol{k} \in \mathcal{K}_{\ell,K_1}} \theta_{i\boldsymbol{k}} \cdot \prod_{j=1}^\ell W_{g,j}^{k_j} \Bigg).
\end{equation*}
The polynomial coefficients can be estimated using standard maximum likelihood methods:
\begin{equation*}
    \hat{\theta}_i = \arg \max_{(\theta_{i\boldsymbol{k}})_{\boldsymbol{k} \in \mathcal{K}_{\ell,K_1}}} \sum_{g=1}^G \Bigg[ D_{ig} \log \Phi\Big(\sum_{\boldsymbol{k} \in \mathcal{K}_{\ell,K_1}} \theta_{i\boldsymbol{k}} \cdot \prod_{j=1}^\ell W_{g,j}^{k_j} \Big) + (1 - D_{ig}) \log \Bigg(1 - \Phi\Big(\sum_{\boldsymbol{k} \in \mathcal{K}_{\ell,K_1}} \theta_{i\boldsymbol{k}} \cdot \prod_{j=1}^\ell W_{g,j}^{k_j} \Big) \Bigg) \Bigg].
\end{equation*}

Once the polynomial coefficients $\hat{\theta}_i$ are estimated, they can be substituted into $P_{ig}$ to obtain the estimated propensity score as
\begin{equation*}
    \widehat{P}_{i g} = \Phi\Bigg(\sum_{\boldsymbol{k} \in \mathcal{K}_{\ell,K_1}} \widehat{\theta}_{i\boldsymbol{k}} \cdot \prod_{j=1}^\ell W_{g,j}^{k_j} \Bigg).
\end{equation*}

The next step is to estimate the joint dependence structure of $V_{0g}$ and $V_{1g}$. Under the second specification in Assumption \ref{as:para_assump}, this dependence is modeled by a Gaussian copula with correlation parameter $\rho$. Consequently, the second step of our procedure focuses on estimating $\rho$. The identification results imply the following equations,
\begin{equation*}
    \begin{aligned}
        & \mathbb{P}(D_{0g} = 1, D_{1g} = 1 \mid P_{0g} = p_0, P_{1g} = p_1) = \Phi_\rho\big(\Phi^{-1}(P_{0g}), \Phi^{-1}(P_{1g})\big), \\
        & \mathbb{P}(D_{0g} = 1, D_{1g} = 0 \mid P_{0g} = p_0, P_{1g} = p_1) = p_0 - \Phi_\rho\big(\Phi^{-1}(P_{0g}), \Phi^{-1}(P_{1g})\big), \\
        & \mathbb{P}(D_{0g} = 0, D_{1g} = 1 \mid P_{0g} = p_0, P_{1g} = p_1) = p_1 - \Phi_\rho\big(\Phi^{-1}(P_{0g}), \Phi^{-1}(P_{1g})\big), \\
        & \mathbb{P}(D_{0g} = 0, D_{1g} = 0 \mid P_{0g} = p_0, P_{1g} = p_1) = 1 - p_0 - p_1 + \Phi_\rho\big(\Phi^{-1}(P_{0g}), \Phi^{-1}(P_{1g})\big).
    \end{aligned}
\end{equation*}

Therefore, $\rho$ can be estimated using the maximum likelihood, substituting the first-stage estimates $\widehat{P}_{0g}$ and $\widehat{P}_{1g}$ for the true propensity scores $P_{0g}$ and $P_{1g}$,
\begin{equation*}
    \begin{aligned}
        \hat{\rho} = \arg \max_{\rho \in [-\varepsilon, \varepsilon]} \sum_{g = 1}^{G} & \Bigg\{ D_{0g}D_{1g} \log\bigg(\Phi_\rho\big(\Phi^{-1}(\widehat{P}_{0g}), \Phi^{-1}(\widehat{P}_{1g})\big)\bigg) +\\
        & D_{0g}(1 - D_{1g}) \log\bigg(\widehat{P}_{0g} - \Phi_\rho\big(\Phi^{-1}(\widehat{P}_{0g}), \Phi^{-1}(\widehat{P}_{1g})\big)\bigg) + \\
        & (1 - D_{0g})D_{1g} \log\bigg(\widehat{P}_{1g} - \Phi_\rho\big(\Phi^{-1}(\widehat{P}_{0g}), \Phi^{-1}(\widehat{P}_{1g})\big)\bigg) + \\
        & (1 - D_{0g})(1 - D_{1g}) \log\bigg(1 - \widehat{P}_{0g} - \widehat{P}_{1g} + \Phi_\rho\big(\Phi^{-1}(\widehat{P}_{0g}), \Phi^{-1}(\widehat{P}_{1g})\big)\bigg) \Bigg\}.
    \end{aligned}
\end{equation*}

The final step involves estimating the marginal treatment response $\mathbb{E}[Y_{ig}(\mathbf{x}, d, d') \mid V_{0g} = v_0, V_{1g} = v_1]$. Under the third specification in Assumption \ref{as:para_assump}, this function admits the following parametric representation,
\begin{equation*}
    \begin{aligned}
        m_{i g}^{(\mathbf{x}, d, d')}(v_0, v_1) \equiv& \mathbb{E}\big[Y_{ig}(\mathbf{x}, d, d') \mid V_{0g} = v_0, V_{1g} = v_1 \big] \\
        =& \mathbf{x}' \beta_{i d d'} + \alpha_{i d d', 0} + \alpha_{i d d',1} \Phi^{-1}(v_0) \\ 
        &+ \alpha_{i d d', 2} \Phi^{-1}(v_1) + \alpha_{i d d', 3} \Phi^{-1}(v_0) \Phi^{-1}(v_1).
    \end{aligned}
\end{equation*}
Hence, the last stage of our procedure focuses on estimating the coefficient vectors $\beta_{i d d'}$ and $c_{i d d'}$. For illustration, consider the case $d = 1$ and $d' = 1$.

By combining the identification results with the third specification in Assumption~\ref{as:para_assump}, the following relationship is obtained:
\begin{equation*}
    \begin{aligned}
        & \mathbb{E}[Y_{ig} D_{0g} D_{1g} \mid X_g = \mathbf{x}, P_{0g} = p_0, P_{1g} = p_1] \\
        =& \int_0^{p_1} \int_0^{p_0} \mathbb{E}\big[U_{ig}(1,1) \mid V_{0g}=v_0, V_{1g}=v_1\big] c_{V_{0g}, V_{1g}}(v_0, v_1) d v_0 d v_1 \\
        &+ \mathbf{x}'\beta_{i 11} \mathbb{P}\big(D_{0g} D_{1g} \mid P_{0g} = p_0, P_{1g} = p_1\big) \\
        =& \alpha_{i 11, 0} \int_0^{p_1} \int_0^{p_0} c_{V_{0g}, V_{1g}}(v_0, v_1) d v_0 v_1 + \alpha_{i 11, 1} \int_0^{p_1} \int_0^{p_0} \Phi^{-1}(v_0) c_{V_{0g}, V_{1g}}(v_0, v_1) d v_0 v_1 \\
        &+ \alpha_{i 11, 2} \int_0^{p_1} \int_0^{p_0} \Phi^{-1}(v_1) c_{V_{0g}, V_{1g}}(v_0, v_1) d v_0 v_1 \\
        &+ \alpha_{i 11, 3} \int_0^{p_1} \int_0^{p_0} \Phi^{-1}(v_0) \Phi^{-1}(v_1) c_{V_{0g}, V_{1g}}(v_0, v_1) d v_0 v_1 \\
        &+ \mathbf{x}'\beta_{i 11} \mathbb{P}\big(D_{0g} D_{1g} \mid P_{0g} = p_0, P_{1g} = p_1\big) \\
        \equiv& \alpha_{i 11, 0} I_{11}^0(p_0, p_1, \rho) + \alpha_{i 11, 1} I_{11}^1(p_0, p_1, \rho) + \alpha_{i 11, 2} I_{11}^2(p_0, p_1, \rho) + \alpha_{i 11, 3} I_{11}^3(p_0, p_1, \rho) \\
        &+ \mathbf{x}'\beta_{i 11} \Phi_\rho\big(p_0, p_1\big). 
    \end{aligned}
\end{equation*}
In the last line, $I_{11}^1(p_0, p_1, \rho), I_{11}^2(p_0, p_1, \rho)$, and $I_{11}^3(p_0, p_1, \rho)$ denote integrals that depend on $p_0, p_1$, and the correlation parameter $\rho$ of the Gaussian copula density $c_{V_{0g}, V_{1g}}(\cdot, \cdot)$ when $d = 1$ and $d' = 1$. Given that the propensity scores and the correlation have been estimated in the previous two steps, $\widehat{P}_{0g}$, $\widehat{P}_{1g}$, and $\hat{\rho}$  are substituted for their true values. The coefficient vectors $\alpha_{i11}$ and $\beta_{i11}$ are then estimated using the following least squares regression,
\begin{equation*}
    \begin{aligned}
        \big(\hat{\alpha}_{i 11}', \hat{\beta}_{i 11}'\big)' = \arg \min_{(\alpha_{i 11}', \beta_{i 11}')'} &\sum_{g = 1}^{G}\bigg[Y_{ig} D_{0g} D_{1g} - \alpha_{i 11, 0} I_{11}^1(\widehat{P}_{0g}, \widehat{P}_{1g}, \hat{\rho})  - \alpha_{i 11, 1} I_{11}^1(\widehat{P}_{0g}, \widehat{P}_{1g}, \hat{\rho}) \\
        &-\alpha_{i 11, 2} I_{11}^3(\widehat{P}_{0g}, \widehat{P}_{1g}, \hat{\rho}) - \alpha_{i 11, 3} I_{11}^4(\widehat{P}_{0g}, \widehat{P}_{1g}, \hat{\rho}) - 
    \widetilde{X}_{g}' \beta_{i11}\bigg]^2,
    \end{aligned}
\end{equation*}
where $\widetilde{X}_{g} \equiv X_g \cdot \mathbb{P}(D_{0g}D_{1g} \mid \widehat{P}_{0g}, \widehat{P}_{1g})$. A similar procedure can be applied to estimate the coefficient vectors $\beta_{i d d'}$ and $\alpha_{i d d'}$ for other treatment combinations $(d, d')$. The estimated marginal treatment response function, $\widehat{m}_{i g}^{(\mathbf{x}, d, d')}(v_0, v_1)$, is obtained by substituting $(\hat{\alpha}_{i 11}', \hat{\beta}_{i 11}')'$ for $(\alpha_{i 11}', \beta_{i 11}')'$.

Finally, the MCDEs and MCSEs are obtained by taking differences of the estimated marginal treatment response functions,
\begin{equation*}
    \begin{aligned}
        &\widehat{\text{MCSE}}_i(\mathbf{x}, d; p_0, p_1) = \widehat{m}_{i g}^{(\mathbf{x}, d, 1)}(p_0, p_1) - \widehat{m}_{i g}^{(\mathbf{x}, d, 0)}(p_0, p_1), d \in \{0,1\}, \\
        &\widehat{\text{MCDE}}_i(\mathbf{x}, d; p_0, p_1) = \widehat{m}_{i g}^{(\mathbf{x}, 1, d)}(p_0, p_1) - \widehat{m}_{i g}^{(\mathbf{x}, 0, d)}(p_0, p_1), d \in \{0,1\}. 
    \end{aligned}
\end{equation*}

\subsubsection{Parametric asymptotic results}
This section introduces a set of assumptions under which the parametric estimators are consistent.

\begin{assumption} \label{as:para_consist1}
    (Parametric first stage) In the first stage of propensity score estimation, for each individual $i \in \{0, 1\}$, we assume that 
    \begin{enumerate}
        \item $\theta_i \in \Theta_i$, where the parameter space $\Theta_i$ is compact.
        \item The true parameter $\theta_{i0}$ is unique.
        \item Let $l(\theta_i; D_{ig}, W_g)$ denote the log-likelihood of individual $i$'s treatment in group $g$:
        \begin{equation*}
            \begin{aligned}
                l(\theta_i; D_{ig}, W_g) =& D_{i g} \log \Phi\bigg(\sum_{k \in \mathcal{K}_{\ell, K_1}} \theta_{i k} \cdot \prod_{j=1}^{\ell} W_{g, j}^{k_j}\bigg) \\
                &+(1-D_{i g}) \log \Bigg(1-\Phi\bigg(\sum_{k \in \mathcal{K}_{\ell, K_1}} \theta_{i k} \cdot \prod_{j=1}^{\ell} W_{g, j}^{k_j}\bigg)\Bigg).
            \end{aligned}
        \end{equation*}
        The log-likelihood function $l(\theta_i; D_{ig}, W_g)$ satisfies the following conditions:
        \begin{enumerate}
            \item [(i)] $\mathbb{E}\big[\sup_{\theta_i \in \Theta_i}|l(\theta_i; D_{ig}, W_g)|\big] < \infty$.
            \item [(ii)] $\mathbb{E}\big[\nabla^2_{\theta_i} l(\theta_i; D_{ig}, W_g)\big]$ exists and is invertible.
            \item [(iii)] $\mathbb{E}\big[\sup_{\theta_i \in \Theta_i}|| \nabla^2_{\theta_i} l(\theta_i; D_{ig}, W_g)||\big] < \infty$.
        \end{enumerate}
    \end{enumerate}
\end{assumption}

Under Assumptions \ref{as:para_assump} and \ref{as:para_consist1}, the estimator $\hat{\theta}_i$ obtained in the first stage is consistent.

\begin{lemma} \label{lemma:para_consist1}
    Suppose Assumptions \ref{as:est_assumption}, \ref{as:para_assump} and \ref{as:para_consist1} hold. Then, for each $i \in \{0, 1\}$, the estimator $\hat{\theta}_i \asto \theta_{i0}$ as $G \rightarrow \infty$. 
\end{lemma}

\begin{proof}
    See Appendix \ref{app:proof_parametric_first}.
\end{proof}

To establish the consistency of the second-stage estimator of the Gaussian copula correlation parameter $\rho$, the following additional assumption is imposed.

\begin{assumption} \label{as:para_consist2}
    (Parametric second stage) In the second stage, to estimate the correlation parameter $\rho$ of the Gaussian copula, we impose the following conditions.
    \begin{enumerate}
        \item The true parameter $\rho_0$ is unique.
        \item Define the log-likelihood of joint treatments in group $g$ as 
        \begin{equation*}
            \begin{aligned}
                &l(\rho, \theta; D_{g}, W_{g}) = \tilde{l}(\rho;D_g, P_g)\\
                \equiv& D_{0g}D_{1g} \log\bigg(\Phi_\rho\big(\Phi^{-1}(P_{0g}), \Phi^{-1}(P_{1g})\big)\bigg) +\\
        & D_{0g}(1 - D_{1g}) \log\bigg(P_{0g} - \Phi_\rho\big(\Phi^{-1}(P_{0g}), \Phi^{-1}(P_{1g})\big)\bigg) + \\
        & (1 - D_{0g})D_{1g} \log\bigg(P_{1g} - \Phi_\rho\big(\Phi^{-1}(P_{0g}), \Phi^{-1}(P_{1g})\big)\bigg) + \\
        & (1 - D_{0g})(1 - D_{1g}) \log\bigg(1 - P_{0g} - P_{1g} + \Phi_\rho\big(\Phi^{-1}(P_{0g}), \Phi^{-1}(P_{1g})\big)\bigg),
            \end{aligned}
        \end{equation*}
        where $P_g \equiv (P_{0g}, P_{1g})$, $P_{ig}$, $i \in \{0, 1\}$, is the function of the first stage parameter $\theta_i$ and the variable $W_g$, $\theta \equiv (\theta_0', \theta_1')'$, and $D_g \equiv (D_{0g}, D_{1g})$. The log-likelihood needs to satisfy
        \begin{enumerate}
            \item [(i)] $\mathbb{E}\big[\sup_{\rho \in [-\varepsilon, \varepsilon]} |l(\rho, \theta; D_{g}, W_{g})|\big] < \infty$.
            \item [(ii)] There exists a function $L(\cdot)$ with $|L(D_g)| < \infty$ almost surely such that for all $d \in \{0, 1\}^2$, $(p, p') \in (0, 1)^2$, and $\rho \in [-\varepsilon, \varepsilon]$, $|\tilde{l}(\rho; d, p) - \tilde{l}(\rho; d, p')| \leq L(d) ||p - p'||$, where $\tilde{l}(\rho; d, p)$ is defined as the second stage log-likelihood given $(P_{0g}, P_{1g}) = p$.
            \item [(iii)] $\mathbb{E}\big[\partial^2 l(\rho, \theta; D_{g}, W_{g}) / \partial \rho^2\big]$ is bounded away from zero.
            \item [(iv)] $\mathbb{E}\big[\sup_{\rho \in [-\varepsilon, \varepsilon]}|\partial^2 l(\rho; D_g, P_g) / \partial \rho^2|\big] < \infty$ and $\mathbb{E}\big[\sup_{\theta \in \Theta_0 \times \Theta_1}||\nabla_{\theta} \frac{\partial l(\rho, \theta; D_{g}, W_{g})}{\partial \rho} ||\big] < \infty$. 
        \end{enumerate}
    \end{enumerate}
\end{assumption}

These conditions allow us to establish the consistency of the second-stage estimator of $\rho$.

\begin{lemma}
    Suppose Assumptions \ref{as:est_assumption}, \ref{as:para_assump}, \ref{as:para_consist1}, and \ref{as:para_consist2} hold. Then, $\hat{\rho} \asto \rho_0$ as $G \rightarrow \infty$.
\end{lemma}

\begin{proof}
    See Appendix \ref{app:proof_parametric_second}.
\end{proof}

Finally, the consistency of the estimated coefficient vector $(\hat{\alpha}_{i 11}', \hat{\beta}_{i 11}')'$ is established, which in turn ensures consistency of the estimated marginal controlled effects.

\begin{theorem} \label{thm:para_consistency}
    (Consistency of parametric MCDEs and MCSEs) Suppose that Assumptions \ref{as:est_assumption}, \ref{as:para_assump}, \ref{as:para_consist1}, and \ref{as:para_consist2} hold. Also, assume that
    \begin{enumerate}
        \item $X'_{Pdd'}X_{Pdd'}$ and $\mathbb{E}[X'_{Pdd'}X_{Pdd'}]$ are nonsingular, where $X_{Pdd'}$ is defined as a $G \times K$ matrix with the $g$-th row as 
    \begin{equation*}
        X_{Pdd'_g} \equiv \big[\Phi_{\rho}(P_{0 g}, P_{1 g}), I_{dd'}^1(P_{0 g}, P_{1 g}, \rho), I_{dd'}^2(P_{0 g}, P_{1 g}, \rho),  I_{dd'}^3(P_{0 g}, P_{1 g}, \rho), \widetilde{X}_g' \big],
    \end{equation*}
    where $P_{ig}$, $i \in \{0, 1\}$, is the function of the first stage parameter $\theta_i$ and the variable $W_g$.
    \item $\|\widehat{X}_{Pdd'} - X_{Pdd'}\|_F^2/G \asto 0$, where $\widehat{X}_{Pdd'}$ is obtained by replacing the true values $P_{0g}$, $P_{1g}$, and $\rho$ in $X_{Pdd'}$ with their estimates $\widehat{P}_{0g}$, $\widehat{P}_{1g}$, and $\hat{\rho}$, respectively. The notation $\|\cdot\|_F$ denotes the Frobenius norm.
    \item Set $\varepsilon_{idd'} = Y_{ig} \mathbbm{1} \big\{D_{0g} = d \big\} \mathbbm{1} \big\{D_{1g} = d' \big\} - X_{Pdd'_g} (\alpha'_{idd'}, \beta'_{idd'})'$. Then, $\text{Var}(\varepsilon_{idd'}) = \sigma^2_{idd'} < \infty$.
    \item Define $\psi_{idd}(\alpha_{idd'}, \beta_{idd'}, \rho, \theta;Y_{ig}, D_{0g}, D_{1g}, W_g) = (Y_{id d'g} - X_{P_g} (\alpha'_{id d'}, \beta'_{idd'})')X_{P_g}$, where $Y_{i d d' g} \equiv Y_{i g} \mathbbm{1}\{D_{0 g}=d, D_{1 g}=d'\}$. It satisfies
    \begin{enumerate}
        \item [(i)] $\mathbb{E}[\sup_{\rho \in [-\varepsilon, \varepsilon]} ||\nabla_\rho \psi_{idd}(\alpha_{idd'}, \beta_{idd'}, \rho, \theta;Y_{ig}, D_{0g}, D_{1g}, W_g)||] < \infty$.
        \item [(ii)] $\mathbb{E}[\sup_{\theta \in \Theta_0 \times \Theta_1} ||\nabla_\theta \psi_{idd}(\alpha_{idd'}, \beta_{idd'}, \rho, \theta;Y_{ig}, D_{0g}, D_{1g}, W_g)||] < \infty$.
    \end{enumerate}
    \end{enumerate} 
    Under the above conditions, $(\hat{\alpha}_{i dd'}', \hat{\beta}_{i dd'}')' \asto (\alpha_{i dd'}', \beta_{i dd'}')'$ as $G \rightarrow \infty$. 
\end{theorem}

\begin{proof}
    See Appendix \ref{app:proof_parametric_third}.
\end{proof} 

Inference regarding the estimator $(\hat{\alpha}_{i dd'}', \hat{\beta}_{i dd'}')'$ is performed using standard nonparametric bootstrap methodologies.  Based on Assumptions \ref{as:para_consist1} and \ref{as:para_consist2}, the conditions in Theorem \ref{thm:para_consistency}, and standard regularity conditions, the bootstrap distribution converges uniformly to the sampling distribution of the estimator\footnote{The regularity conditions include stochastic equicontinuity and a quadratic remainder condition. A formal proof is left for future work.} \citep{romano2012uniform}. Therefore, standard nonparametric bootstrap methods, such as resampling the data and recomputing all stages, are expected to yield valid inference.

After establishing the consistency and the asymptotic distibutions of $(\hat{\alpha}_{i dd'}', \hat{\beta}_{i dd'}')'$, the consistency and asymptotic distributions of $\widehat{\text{MCSE}}_i(\mathbf{x}, d; p_0, p_1)$ and $\widehat{\text{MCDE}}_i(\mathbf{x}, d; p_0, p_1)$ follow directly from the continuous mapping theorem. This result obtains by the continuity of $\widehat{\text{MCSE}}_i$ and $\widehat{\text{MCDE}}_i$ as functions of $(\hat{\alpha}_{i dd'}', \hat{\beta}_{i dd'}')'$.

\section{Simulation and Application} \label{sec:simulation_application}

\subsection{Parametric Simulation} \label{sec:para_simulation}
This section presents a Monte Carlo simulation to assess the validity of the proposed parametric estimation methods.

For each Monte Carlo replication, I generate $G$ i.i.d. groups, where each group $g$ consists of two members indexed by $i \in \{0, 1\}$. I draw the group instrument vector, $Z_{g} = (Z_{0g}, Z_{1g})$, i.i.d. from a bivariate normal distribution $N(0, \Sigma_Z)$ with $\Sigma_Z = (1, 0.1; 0.1, 1)$. The correlation of $Z_{0g}$ and $Z_{1g}$ is not zero, since I allow the instruments of group members to be correlated. I also the group-level unobserved heterogeneity vector, $(\widetilde{V}_{0g}, \widetilde{V}_{1g})$, i.i.d. from a bivariate normal distribution $N(0, \Sigma_V)$ with $\Sigma_V = (1, 0.2; 0.2, 1)$ and independent of the instrument vector $Z_g$. By construction, $\widetilde{V}_{ig}$, $i \in \{0, 1\}$, follows a standard normal distribution. Additionally, the copula linking the normalized unobserved heterogeneity $V_{0g}$ and $V_{ig}$, where $\widetilde{V}_{ig} = \Phi(\widetilde{V}_i)$, is a Gaussian copula with correlation $\rho = 0.2$. These specifications are consistent with Assumption \ref{as:Vdist} and the first two conditions in Assumption \ref{as:para_assump}. 

I construct the following model to generate individual's treatment and potential outcome.
\begin{equation*}
    \left\{\begin{aligned}
&D_{0g}= \mathbbm{1}\big\{\widetilde{V}_{0g} \leq Z_{0g} + 0.5 Z_{1g} \big\}\\
&D_{1g} = \mathbbm{1} \big\{\widetilde{V}_{1g} \leq Z_{1g} - 0.5 Z_{0g}\} \\
&Y_{ig}(1, 1) = 1 + 0.5 U_g + 2 \widetilde{V}_{ig} + \widetilde{V}_{(1-i)g} - \widetilde{V}_{ig} \widetilde{V}_{(1-i)g}, i = 0,1 \\
& Y_{ig}(1, 0) = 3 + 0.5 U_g + 2 \widetilde{V}_{ig} + \widetilde{V}_{(1-i)g} - \widetilde{V}_{ig} \widetilde{V}_{(1-i)g}, i = 0,1 \\
& Y_{ig}(0, 1) = 3 + 0.5 U_g + 2 \widetilde{V}_{ig} - \widetilde{V}_{ig} \widetilde{V}_{(1-i)g}, i = 0,1 \\
& Y_{ig}(0, 0) = 2 + 0.5 U_g + 2 \widetilde{V}_{ig} - \widetilde{V}_{ig} \widetilde{V}_{(1-i)g}, i = 0,1,
\end{aligned}\right.
\end{equation*}
where the group-level disturbance $U_g \in \mathbb{R}$ is generated i.i.d. from the uniform distribution $\mathcal{U}(0, 1)$ and is independent of $(Z_{0g}, Z_{1g}, \tilde{V}_{0g}, \tilde{V}_{1g})$. The observed individual outcome $Y_{ig}$ is derived from 
\begin{equation*}
    \begin{aligned}
Y_{ig}=& \big[Y_{ig}(1,1) D_{(1-i)g}+Y_{ig}(1,0)(1-D_{(1-i)g})\big] D_{ig} \\
&+ \big[Y_{ig}(0,1) D_{(1-i)g}+Y_{ig}(0,0)(1-D_{(1-i)g})\big](1-D_{ig}), i = 0, 1.
\end{aligned}
\end{equation*}

In our data generating process, the instrument vector $Z_g$ is independent of the unobserved heterogeneities and potential outcomes, $(\widetilde{V}_{ig}, U_g)_{i, d, d' \in \{0,1\}}$, satisfying Assumption \ref{as:RA}. Moreover, $Z_g$ does not directly affect the outcome $Y_{ig}(d, d')$, in accordance with Assumption \ref{as:er}. The threshold function $h_i(\cdot)$ in the treatment assignment equation is specified as a first-order polynomial in the instrument, satisfying the first condition in Assumption \ref{as:para_assump}. For the potential outcomes, their conditional means given $V_{0g}$ and $V_{1g}$ satisfy the third condition in Assumption \ref{as:para_assump}. Therefore, the data generating process satisfies all identification and parametric assumptions.

I apply the method in Section \ref{sec:para_procedure} to estimate and construct $95\%$ confidence intervals for the marginal controlled spillover (MCSE) and direct effects (MCDE) at selected evaluated points $(p_0, p_1)$. In the final step of computation, directly evaluating the integrals $I^j_{dd'}$, $j = 0, 1, 2, 3, 4$, at each estimated $(\widehat{P}_{0g}, \widehat{P}_{1g})$ is analytically intractable. To address this, I approximate the integrals using numerical integration. Specifically, I employ the Gauss-Hermite quadrature method, which I have verified to be both accurate and computationally efficient.

I arbitrarily select the following evaluation points,
\begin{equation*}
    (p_0, p_1) = (0.3, 0.7), (0.4, 0.6), (0.5, 0.5), (0.6, 0.4), (0.7, 0.3),
\end{equation*}
for which the true MCSEs and MCDEs can be readily computed. I conduct $500$ Monte Carlo replications for each of four sample sizes, $G = 1000, 3000, 5000, 10000$. Table \ref{tab:para_sim} reports the coverage rates for the MCSEs, MCDEs, and the correlation parameter $\rho$.

\begin{table}[htbp]
  \centering
  \caption{Coverage Rate of 95\% Confidence Intervals for Parametric Estimators}
    \begin{tabular}{lrrrrrr}
    \toprule
          & \multicolumn{6}{c}{Coverage rate} \\
\cmidrule{2-7}          & \multicolumn{1}{c}{(0.3,0.7)} & \multicolumn{1}{c}{(0.4,0.6)} & \multicolumn{1}{c}{(0.5,0.5)} & \multicolumn{1}{c}{(0.6,0.4)} & \multicolumn{1}{c}{(0.7,0.3)} & \multicolumn{1}{c}{$\rho$} \\
    \midrule
    \multicolumn{7}{c}{Panel A1: MCDE ($G=1000$)} \\
    $d = 1$ & 0.95  & 0.952 & 0.964 & 0.972 & 0.958 & \multirow{2}[1]{*}{0.948} \\
    $d = 0$ & 0.96  & 0.958 & 0.962 & 0.96  & 0.958 &  \\
    \midrule
    \multicolumn{7}{c}{Panel A2: MCSE ($G=1000$)} \\
    $d = 1$ & 0.958 & 0.962 & 0.972 & 0.966 & 0.96  & \multirow{2}[1]{*}{0.948} \\
    $d = 0$ & 0.964 & 0.968 & 0.956 & 0.954 & 0.942 &  \\
    \midrule
    \multicolumn{7}{c}{Panel B1: MCDE ($G=3000$)} \\
    $d = 1$ & 0.952 & 0.944 & 0.946 & 0.946 & 0.958 & \multirow{2}[1]{*}{0.942} \\
    $d = 0$ & 0.952 & 0.944 & 0.946 & 0.936 & 0.936 &  \\
    \midrule
    \multicolumn{7}{c}{Panel B2: MCSE ($G=3000$)} \\
    $d = 1$ & 0.952 & 0.948 & 0.948 & 0.94  & 0.942 & \multirow{2}[1]{*}{0.942} \\
    $d = 0$ & 0.938 & 0.93  & 0.932 & 0.944 & 0.932 &  \\
    \midrule
    \multicolumn{7}{c}{Panel C1: MCDE ($G=5000$)} \\
    $d = 1$ & 0.942 & 0.95  & 0.95  & 0.954 & 0.932 & \multirow{2}[1]{*}{0.94} \\
    $d = 0$ & 0.952 & 0.94  & 0.928 & 0.918 & 0.936 &  \\
    \midrule
    \multicolumn{7}{c}{Panel C2: MCSE ($G=5000$)} \\
    $d = 1$ & 0.944 & 0.948 & 0.922 & 0.924 & 0.926 & \multirow{2}[1]{*}{0.94} \\
    $d = 0$ & 0.958 & 0.956 & 0.938 & 0.946 & 0.966 &  \\
    \midrule
    \multicolumn{7}{c}{Panel D1: MCDE ($G=10000$)} \\
    $d = 1$ & 0.932 & 0.938 & 0.948 & 0.954 & 0.944 & \multirow{2}[1]{*}{0.94} \\
    $d = 0$ & 0.926 & 0.936 & 0.948 & 0.946 & 0.958 &  \\
    \midrule
    \multicolumn{7}{c}{Panel D2: MCSE ($G=10000$)} \\
    $d = 1 $& 0.938 & 0.95  & 0.946 & 0.94  & 0.954 & \multirow{2}[1]{*}{0.94} \\
    $d = 0$ & 0.95  & 0.952 & 0.95  & 0.954 & 0.95  &  \\
    \bottomrule
    \end{tabular}%
    \footnotesize

    \vspace{0.5em}
    \textbf{Note:} The reported Monte Carlo coverage results are based on 500 replications, with the number of groups set to G = 1000, 3000, 5000, 10000.
  \label{tab:para_sim}%
\end{table}%

For the MCSEs and MCDEs, when the sample size is $G = 1000$, the coverage rates are already close to, but slightly above, 95\% for most parameters. As the sample size increases to $G = 3000$, the coverage rates decrease slightly yet remain close to 95\%, with a few parameters falling just below this threshold. For larger sample sizes, the coverage rates for all parameters stabilize around 95\%. The coverage rates for $\rho$ are also close to 95\% across all sample sizes. These simulation results support the validity of our identification strategy and parametric estimation methods.

\subsection{Application: Returns to education in best-friend relationships} \label{subsec:para_application}

In the empirical analysis, I estimate the direct and spillover effects of returns to education among the best-friend groups. I use data from the National Longitudinal Study of Adolescent to Adult Health (Add Health)\footnote{This research uses data from Add Health, funded by grant P01 HD31921 (Harris) from the Eunice Kennedy Shriver National Institute of Child Health and Human Development (NICHD), with cooperative funding from 23 other federal agencies and foundations. Add Health is currently directed by Robert A. Hummer and funded by the National Institute on Aging cooperative agreements U01 AG071448 (Hummer) and U01AG071450 (Hummer and Aiello) at the University of North Carolina at Chapel Hill. Add Health was designed by J. Richard Udry, Peter S. Bearman, and Kathleen Mullan Harris at the University of North Carolina at Chapel Hill. No direct support was received from grant P01 HD31921 for this analysis.}, a nationally representative longitudinal survey that follows a cohort of U.S. adolescents from grades 7-12 (1994-95 school year) into adulthood. The dataset contains rich information on respondents' family background and detailed friendship networks during adolescence, as well as education attainment and income in adulthood. This unique combination of longitudinal social, demographic, and economic data makes Add Health well suited for studying the long-term effects of adolescent friendships.

The Add Health dataset collects detailed friendship information during adolescence in both the in-home and in-school components of the Wave I survey. In each component, respondents are asked to list up to five male and five female friends, ranked from best to fifth best. I construct best-friend groups, each consisting of two respondents, by matching individuals who mutually nominate each other as their best friend. Following \citet{card2013peer}, I first identify best-friend pairs from the Wave I in-home interviews. I then match any remaining mutually nominated best-friend pairs from the Wave I in-school interviews. Because respondents can nominate the best friend of each gender, I prioritize opposite-gender pairs: if a respondent appears in two different best-friend groups, I retain the group consisting of opposite-gender best friends. 
%{\color{blue} (Find the evidence to justify the selection of opposite-gender best friends.)} 

The relationship between an individual's own education and their income has been extensively studied in the economics literature. In contrast, relatively little attention has been paid to how a best friend's education attainment influences an individual's earnings. Such an effect may operate through two competing channels.  

In this empirical study, I investigate the effect of a best friend's education attainment on an individual's earnings and assess which channel, information sharing or competition, plays the dominant role within best-friend networks. Importantly, our identification framework assumes that spillover effects occur only within the same network and do not extend across different networks. In the context of returns to education, this implies that any effect of another person's education is restricted to the identified best friend, with no cross-pair spillovers. I take the total personal yearly pre-tax income from the Wave III in-home survey and apply a natural logarithm transformation to construct the outcome variable $Y$. The binary treatment variable $D$ is set to 1 if the individual has completed at least 16 years of education and 0 otherwise. I include the age, gender, race, health status, and family income as the controlled covariates $X$. I assume that, conditional on the observed covariates, the coefficients $(\alpha'_{idd'}, \beta'_{idd'})'$ in the potential outcome equations are identical for individuals $i \in \{0, 1\}$ within the same group. 

For the continuous instruments $Z$, I construct measures based on the average parental education level of the individual's non-best friends, defined as all listed friends who are not ranked as the best friend. The average parental education of non-best friends may influence an individual's education attainment through channels such as shaping aspirations, fostering self-confidence, or behavioral  sharing \citep{cools2019girls}. Furthermore, conditional on covariates capturing demographic and socioeconomic characteristics, the family background of non-best friends is plausibly independent of the individual's unobserved heterogeneity. This is because weaker social ties, such as those with non-best friends, are less likely to exhibit the strong peer spillovers characteristic of best-friend relationships, and any residual correlation in unobservables is unlikely to persist once observed similarities are controlled for. Therefore, Assumption \ref{as:RA} is likely to hold in this context. 

The average parental education level of non-best friends during adolescence is unlikely to have a direct effect on an individual's yearly income in adulthood. This is because weaker social ties, such as those with non-best friends, generally lack the sustained and intensive interactions needed to shape long-term labor market outcomes. Unlike best friends, non-best friends are less likely to share close personal networks, exchange detailed career information, or provide direct referrals in the labor market. Moreover, by adulthood, many of these weaker ties from adolescence are no longer active, further limiting the scope for any direct influence on earnings. Therefore, any effect of non-best friends' parental education on the individual's income is likely to operate indirectly through its influence on the individual's own education attainment, rather than through direct channels. Hence, Assumption \ref{as:er} is plausibly satisfied in this setting.

I also require that the education decisions of best friends do not directly affect one another. This assumption is plausible because, while best friends may share aspirations or study habits, the final decision on how many years of education to pursue is typically determined by individual specific factors, such as academic ability, that are not directly changed by the best friend's decision. Therefore, any influence between best friends' education outcomes is more likely to operate indirectly through shared environments or information exchange, which aligns with the simultaneous incomplete information framework underlying our setting, rather than through direct strategic interaction in determining each other's years of schooling.

After excluding best-friend pairs in which both members have missing values for the treatment $D$ or the instrument $Z$, the sample comprises 1,019 best-friend pairs. Given the limited sample size, the parametric framework outlined in Section~\ref{sec:para_procedure} is applied under the parametric conditions specified in Assumption~\ref{as:para_assump}. The estimated correlation between best friends' unobservables $V_{0g}$ and $V_{1g}$ is 0.36, indicating a positive dependence structure among unobservables within best-friend networks. 

\begin{figure}[!htt]
        \centering
        \caption{Estimation and Confidence Bands for MCDEs and MCSEs of Returns to Education}
        \includegraphics[width = \textwidth]{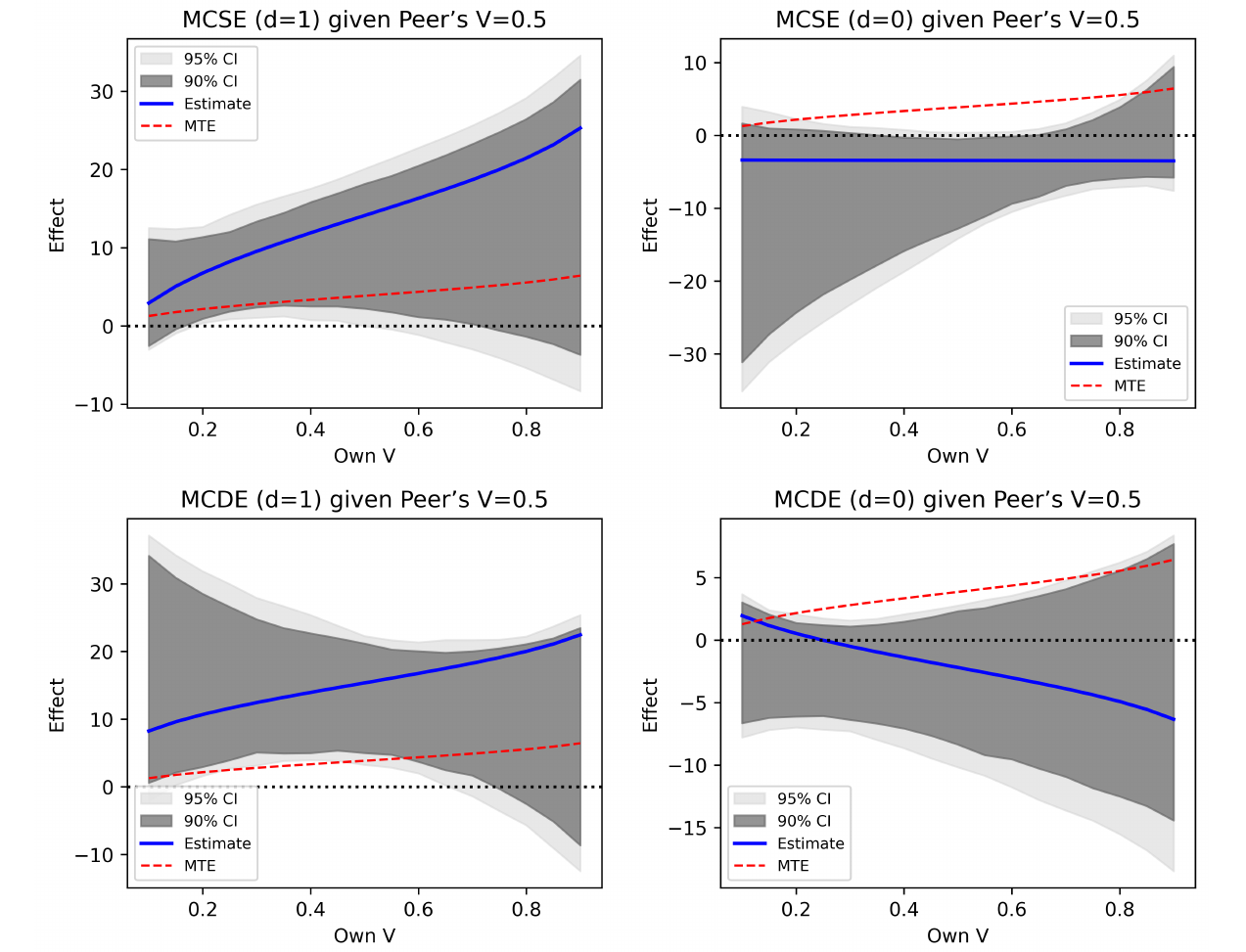}
        \label{fig:application_mcse_msde}

        \footnotesize

    \vspace{0.5em}
    \textbf{Note:} The blue solid line shows the point estimates of the MCSEs and MCDEs, while the light and dark gray shaded regions represent the 95\% and 90\% confidence intervals. The red dotted line indicates the estimated parametric standard MTEs.
\end{figure}

Figure \ref{fig:application_mcse_msde} plots the point estimates and the 90\% and 95\% confidence intervals of the marginal controlled spillover and direct effects by conditioning on the peer's unobservable $V_{-i} = 0.5$ and varing the value of individual's own unobservable. The covariates are fixed at their sample means. The blue solid line depicts the point estimates of the MCSEs and MCDEs, while the light and dark gray shaded areas represent the 95\% and 90\% confidence intervals, respectively. The red dotted line corresponds to the estimated parametric standard MTEs, which deviate substantially from the estimated MCSEs and MCDEs and lie outside their confidence intervals in most cases. This divergence provides empirical evidence of spillover effects between best friends, indicating that the standard MTE framework fails to retain a causal interpretation in the presence of such spillovers.

The results reveal substantial heterogeneity across these parameters. In particular, the estimates of the MCDEs with $d = 1$, which capture the direct effect of completing at least 16-year education given the best friend has completed at least 16 years, are positive and statistically significant at the 5\% level across most values of the individual unobservable $V_i$. However, the estimates of MCDEs with $d = 0$, which measure the direct effect of completing at least 16 years of education given the best friend has not completed this level, are not statistically significant, even at the 10\% level, across all values of the individual unobservable $V_i$. This discrepancy may reflect complementarities in human capital accumulation within best-friend pairs, consistent with the first channel discussed earlier: a highly educated best friend can provide valuable labor market information and opportunities that enhance the returns to one's own education. When both friends attain higher education, they may reinforce each other's labor market prospects through stronger professional networks, mutual encouragement in career development, or joint access to high-return opportunities. In contrast, when the best friend has lower education attainment, such reinforcing mechanisms may be absent, weakening the direct effect of one's own education on earnings.

Figure \ref{fig:application_mcse_msde} also presents the estimated MCSEs along with their confidence intervals. The MCSEs with $d = 1$, which capture the spillover effect of the best friend completing at least 16 years of education given the individual has completed 16 years, are significantly positive at the 5\% level for some values of the individual unobservable $V_i$. In contrast, the MCSEs with $d = 0$, which measure the spillover effect of the best friend completing at least 16 years of education given the individual has not completed 16 years, are even significantly negative at the 10\% level when $V_i$ is around 0.5 (approximately the value of the peer's unobservable $V_{-i}$), suggesting potential adverse spillover effects for some individuals. These patterns are consistent with the two channels through which a best friend's education attainment may affect an individual's earnings. The findings suggest that the information and opportunity channel dominates the competition channel when the individual is also highly educated, leading to positive and significant spillover effects. Conversely, when the individual has not completed 16 years of education, the competition channel appears to dominate, particularly among pairs with similar values of unobserved heterogeneity, resulting in negative estimated spillover effects. This asymmetry suggests complementarities in human capital and opportunity sharing among equally educated peers, and the potential for relative disadvantage when education attainment differs within a best-friend pair.

\section{Extensions} \label{sec:extension}

The baseline framework can be generalized to accommodate various settings in which spillovers occur within predefined groups.
First, point identification of the marginal controlled spillover and direct effects is established when outcomes depend on an exposure mapping function, rather than the full vector of group members' treatment statuses.
%Second, we develop partial identification methods for cases in which researchers only observe discrete exogenous characteristics that can serve as instruments, rather than continuously distributed ones. 
Subsequently, Appendix~\ref{sec:cts_treatment} extends the analysis to environments with continuous endogenous treatments, demonstrating that the marginal controlled spillover and direct effects remain point identified in such cases.

\subsection{Exposure to functions of peers' treatments}

\subsubsection{Setting}
In many applications, the predetermined groups within which spillovers occur can be large or vary in size. For example, when groups are defined at the level of schools, villages, or communities. In such cases, modeling outcomes as a function of the entire vector of group members' treatments may become infeasible. To address this issue, I instead adopt a framework in which the outcome of unit $i$ in group $g$, denoted $Y_{ig}$, depends on the unit's own treatment $D_{ig}$ and on a known function of the full vector of group treatments, denoted $H_g$. This function $H_g$ summarizes the group's effective treatment, consistent with the notion of an effective treatment in \citet{manski2013identification} and the exposure mapping framework of \citet{aronow2017estimating}. By reducing the dimensionality of peer treatments to an interpretable exposure measure, this approach allows for the analysis of spillovers in large or heterogeneous groups while maintaining tractable identification and interpretation.

I consider a sample of $G$ independent and identically distributed groups, indexed by $g = 1, \cdots, G$, where spillovers are restricted to occur within groups and not across them. Unlike the baseline framework, each group now consists of $n_g$ members, where the group size $n_g$ is allowed to vary across groups. To capture peer effects in this heterogeneous group size setting, I assume that the outcome of interest depends not on the full treatment vector but rather on a group-level exposure mapping, $H_g: \boldsymbol{D}_g \mapsto \mathbb{R}$, where $\boldsymbol{D}_g$ denotes the vector of individual treatment assignments within group $g$. This mapping $H_g$ is assumed to be continuous and correctly specified by the researcher. A common and tractable specification is the proportion of treated individuals in the group, given by $H_g = \sum_{i=1}^{n_g} D_{ig} / n_g$. Because treatment assignments $\boldsymbol{D}_g$ are observed, researchers can directly recover the realized values of $H_g$ for each group. %{\color{red} Add an example.}

I specify the outcome for individual $i$ in group $g$ as $Y_{ig} = Y_{ig}(D_{ig}, H_g, U_{ig}, U_{-i,g})$, so that outcomes may depend on the individual's own treatment $D_{ig}$, the continuous group-level exposure $H_g$, and both the individual's unobserved characteristics $U_{ig}$ and those of her peers $U_{-i,g}$. I let $Y_{ig}(d, h)$ represents the potential outcome for unit $i$ given $D_{ig} = d$ and $H_g = h$. 

%Additionally, each unit $i$ is randomly assigned an instrument $Z_{ig} \in \mathbb{R}^l$, and there is an instrument $Z_g \in \mathbb{R}^k$ randomly assigned at the group level.

I formulate the following equations to model spillovers that operate through the group-level exposure $H_g$. Throughout this section, I retain the subscript $g$ to distinguish individual level variables (indexed by $ig$) from group level variables (indexed by $g$), thereby clarifying how exposure-driven spillovers enter the model.

\begin{equation} \label{eq:model_mixed_trt}
    \left\{\begin{array}{l}
        Y_{ig} =  Y_{i g} (1, H_g) D_{i g} + Y_{i g} (0, H_g)(1 - D_{i g}) \\
        D_{i g} = \mathbbm{1}\left\{V_{i g} \leq h_i(Z_{ig}, Z_{-ig})\right\} \\
        H_g = m(Z_g, \varepsilon_g)
        \end{array}\right.
\end{equation}

In Equation \eqref{eq:model_mixed_trt}, I formulate the outcome equation within the classical potential outcomes framework, specifying that each individual's outcome depends on her own binary treatment status, $D_{ig} \in \{0, 1\}$, as well as the group-level exposure $H_g$. A key feature of our framework is the recognition that both the individual treatment $D_{ig}$ and the group exposure $H_g$ may be endogenous. Specifically, $D_{ig}$ may correlate with unobserved individual-level characteristics that also influence the outcome. Likewise, the group exposure $H_g$, which is defined as a function of all group members' treatments, may depend on group-level unobservables that also affect the individual outcome.

To address the endogeneity of the individual treatment $D_{ig}$, I model it using a single-threshold crossing rule, analogous to the specification in the basic setting. Specifically, individual $i$ selects into treatment if the unobserved characteristic $V_{ig}$ falls below a threshold $h_i(Z_{ig}, Z_{-ig})$. The threshold function depends on the vector of instruments assigned to individual $i$, $Z_{ig}$, or additionally on the instruments assigned to other group members, $Z_{-ig}$. As before, I do not impose any functional form restrictions on the threshold function $h_i$ to preserve flexibility in how instruments affect treatment selection. In addition, the subscript $i$ allows for heterogeneity in threshold functions across individuals within the same group.

To account for the potential endogeneity of the group-level exposure $H_g$, I introduce a group instrument $Z_g$ and assume that $H_g$ follows a reduced-form relationship given by $H_g = m(Z_g, \varepsilon_g)$, where $\varepsilon_g \in \mathbb{R}$ represents an unobserved group-specific characteristic. The group instrument $Z_g$ may take various forms. For instance, it may correspond to the full vector of individual instruments $(Z_{ig})_{i \in \{1, \cdots, n_g\}}$, or to an aggregate statistic such as the average instrument level within the group. The random variable $\varepsilon_g$ captures latent group-level heterogeneity, potentially containing factors such as the group's social cohesion or the dependence structure among individual-level unobservables $(V_{ig})_{i \in \{1, \cdots, n_g\}}$. I impose no functional form restrictions on $m(\cdot)$ to maintain flexibility in the modeling of group exposure. In addition, I do not restrict the dependence structure between the individual unobservable $V_{ig}$ and the group-level unobservable $\varepsilon_g$, allowing for arbitrary correlation between individual- and group-level latent factors.

\begin{remark}(Reduced function of $H_g$) \label{remark:reduced_Hg}
    To explain the reduced-form function of $H_g$, consider a scenario where the exposure function $H_g$ is defined as the average treatment level within group $g$, $H_g = \sum_{i=1}^{n_g} D_{ig} / n_g$, where $n_g$ denotes the number of members in group $g$, which may vary across groups. Let $\mathcal{I}_g$ represent the set of indices for individuals in group $g$. Assume that the group comprises two types of individuals:
    \begin{enumerate}
        \item Type 1: Individuals indexed by $i \in \mathcal{I}_g^1 \subseteq \mathcal{I}_g$, which have unobserved individual unobservable $V_{ig} = \varepsilon_g$, $\varepsilon_g \in (0, 1)$.
        \item Type 2: Individuals indexed by $j \in \mathcal{I}_g^2 = \mathcal{I}_g \setminus \mathcal{I}_g^1$, with individual unobservable $V_{jg} = 1 - \varepsilon_g$.
    \end{enumerate}
    Here, $\varepsilon_g \in (0, 1)$ captures unobserved heterogeneity at the group level, influencing the individual-level unobservables for both types. Additionally, assume that Type 1 individuals constitute an $\varepsilon_g$-proportion of the group, i.e., $|\mathcal{I}_g^1| / |\mathcal{I}_g| = \varepsilon_g$. Furthermore, suppose that individual treatment decisions depend on a group-level instrument  $Z_g$, $D_{ig} = \mathbbm{1}\{V_{ig} \leq h(Z_g)\}$.
    
    Under these assumptions, the exposure $H_g$ can be expressed as an explicit function of $\varepsilon_g$ and $Z_g$,
    \begin{equation*}
        \begin{aligned}
            H_g =& \frac{1}{n_g}\sum_{i=1}^{n_g} D_{ig} = \frac{1}{n_g} \sum_{i \in \mathcal{I}_g^1} D_{ig} + \frac{1}{n_g} \sum_{i \in \mathcal{I}_g^2} D_{ig} \\
            =& \varepsilon_g \mathbbm{1}\{\varepsilon_g \leq h(Z_g)\} + (1-\varepsilon_g) \mathbbm{1}\{1-\varepsilon_g \leq h(Z_g)\}.
        \end{aligned}
    \end{equation*}
    In this framework, $\varepsilon_g$ not only reflects the proportion of each individual type within the group but also captures the unobserved heterogeneity among different types of group members. More generally, the exposure level $H_g$ can be represented as an unknown reduced-form function of the group-level unobservable $\varepsilon_g$ and the instrument $Z_g$, where $\varepsilon_g$ can be interpreted as a scalar latent variable that fully summarizes the group-level unobserved heterogeneity relevant for determining exposure. Formally, I write $H_g = m(Z_g, \varepsilon_g)$, with $m(\cdot)$ left unspecified. One convenient interpretation is to view $m(z,e)$ as the quantile function of the conditional distribution of exposure, $Q_{H_g \mid Z_g=z}(e)$, and to define $\varepsilon_g$ as the corresponding conditional cumulative distribution function, $\varepsilon_g = F_{H_g \mid Z_g}(H_g)$. 
\end{remark}

\begin{remark}(Random saturation framework)
    Our exposure-mapping framework can accommodate randomized saturation designs, such as those studied in \citet{ditraglia2023identifying}, where each group is randomly assigned a saturation level—defined as the proportion of individuals offered treatment within the group. In this context, the group-level instrument $Z_g$ corresponds to the randomized saturation assignment for group $g$, while the group exposure $H_g$ reflects the average treatment take-up within the group. Unlike their approach, which requires one-sided noncompliance and an individualized offer response assumption, our methodology does not impose restrictions on the compliance behavior and allow the treatment take-up to depend on other group members' treatment assignment, thereby capturing richer forms of spillovers and strategic behavior.
    
Crucially, our identification strategy is based on a reduced-form modeling of the group exposure and outcome equations. This approach eliminates the need to specify random coefficient models or to model the distribution of compliance types explicitly. Within our framework, the group-level unobservable $\varepsilon_g$ can be viewed as a scalar proxy for the fraction of compliers in the spirit of \cite{ditraglia2023identifying}, while more broadly capturing latent heterogeneity in group responsiveness without imposing restrictive parametric assumptions.
\end{remark}

In the exposure mapping framework, I redefine the marginal controlled spillover effect (MCSE) and the marginal controlled direct effect (MCDE) relative to the basic setting by explicitly conditioning on both the individual-specific unobservable $V_{ig}$ and the group-level unobservable $\varepsilon_g$, which accommodates heterogeneity at both the individual and group levels.

\begin{definition} (MCSE and MCDE: Exposure mapping setting) 
    Consider the model specified in Equation \eqref{eq:model_mixed_trt}.
    \begin{enumerate}
        \item Fix the treatment of unit $i$ in group $g$ to be $D_{ig} = d$, $d \in \{0, 1\}$. I define the marginal controlled spillover effect (MCSE) of changing the exposure level from $h$ to $h'$, where $h,h' \in \mathbb{R}$, conditional on the individual unobservable $V_{ig} = p_0$ and the group-level unobservable $\varepsilon_g = p_1$, as
        \begin{equation*}
            \operatorname{MCSE}_{ig}(d, h', h; p_0, p_1) \equiv \mathbb{E}\big[Y_{ig}(d, h')-Y_{ig}(d, h) \mid V_{ig}=p_0, \varepsilon_g=p_1\big].
        \end{equation*}
        \item Fix the exposure level in group $g$ as $H_g = h$, $h \in \mathbb{R}$. I define the marginal controlled direct effect (MCDE) for individual $i$, conditional on the individual unobservable $V_{ig} = p_0$ and the group-level unobservable $\varepsilon_g = p_1$, as
        \begin{equation*}
            \operatorname{MCDE}_{ig}(h; p_0, p_1) \equiv \mathbb{E}\big[Y_{ig}(1, h)-Y_{ig}(0, h) \mid V_{ig}=p_0, \varepsilon_g=p_1\big]
        \end{equation*}
    \end{enumerate}
\end{definition}

\subsubsection{Identification}
Identification in this setting is achieved under Assumptions \ref{as:RA_mixed}-\ref{as:mon_m}.

\begin{assumption}(Random assignment: Exposure mapping setting) \label{as:RA_mixed}
    I assume that the instrument vector is randomly assigned across groups, so that for any $g \in \{1, \cdots, G\}$,
    \begin{equation*}
        (Z_{ig}, Z_g)_{i \in \{1, \cdots, n_g\}} \indep \big(Y_{ig}(d, h), V_{ig}, \varepsilon_g\big)_{d \in \{0, 1\}, h \in \mathbb{R}, i \in \{1, \cdots, n_g\}}.
    \end{equation*}
    Additionally, the instruments $(Z_{ig}, Z_g)_{i \in \{1, \cdots, n_g\}}$ satisfy an exclusion restriction in that they do not directly affect the outcome $Y_{ig}$.
\end{assumption}

Assumption \ref{as:RA_mixed} requires that the vector of instruments assigned to individuals and the group must be randomly assigned at the group level, such that they are independent of all potential outcomes, as well as of both individual- and group-level unobservables. Moreover, under the model structure in Equation \eqref{eq:model_mixed_trt}, the instruments also satisfy the exclusion restriction, in the sense that they influence outcomes only through their effect on treatment take-ups and exposure, and do not directly enter the outcome equation.

\begin{assumption}(Monotonicity of $m$) \label{as:mon_m}
    Given the instrument values $Z_g = z \in \mathbb{R}^k$, the function $m(z, e)$ is continuous and strictly monotonic in $e$.
\end{assumption}

The monotonicity condition in Assumption \ref{as:mon_m} ensures that the group-level treatment $H_g$ is a one-to-one mapping of the group-level unobservable $\varepsilon_g$, conditional on the instruments. Specifically, for any given $Z_g = z$, the reduced-form relation $H_g \mid (Z_g = z) = m(z, \varepsilon_g)$ can be inverted with respect to $\varepsilon_g$, yielding $ \varepsilon_g \mid (Z_g = z) = m_z^{-1}(H_g) $. Thus, under the random assignment assumption \ref{as:RA_mixed}, I obtain the control function representation $\varepsilon_g = m_{Z_g}^{-1}(H_g)$. As established in \citet{goff2024testing}, if the conditional distribution $F_{H_g \mid Z_g}$ is strictly increasing and continuous, then the monotonicity condition is not an additional structural restriction but instead follows directly from the reduced-form interpretation of the exposure mapping function discussed in Remark \ref{remark:reduced_Hg}.

I define the individual-level propensity score function for unit i in group g, consistent with previous definitions, as $P_{ig}(z) \equiv \mathbb{P}(D_{ig} = 1 \mid Z_{ig} = z)$. In addition, I define the group-level propensity score function for group $g$ as $P_g(z, h) \equiv \mathbb{P}(H_g \leq h \mid Z_g = z)$. I denote by $\mathcal{P}$ the support of the joint propensity score function $(P_{ig}, P_g)$. As shown in the Appendix \ref{app:proof_exposure}, the individual-level propensity score function identifies the individual threshold function $h_i(\cdot)$, while the group-level propensity score function identifies the inverse of the exposure mapping, $m_{Z_g}^{-1}(H_g)$. Taken together, these results imply that both individual- and group-level propensity score functions can be used as control functions, allowing us to account for unobserved heterogeneity in the outcome equation and thereby achieve identification of the marginal controlled spillover and direct effects. 

\begin{theorem}(Identifying MCSEs and MCDEs: Exposure mapping setting)
    Consider the model specified in Equation \eqref{eq:model_mixed_trt}. Suppose that Assumptions \ref{as:Vdist}, \ref{as:RA_mixed}, and \ref{as:mon_m} hold. For $d \in \{0, 1\}$, $h \in \mathbb{R}$, and $(p_0, p_1) \in \mathcal{P}$, I impose the following additional conditions: (i) $\mathbb{E}[Y_{i g} D_{i g} \mid H_g=h, P_{i g}(Z_g)=p_0, P_g(Z_g, H_g)=p_1]$ and $\mathbb{P}(D_{ig} = 1 \mid H_g = h, P_{ig}(Z_g) = p_0, P_g(Z_g, H_g) = p_1)$ are differentiable with respect to $p_0$; (ii) the marginal treatment response functions 
    \begin{equation*}
        m_{ig}^{(d, h)}\left(p_0, p_1\right) \equiv \mathbb{E}\left[Y_{ig}\left(d, h\right) \mid V_{ig}=p_0, \varepsilon_g=p_1\right]
    \end{equation*}
    are continuous; and (iii) the conditional density $f_{V_{ig} \mid \epsilon_g}$ is bounded from above and away from zero. 
    
    Then, $m_{ig}^{(d, h)}(p_0, p_1)$ is identified as 
    \begin{equation*}
        \begin{aligned}
            &\frac{\partial}{\partial p_0} \mathbb{E}\left[Y_{ig} \mathbbm{1} \{D_{ig} = d \} \mid H_g = h, P_{ig}\left(Z_g\right) = p_0, P_g\left(Z_g, H_g\right) = p_1\right] \Big/ \\
            &\frac{\partial}{\partial p_0} \mathbb{P}\left(D_{ig} = d \mid H_g = h, P_{ig}(Z_g) = p_0, P_g(Z_g, H_g) = p_1\right).
        \end{aligned}
    \end{equation*}
    By taking appropriate differences of the marginal treatment response functions, I obtain identification of the marginal controlled spillover effect (MCSE) and the marginal controlled direct effect (MCDE).
\end{theorem}

\begin{proof}
    See Appendix \ref{app:proof_exposure}.
\end{proof}

\section{Conclusion}
This paper develops a general framework for identifying causal effects in environments with within-group spillovers and endogenous treatment decisions. By relaxing the Stable Unit Treatment Value Assumption (SUTVA), the framework accommodates settings in which an individual's outcome depends not only on her own treatment but also on the treatment selection of group members. It further allows each individual's treatment decision to depend on instruments assigned to other group members.

The paper introduces two classes of causal parameters, he generalized local average controlled spillover and direct effects (LACSEs and LACDEs) and the marginal controlled spillover and direct effects (MCSEs and MCDEs), which extend the standard local average and marginal treatment effect frameworks to settings with spillovers. The LACSEs and LACDEs quantify peer and own treatment effects for specific subpopulations, while the MCSEs and MCDEs capture these effects conditional on continuous values of unobserved heterogeneity within groups. The paper formally establishes general conditions for the point identification of the LACSEs and LACDEs, characterizing the instrumental variation necessary for identification regardless of whether the instruments are discrete or continuous. It also shows that the MCSEs and MCDEs are nonparametrically point identified from continuous instrumental variation without imposing functional form restrictions on the outcome equation or on the joint distribution of unobserved characteristics within the group, thereby accommodating flexible forms of spillover structures.

These results extend existing approaches to causal inference with spillovers, showing that the MCSE-MCDE framework provides a natural generalization of the standard marginal treatment effect (MTE) model to spillover settings. Furthermore, the paper establishes that these marginal controlled effects serve as building blocks for policy-relevant treatment parameters (PRTEs), enabling the evaluation of both direct and spillover effects under counterfactual policy interventions.

For estimation and inference, the paper develops a semiparametric estimation strategy that builds on \citet{carneiro2009estimating}, extending it to accommodate within-group spillovers while mitigating the curse of dimensionality associated with covariates. Asymptotic properties of the semiparametric estimators are derived, and a parametric estimation framework is proposed as a practical complement when sample sizes are limited or group sizes are large. Monte Carlo simulations demonstrate that the parametric estimators perform well in finite samples, providing accurate estimates and confidence interval coverage.

An empirical application using the National Longitudinal Study of Adolescent to Adult Health (Add Health) illustrates the framework's practical relevance. The analysis examines how education attainment affects long-term earnings within best-friend networks. The results indicate positive dependence between friends' unobserved characteristics and reveal heterogeneous spillover patterns, showing that both the magnitude and direction of peer influences vary with individuals' and their best friends' education attainment.

Finally, the paper extends the framework to settings with exposure mappings, where outcomes depend on a known function of group members' treatments rather than the full treatment vector. This generalization broadens the applicability of the framework to environments with varying or large group sizes, while preserving nonparametric point identification under continuous instrumental variation. The appendix further extends the analysis to settings with continuous endogenous treatments and establishes point identification results under continuous instruments.

Overall, the proposed framework offers an econometric foundation for identifying and estimating causal effects in the presence of within-group spillovers and endogenous treatments. It provides theoretical and practical tools for studying a wide range of social, education, and economic interactions. Future research could extend the framework by modeling endogenous group formation, and by developing more efficient semiparametric estimation methods to enhance finite-sample performance.

%%%%%%%%%%%%%%%%%%%%%%%%%%%%%%%%%%%%%%%%%%%%%%%%%%%%%%%%%%%%%%%%%%%%%%%%
%%%% Bibliography
\clearpage
\onehalfspacing
\bibliographystyle{jpe}
\renewcommand\refname{Reference}
\bibliography{Template_Bib}

%%%%%%%%%%%%%%%%%%%%%%%%%%%%%%%%%%%%%%%%%%%%%%%%%%%%%%%%%%%%%%%%%%%%%%%%
%%%% Appendix
%{\noindent\Large\textbf{Appendix}}
%\clearpage

\begin{appendices}

\section{Further Details of the Modeling Framework}

\subsection{Endogenous Effects in the Outcome} \label{app:model_endo_effect}
    Some studies (e.g., \citealp{bramoulle2009identification}) model spillover effects using a system of structural equations in which a unit's outcome may directly depend on the outcomes of their peers. When outcomes among group members influence one another, such interactions are typically referred to as endogenous effects. In contrast, this framework does not explicitly model endogenous effects, since $Y_i$ does not directly depend on the outcomes of unit $i$'s peers, $Y_{-i}$. However, this should not be interpreted as ruling out the possibility of endogenous effects within the model. As discussed in \cite{manski2013identification}, the potential outcome $Y_i(d, d')$ can be interpreted as a reduced-form solution of underlying structural models that include endogenous effects.

    For illustration, consider a system of linear structural equations for $Y_i$, where $i \in \{0, 1\}$, within group $g$. The group subscript $g$ is omitted for notational simplicity.
    \begin{equation*}
        \begin{aligned}
            & Y_0=\alpha_0+\alpha_1 D_{0}+\alpha_2 D_{1}+\alpha_3 Y_{1}+U_{0}+\gamma_1 U_{1}, \\
            & Y_1=\beta_0+\beta_1 D_{1}+\beta_2 D_{0}+\beta_3 Y_{0}+U_{1}+\gamma_2 U_{0}, \alpha_3 \beta_3 \neq 1
        \end{aligned}
    \end{equation*}
    where $D_i$ denotes unit $i$'s treatment and $U_i$ captures unobserved factors. One can solve the system to obtain a reduced-form expression for $Y_i$ as a linear function of $D_i$ and $D_{-i}$, without explicitly involving $Y_{-i}$:
    \begin{equation*}
        \begin{aligned}
            Y_{0}= & \frac{\alpha_0+\alpha_3 \beta_0+\left(\alpha_1+\alpha_3 \beta_2\right) D_{0}+\left(\alpha_2+\alpha_3 \beta_1\right) D_{1}}{1-\alpha_3 \beta_3} \\
            & +\frac{\left(1+\alpha_3 \gamma_2\right) U_{0}+\left(\gamma_1+\alpha_3\right) U_{1}}{1-\alpha_3 \beta_3}, \\
            Y_{1 g}= & \frac{\beta_0+\beta_3 \alpha_0+\left(\beta_1+\beta_3 \alpha_2\right) D_{1}+\left(\beta_2+\beta_3 \alpha_1\right) D_{0}}{1-\alpha_3 \beta_3} \\
            & +\frac{\left(1+\beta_3 \gamma_1\right) U_{1}+\left(\beta_3+\gamma_2\right) U_{0}}{1-\alpha_3 \beta_3}.
        \end{aligned}
    \end{equation*}
    Therefore, $Y_i(d, d')$ can be interpreted as solutions for $Y_i$ when the treatment assignments are set to $D_i = d$ and $D_{-i} = d'$. This interpretation remains valid when the structural functions are nonlinear. In the nonlinear case, the link between the potential outcome equations and the underlying structural models is less transparent, though it can still be derived by researchers in the context of specific applications. In this paper, however, the focus is on the reduced-form treatment effects of a unit's' own and peers' treatments, rather than on the structural parameters embedded in the structural equations.

\subsection{Simultaneous Incomplete Information Game} \label{app:incomplete_info_game}

    In this framework, unit $i$'s treatment decision does not directly depend on her peers' treatment choices $D_{-i}$, implying the absence of strategic interaction in treatment take-up. This structure is consistent with a simultaneous incomplete information game, as analyzed in \cite{aradillas2010semiparametric}. In such a framework, $V_i$ represents private information observed only by unit $i$, while $Z = (Z_i, Z_{-i})$ is a vector of public signals observed by all group members. Each unit $i$ forms a subjective belief about the joint distribution $P_i(D_i = 1, D_{-i} = 1 \mid Z)$, and makes an optimal decision accordingly. The decision rule for unit $i$ can be derived as 
    \begin{equation*}
        D_i=\mathbbm{1}\{V_i \leq \alpha_i \underbrace{P_i\big(D_{-i}=1 \mid D_i=1, Z\big)}_{\text {Unit i's belief, function of } Z}\}.
    \end{equation*}
The optimal decision function satisfies the single threshold crossing structure imposed in our setting, with unit $i$'s belief captured by the threshold function $h_i(Z)$. \cite{aradillas2010semiparametric} provides conditions for the existence and uniqueness of equilibrium beliefs. For a detailed discussion of the simultaneous-move game of incomplete information, see \cite{aradillas2010semiparametric}.
    
In contrast, other studies, such as \cite{balat2023multiple} and \cite{hoshino2023treatment}, focus on settings with direct strategic interactions between $D_i$ and $D_{-i}$. However, those models do not allow unit $i$'s treatment to be directly influenced by peer instruments $Z_{-i}$. Both incomplete information games and models with strategic interaction are empirically relevant, and their associated identification strategies can be viewed as complementary to the framework developed in this paper.

\section{Proof of Identification for LACSEs and LACDEs} \label{app:id_lacse_lacde}

\textit{Identifying the generalized LACSEs: }Suppose that two distinct pairs of propensity scores, $(p_0, p_1)$ and $(p_0, p_1')$, exist in $\mathcal{P}$ with $p_1' > p_1$. The following discussion illustrates how this variation can be exploited to identify the local average controlled spillover and direct effects described in Item 1 of Theorem \ref{thm:id_lace}.

Given $(p_0, p_1) \in \mathcal{P}$, consider the observed conditional expectation $\mathbb{E}[Y_i D_i D_{-i} \mid P_i = p_0, P_{-i} = p_1]$. This expectation identifies the average potential outcome $Y_i(1, 1)$ for the subpopulation whose unobserved characteristics lie in the region $\{V_i \leq p_0, V_{-i} \leq p_1\}$:
\begin{equation*} 
    \begin{aligned}
\mathbb{E}\big[Y_i D_i D_{-i} \mid P_i=p_0, P_{-i}=p_1\big] & =\mathbb{E}\big[Y_i(1,1) \mathbbm{1}\left\{V_i \leq p_0, V_{-i} \leq p_1\right\} \mid P_i=p_0, P_{-i}=p_1\big] \\
& =\mathbb{E}\big[Y_i(1,1) \mathbbm{1}\left\{V_i \leq p_0, V_{-i} \leq p_1\right\}\big],
\end{aligned}
\end{equation*}
where the second equality holds because the propensity scores are functions of the instrumental variables and are therefore independent of the potential outcomes and unobservables $(V_i, V_{-i})$, as implied by Assumption \ref{as:RA}. Similarly, evaluating $\mathbb{E}[Y_i D_i D_{-i} \mid \cdot, \cdot]$ at $(p_0, p_1')$ identifies the mean potential outcome $Y_i(1, 1)$ for the subpopulation with $\{V_i \leq p_0, V_{-i} \leq p_1'\}$. The difference between this and the expectation evaluated at $(p_0, p_1)$ identifies the average potential outcome $Y_i(1, 1)$ for the subpopulation characterized by $\{V_i \leq p_0, p_1 < V_{-i} \leq p_1'\}$,
\begin{equation} \label{eq:id_lace1}
    \begin{aligned}
& \mathbb{E}\big[Y_i D_i D_{-i} \mid P_i=p_0, P_{-i}=p_1'\big]-\mathbb{E}\big[Y_i D_i D_{-i} \mid P_i=p_0, P_{-i}=p_1\big] \\
= & \mathbb{E}\big[Y_i(1,1) \mathbbm{1}\{V_i \leq p_0, p_1<V_{-i} \leq p_1'\}\big], 
\end{aligned}
\end{equation}
which implies that the subpopulation characterized by $\{V_i \leq p_0, p_1 < V_{-i} \leq p_1'\}$ switches to the treatment combination $(D_i, D_{-i}) = (1, 1)$ when the peer's propensity score increases from $p_1$ to $p_1'$, holding unit $i$'s score fixed at $p_0$.

Applying the same logic to the conditional expectations $\mathbb{E}[Y_i D_i (1 - D_{-i}) \mid \cdot, \cdot]$ evaluated at $(p_0, p_1)$ and $(p_0, p_1')$ yields the negative of the average potential outcome $Y_i(1, 0)$ for the same subpopulation, 
\begin{equation} \label{eq:id_lace2}
    \begin{aligned}
& \mathbb{E}\big[Y_i D_i (1 - D_{-i}) \mid P_i=p_0, P_{-i}=p_1'\big]-\mathbb{E}\big[Y_i D_i (1 - D_{-i}) \mid P_i=p_0, P_{-i}=p_1\big] \\
= & - \mathbb{E}\big[Y_i(1,0) \mathbbm{1}\{V_i \leq p_0, p_1<V_{-i} \leq p_1'\}\big], 
\end{aligned}
\end{equation}
indicating that these units transition away from the treatment combination $(D_i, D_{-i}) = (1, 0)$ under this change in propensity scores.

Summing Equations \eqref{eq:id_lace1} and \eqref{eq:id_lace2} identifies
\begin{equation*}
    \begin{aligned}
& \mathbb{E}\big[Y_i D_i \mid P_i=p_0, P_{-i}=p_1'\big]-\mathbb{E}\big[Y_i D_i \mid P_i=p_0, P_{-i}=p_1\big] \\
= & \mathbb{E}\big[\big(Y_i(1, 1) - Y_i(1,0)\big) \mathbbm{1}\{V_i \leq p_0, p_1<V_{-i} \leq p_1'\}\big],
\end{aligned}
\end{equation*}
which is an average spillover effect for the subpopulation $\{V_i \leq p_0, p_1 < V_{-i} \leq p_1'\}$, holding unit $i$'s treatment fixed at $D_i = 1$. Additionally, the share of the subpopulation $\{V_i \leq p_0, p_1 < V_{-i} \leq p_1'\}$ can be identified from 
\begin{equation*}
    \mathbb{E}\big[D_i D_{-i} \mid P_i = p_0, P_{-i} = p_1'\big] - \mathbb{E}\big[D_i D_{-i} \mid P_i = p_0, P_{-i} = p_1\big],
\end{equation*}
allowing the local average controlled spillover effect $\operatorname{LACSE}_i^{(1)}(V_i \leq p_0, p_1 < V_{-i} \leq p_1')$ to be identified as
\begin{equation*}
    \frac{\mu_{i, i}^{(1)}\left(p_0, p_1^{\prime}\right)-\mu_{i, i}^{(1)}\left(p_0, p_1\right)}{C\left(p_0, p_1^{\prime}\right)-C\left(p_0, p_1\right)} = \operatorname{LACSE}_i^{(1)}(V_i \leq p_0, p_1 < V_{-i} \leq p_1'),
\end{equation*}
where $\mu_{i, i}^{(d)}(p_0, p_1) \equiv \mathbb{E}[Y_i \mathbbm{1}\{D_i=d\} \mid P_i=p_0, P_{-i}=p_1]$, $d \in \{0,1\}$, and $C(p_0, p_1) = \mathbb{E}[D_i D_{-i} \mid P_i = p_0, P_{-i} = p_1]$.

Replacing $D_i$ with $(1 - D_i)$ in Equations~\eqref{eq:id_lace1} and~\eqref{eq:id_lace2} and repeating the same steps yield 
\begin{equation*}
    \begin{aligned}
& \mathbb{E}\big[Y_i (1 - D_i) \mid P_i=p_0, P_{-i}=p_1'\big]-\mathbb{E}\big[Y_i (1 - D_i) \mid P_i=p_0, P_{-i}=p_1\big] \\
= & \mathbb{E}\big[\big(Y_i(0, 1) - Y_i(0,0)\big) \mathbbm{1}\{V_i > p_0, p_1<V_{-i} \leq p_1'\}\big], 
\end{aligned}
\end{equation*}
which identifies the average spillover effect for the subpopulation $\{V_i > p_0, p_1<V_{-i} \leq p_1'\}$, holding unit $i$'s treatment fixed at $D_i = 0$. Dividing by the proportion of the subpopulation, which is identified from 
\begin{equation*}
    \big(p_1' - p_1\big) - \big[C(p_0, p_1') - C(p_0, p_1)\big] = \mathbb{P}\big(V_i > p_0, p_1<V_{-i} \leq p_1'\big), 
\end{equation*}
then the local average controlled spillover effect $\operatorname{LACSE}_i^{(0)}(V_i>p_0, p_1<V_{-i} \leq p_1')$ can be identified as 
\begin{equation*}
    \frac{\mu_{i, i}^{(0)}\left(p_0, p_1^{\prime}\right)-\mu_{i, i}^{(0)}\left(p_0, p_1\right)}{\left(p_1^{\prime}-p_1\right)-\left[C\left(p_0, p_1^{\prime}\right)-C\left(p_0, p_1\right)\right]} = \operatorname{LACSE}_i^{(0)}(V_i>p_0, p_1<V_{-i} \leq p_1').
\end{equation*}

\textit{Identifying the generalized LACDEs: }Suppose there exist two pairs of propensity scores, $(p_0, p_1)$ and $(p_0', p_1)$, with $p_0' > p_0$. The difference in the observed conditional expectations $\mathbb{E}[Y_i D_i D_{-i} \mid \cdot, \cdot]$ evaluated at these two points identifies
\begin{equation} \label{eq:id_lace3}
    \begin{aligned}
& \mathbb{E}\big[Y_i D_i D_{-i} \mid P_i=p_0', P_{-i}=p_1\big]-\mathbb{E}\big[Y_i D_i D_{-i} \mid P_i=p_0, P_{-i}=p_1\big] \\
= & \mathbb{E}\big[Y_i(1,1) \mathbbm{1}\{p_0 < V_i \leq p_0', V_{-i} \leq p_1\}\big], 
\end{aligned}
\end{equation}
indicating that the subpopulation with unobserved characteristics $\{p_0 < V_i \leq p_0', V_{-i} \leq p_1\}$ transitions to the treatment configuration $(D_i, D_{-i}) = (1,1)$ when the propensity scores shift from $(p_0, p_1)$ to $(p_0', p_1)$. Similarly, the difference in conditional expectations $\mathbb{E}[Y_i (1 - D_i) D_{-i} \mid \cdot, \cdot]$ evaluated at $(p_0, p_1)$ and $(p_0', p_1)$ identifies
\begin{equation} \label{eq:id_lace4}
    \begin{aligned}
& \mathbb{E}\big[Y_i (1 - D_i) D_{-i} \mid P_i=p_0', P_{-i}=p_1\big]-\mathbb{E}\big[Y_i (1 - D_i) D_{-i} \mid P_i=p_0, P_{-i}=p_1\big] \\
= & -\mathbb{E}\big[Y_i(0,1) \mathbbm{1}\{p_0 < V_i \leq p_0', V_{-i} \leq p_1\}\big], 
\end{aligned}
\end{equation}
implying that the same subpopulation moves away from the treatment configuration $(D_i, D_{-i}) = (0,1)$ under this change in the propensity scores. 

Summing Equations \eqref{eq:id_lace3} and \eqref{eq:id_lace4} identifies 
\begin{equation*}
    \begin{aligned}
        &\mathbb{E}\big[Y_i D_{-i} \mid P_i = p_0', P_{-i} = p_1\big] - \mathbb{E}\big[Y_i D_{-i} \mid P_i = p_0, P_{-i} = p_1\big] \\
        =& \mathbb{E}\big[\big(Y_i(1, 1) - Y_i(0, 1)\big) \mathbbm{1}\{p_0 < V_i \leq p_0', V_{-i} \leq p_1\} \big],
    \end{aligned}
\end{equation*}
which is the average direct effect for the subpopulation $\{p_0 < V_i \leq p_0', V_{-i} \leq p_1\}$, holding the peer's treatment fixed at $D_{-i} = 1$. The share of this subpopulation can be identified from
\begin{equation*}
    \mathbb{E}\big[D_i D_{-i} \mid P_i = p_0', P_{-i} = p_1\big] - \mathbb{E}\big[D_i D_{-i} \mid P_i = p_0, P_{-i} = p_1\big],
\end{equation*}
so the local average controlled direct effect $\operatorname{LACDE}_i^{(1)}(p_0<V_i \leq p_0', V_{-i} \leq p_1)$ is given by
\begin{equation*}
    \frac{\mu_{i,-i}^{(1)}\left(p_0^{\prime}, p_1\right)-\mu_{i,-i}^{(1)}\left(p_0, p_1\right)}{C\left(p_0^{\prime}, p_1\right)-C\left(p_0, p_1\right)} = \operatorname{LACDE}_i^{(1)}(p_0<V_i \leq p_0', V_{-i} \leq p_1),
\end{equation*}
where $\mu_{i,-i}^{(d)}(p_0, p_1) \equiv \mathbb{E}[Y_i \mathbbm{1}\{D_{-i}=d\} \mid P_i=p_0, P_{-i}=p_1]$, for $d \in \{0,1\}$.

Replacing $D_{-i}$ with $(1 - D_{-i})$ in Equations \eqref{eq:id_lace3} and \eqref{eq:id_lace4} and repeating the same reasoning identifies
\begin{equation*}
    \begin{aligned}
        &\mathbb{E}\big[Y_i (1 - D_{-i}) \mid P_i = p_0', P_{-i} = p_1\big] - \mathbb{E}\big[Y_i (1 - D_{-i}) \mid P_i = p_0, P_{-i} = p_1\big] \\
        =& \mathbb{E}\big[\big(Y_i(1, 0) - Y_i(0, 0)\big) \mathbbm{1}\{p_0 < V_i \leq p_0', V_{-i} > p_1\} \big],
    \end{aligned}
\end{equation*}
which corresponds to the average direct effect for the subpopulation $\{p_0 < V_i \leq p_0', V_{-i} > p_1\}$, holding the peer $-i$'s treatment fixed at $D_{-i} = 0$. 

Dividing by the fraction of this subpopulation, which is identified as
\begin{equation*}
    \big(p_0' - p_0\big) - \big[C(p_0', p_1) - C(p_0, p_1)\big] = \mathbb{P}\big(p_0 < V_i \leq p_0', V_{-i}  > p_1\big),
\end{equation*}
yields the local average controlled direct effect
\begin{equation*}
\frac{\mu_{i,-i}^{(0)}(p_0', p_1) - \mu_{i,-i}^{(0)}(p_0, p_1)}
{\big(p_0' - p_0\big) - \big[C(p_0', p_1) - C(p_0, p_1)\big]} = \operatorname{LACDE}_i^{(0)}(p_0 < V_i \leq p_0', V_{-i} > p_1).
\end{equation*}

\textit{Identifying the generalized LACSEs and LACDEs:} Suppose now that the support of the propensity scores exhibits greater variation, such that four distinct pairs of propensity scores, $(p_0, p_1)$, $(p_0, p_1')$, $(p_0', p_1)$, and $(p_0', p_1')$, exist in $\mathcal{P}$ with $p_0' > p_0$ and $p_1' > p_1$. These variations enable the identification of $\operatorname{LACSE}_i^{(d)}(P)$ and $\operatorname{LACDE}i^{(d)}(P)$ for $d \in {0,1}$, corresponding to the subpopulation with unobserved characteristics in the region $\{p_0 < V_i \leq p_0', p_1 < V{-i} \leq p_1'\}$.

For example, applying the identification strategy developed earlier to the points $(p_0', p_1)$ and $(p_0', p_1')$ identifies the average spillover effect for the subpopulation $\{V_i \leq p_0', p_1 < V_{-i} \leq p_1'\}$, holding unit $i$'s treatment fixed at $D_i = 1$, as
\begin{equation*}
    \begin{aligned}
& \mathbb{E}\big[Y_i D_i \mid P_i=p_0', P_{-i}=p_1'\big]-\mathbb{E}\big[Y_i D_i \mid P_i=p_0', P_{-i}=p_1\big] \\
= & \mathbb{E}\big[\left(Y_i(1,1)-Y_i(1,0)\right) \mathbbm{1}\left\{V_i \leq p_0', p_1<V_{-i} \leq p_1'\right\}\big].
\end{aligned}
\end{equation*}
Subtracting the previously identified average spillover effect for the subpopulation $\{V_i \leq p_0, p_1 < V_{-i} \leq p_1'\}$ yields
\begin{equation*}
    \begin{aligned}
        &\big(\mathbb{E}\big[Y_i D_i \mid P_i=p_0', P_{-i}=p_1'\big]-\mathbb{E}\big[Y_i D_i \mid P_i=p_0', P_{-i}=p_1\big]\big) \\
        & - \big(\mathbb{E}\big[Y_i D_i \mid P_i=p_0, P_{-i}=p_1'\big]-\mathbb{E}\big[Y_i D_i \mid P_i=p_0, P_{-i}=p_1\big]\big) \\
        =&\mathbb{E}\big[\left(Y_i(1,1)-Y_i(1,0)\right) \mathbbm{1}\{p_0 < V_i \leq p_0', p_1<V_{-i} \leq p_1'\}\big],
    \end{aligned}
\end{equation*}
which identifies the average spillover effect for the subpopulation
$\{p_0 < V_i \leq p_0', p_1 < V_{-i} \leq p_1'\}$, holding unit $i$'s treatment fixed at $D_i = 1$.

Since the proportion of the subpopulation $\{p_0 < V_i \leq p_0', p_1<V_{-i} \leq p_1'\}$ can be identified from 
\begin{equation*}
    \begin{aligned}
        &\big(\mathbb{E}[D_i D_{-i} \mid P_i=p_0', P_{-i}=p_1']-\mathbb{E}[D_i D_{-i} \mid P_i=p_0', P_{-i}=p_1]\big) \\
        &- \big(\mathbb{E}[D_i D_{-i} \mid P_i=p_0, P_{-i}=p_1']-\mathbb{E}[D_i D_{-i} \mid P_i=p_0, P_{-i}=p_1]\big),
    \end{aligned}
\end{equation*}
the $\operatorname{LACSE}_i^{(1)}\big(p_0 < V_i \leq p_0', p_1<V_{-i} \leq p_1'\big)$ is obtained as 
\begin{equation*}
    \frac{\left[\mu_{i, i}^{(1)}\left(p_0^{\prime}, p_1^{\prime}\right)-\mu_{i, i}^{(1)}\left(p_0^{\prime}, p_1\right)\right]-\left[\mu_{i, i}^{(1)}\left(p_0, p_1^{\prime}\right)-\mu_{i, i}^{(1)}\left(p_0, p_1\right)\right]}{\left[C\left(p_0^{\prime}, p_1^{\prime}\right)-C\left(p_0^{\prime}, p_1\right)\right]-\left[C\left(p_0, p_1^{\prime}\right)-C\left(p_0, p_1\right)\right]},
\end{equation*}
where $\mu_{i, i}^{(d)}(p_0, p_1) \equiv \mathbb{E}[Y_i \mathbbm{1}\{D_i = d\} \mid P_i = p_0, P_{-i} = p_1]$ and $C(p_0, p_1) \equiv \mathbb{E}[D_i D_{-i} \mid P_i = p_0, P_{-i} = p_1]$.

Applying a similar strategy, the average spillover and direct effects for the subpopulations corresponding to the propensity score pairs $(p_0', p_1)$ and $(p_0', p_1')$ can also be identified. By taking differences between the previously identified average effects for the subpopulations associated with $(p_0, p_1)$ and $(p_0, p_1')$, one can identify $\operatorname{LACSE}_i^{(d)}(P)$ and $\operatorname{LACDE}_i^{(d)}(P)$ for $d \in \{0,1\}$, where $P = \{p_0 < V_i \leq p_0', p_1<V_{-i} \leq p_1'\big\}$. These identification results correspond to Item 3 of Theorem~\ref{thm:id_lace}.

\section{Local Average Effects with Binary Instruments} \label{app:id_binary_iv}
This section illustrates that when the instrumental variables exhibit limited variation, additional restrictions are necessary to identify the local average controlled effects described in Theorem~\ref{thm:id_lace}.

Consider the case in which the instrument is binary, $Z_i \in {0,1}$, for all units $i$ across groups. As discussed in Remark~\ref{remark:id_binary_iv}, the corresponding propensity scores take on only four possible values. To apply the identification results in Theorem~\ref{thm:id_lace}, these propensity scores must satisfy certain equalities. Specifically, identification of LACSEs requires that any two of $P_i(0,0)$, $P_i(0,1)$, $P_i(1,0)$, or $P_i(1,1)$ be equal, while identification of LACDEs requires that any two of $P_{-i}(0,0)$, $P_{-i}(0,1)$, $P_{-i}(1,0)$, or $P_{-i}(1,1)$ be equal. 

The \textit{one-sided noncompliance} condition that frequently imposed in the literature is a special case of these requirements. It assumes that no unit takes the treatment unless assigned to it, which translates to
\begin{equation*}
    P_i(0, 0) = P_i(0, 1) = 0, \quad P_{-i}(0, 0) = P_{-i}(0, 1) = 0.
\end{equation*}
Assume that the propensity scores can be ordered as
\begin{equation*}
    P_i(0, 0) \leq P_i(0, 1) \leq P_i(1, 0) \leq P_i(1, 1),
\end{equation*}
The four observed pairs of propensity scores corresponding to different instrument assignments are illustrated as black dots in Figure \ref{fig:app_id_binary_iv}.

\begin{figure}[!htt]
        \centering
        \caption{Identifying Local Average Controlled Effects via Binary Instrument}
        \includegraphics[width = \textwidth]{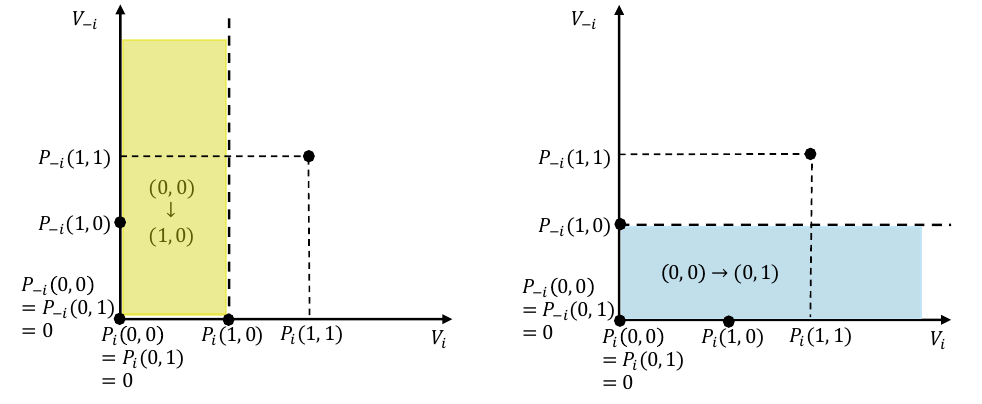}
        \label{fig:app_id_binary_iv}
\end{figure}

Under one-sided noncompliance, consider first a change in the instrument assignment from $(Z_i, Z_{-i}) = (0,0)$ to $(Z_i, Z_{-i}) = (1,0)$. In this case, unit~$i$'s propensity score increases from $P_i(0,0)$ to $P_i(1,0)$, while the peer $-i$'s propensity score remains constant at $P_{-i}(0,0) = P_{-i}(0,1) = 0$. According to Item 2 of Theorem \ref{thm:id_lace}, this variation identifies the local average controlled direct effect $\operatorname{LACDE}_i^{(1)}(0 < V_i \leq P_{i}(1, 0), 0 < V_{-i} < 1)$, as illustrated by the yellow-shaded area in the left panel of Figure \ref{fig:app_id_binary_iv}. 

Similarly, if the instrument assignment changes from $(Z_i, Z_{-i}) = (0,0)$ to $(Z_i, Z_{-i}) = (0,1)$, the peer $-i$'s propensity score shifts from $P_{-i}(0,0)$ to $P_{-i}(1,0)$, while the unit~$i$'s propensity score remains constant at $P_i(0,0) = P_i(0,1) = 0$. In this case, Item 1 of Theorem \ref{thm:id_lace} implies identification of the local average controlled spillover effect $\operatorname{LACSE}_i^{(0)}(0 < V_i < 1, 0 < V_{-i} < P_{-i}(1, 0))$, as shown by the blue-shaded area in the right panel of Figure \ref{fig:app_id_binary_iv}. The identified local average controlled spillover and direct effects correspond to the same causal parameters identified in \cite{vazquez2023causal}, which studies a similar setting with spillovers in both outcomes and endogenous treatment using a binary instrumental variable.

Without additional equalities among the propensity scores, it is generally difficult to identify the generalized local average controlled spillover and direct effects. The challenge arises because limited variation in the instrumental variables prevents constructing two distinct pairs of propensity scores in which one unit's propensity score changes while the other's remains fixed. Such variation is crucial for identification, as the model structure embedded in the treatment selection equation implies a monotonicity condition for each unit's treatment decision. Specifically, if the propensity scores satisfy
\begin{equation*}
    P_i(0, 0) \leq P_i(0, 1) \leq P_i(1, 0) \leq P_i(1, 1),
\end{equation*}
then the treatment selection mechanism implies the corresponding monotonicity in potential treatments:
\begin{equation*}
    D_i(0, 0) \leq D_i(0, 1) \leq D_i(1, 0) \leq D_i(1, 1),
\end{equation*}
where $D_i(z, z')$ denotes the potential treatment under the instrument assignment $(Z_i, Z_{-i}) = (z, z')$.
This condition ensures that each individual's treatment increases monotonically with instrument assignments, so that when one unit's propensity score changes while the other's remains fixed, we can isolate and identify the corresponding local average controlled spillover or direct effect.

However, when considering the joint treatment vector $(D_i, D_{-i})$ for both group members, the same structure does not guarantee monotonicity at the pairwise level. For example, if $D_i(0,1) < D_i(1,0)$ holds for all units across groups, shifting the instrument assignment from $(Z_i, Z_{-i}) = (0,1)$ to $(Z_i, Z_{-i}) = (1,0)$ leads to $D_i(1,0) > D_i(0,1)$ and $D_{-i}(0,1) < D_{-i}(1, 0)$. In this case, unit $i$ switches from untreated to treated while the peer $-i$ switches in the opposite direction, violating monotonicity at the group level.

To restore monotonicity and achieve identification, it is therefore necessary to impose a condition ensuring that at least two of the propensity scores remain constant when the instrument changes. This restriction allows the propensity score of one unit to change while the other's remains constant, creating variation in treatment decisions where only one unit alters its treatment status. Such variation is essential for identifying the local average controlled spillover and direct effects.

A graphical representation helps clarify this point. Consider the case in which the propensity score values do not satisfy the required equality conditions. When this occurs, local average controlled effects cannot be identified using a binary instrumental variable. Suppose the propensity scores follow a strict ordering,
\begin{equation*}
    P_i(0, 0) < P_i(0, 1) < P_i(1, 0) < P_i(1, 1)
\end{equation*}
for all units, so that none of the equality conditions hold. In this case, the four observed pairs of propensity scores are represented by the black dots in Figure \ref{fig:app_id_fail_binary}, which fail to form the “vertices” of a rectangle, illustrating the absence of sufficient variation needed for identification.

\begin{figure}[!htt]
        \centering
        \caption{Failure of Point Identification with a Binary Instrument}
        \includegraphics[width = 0.98\textwidth]{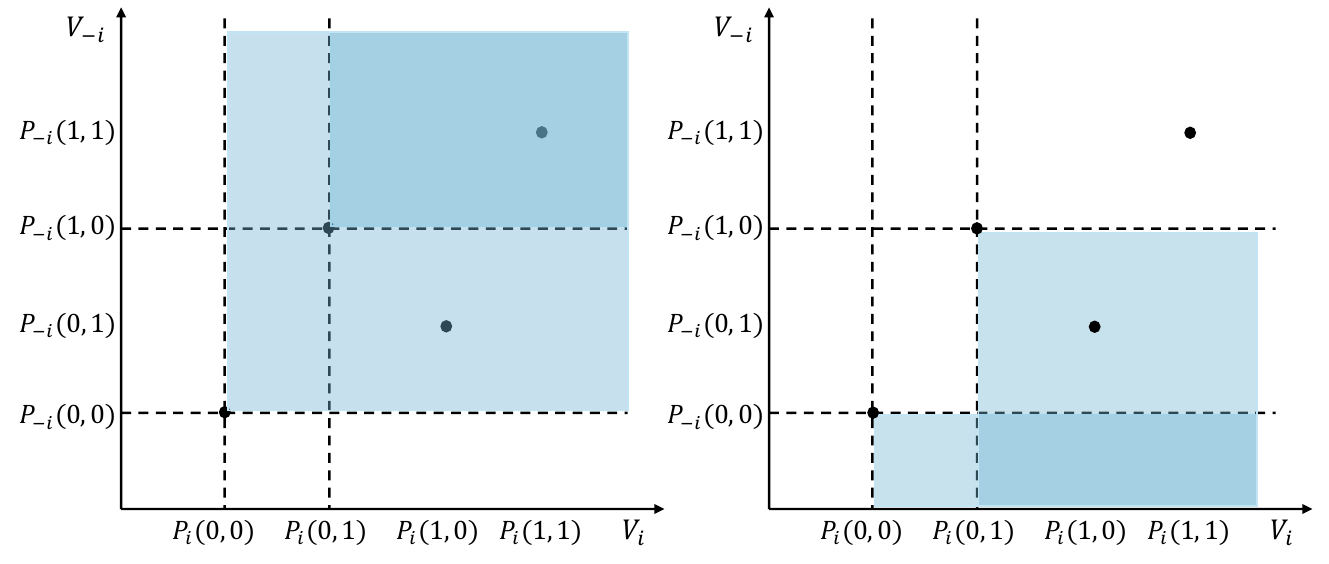}
        \label{fig:app_id_fail_binary}
\end{figure}

To illustrate the challenge, consider two pairs of propensity scores, $(P_i(0,0), P_{-i}(0,0))$ and $(P_i(0,1), P_{-i}(1,0))$, and use them to examine identification of the local average controlled spillover effect while holding unit $i$'s treatment fixed at $D_i = 0$. Following the identification logic developed earlier, the observed conditional expectations $\mathbb{E}[Y_i (1 - D_i) D_{-i} \mid \cdot, \cdot]$ and $\mathbb{E}[Y_i (1 - D_i)(1 - D_{-i}) \mid \cdot, \cdot]$ are evaluated at these two pairs of propensity scores. 
Taking the difference
\begin{equation*}
    \begin{aligned}
        &\mathbb{E}[Y_i (1 - D_i) (1 - D_{-i}) \mid P_i = P_i(0, 1), P_{-i} = P_{-i}(1, 0)] \\
        &- \mathbb{E}[Y_i (1 - D_i) (1 - D_{-i}) \mid P_i = P_i(0, 0), P_{-i} = P_{-i}(0, 0)],
    \end{aligned}
\end{equation*}
identifies the average potential outcome $Y_i(0,0)$ for the subpopulation corresponding to the lighter blue L-shaped region in the left panel of Figure \ref{fig:app_id_fail_binary}.

Taking another difference.
\begin{equation*}
    \begin{aligned}
        &\mathbb{E}[Y_i (1 - D_i) D_{-i} \mid P_i = P_i(0, 1), P_{-i} = P_{-i}(1, 0)] \\
        &- \mathbb{E}[Y_i (1 - D_i) D_{-i} \mid P_i = P_i(0, 0), P_{-i} = P_{-i}(0, 0)],
    \end{aligned}
\end{equation*}
does not identify the average potential outcome $Y_i(0,1)$ for a well-defined subpopulation, because the regions of $(V_i, V_{-i})$ corresponding to these two treatment realizations, evaluated at the two pairs of propensity scores, partially overlap without one fully containing the other. As illustrated by the two blue-shaded rectangles in the right panel of Figure \ref{fig:app_id_fail_binary}, this lack of nesting violates the monotonicity condition at the group level.

In this case, it becomes impossible to identify averages of two distinct potential outcomes, such as $Y_i(0,0)$ and $Y_i(0,1)$, for the same subpopulation. Without this, taking their difference to identify local average controlled effects is infeasible. Therefore, the generalized local average controlled spillover and direct effects cannot be point identified unless additional equality conditions on the propensity scores are imposed.

Since propensity scores can be recovered from observed data, these equality conditions can be directly tested. If the required conditions are not satisfied, additional variation in the instrumental variables is necessary to achieve point identification.

\section{Proof of Identification for MCSEs and MCDEs} 

\subsection{Identifying Copula and Copula Density} \label{appendix:id_copula}

Under Assumptions \ref{as:RA}-\ref{as:Vdist}, the propensity score for unit $i$ can be shown as
\begin{equation} \label{eq:p_score1}
    \begin{aligned}
        &\mathbb{P}\left(D_{i} = 1 \mid Z_{i} = z_0, Z_{-i} = z_1\right) \\
        =& \mathbb{P}\left(V_{i} \leq h_i(Z_{i}, Z_{-i}) \mid Z_{i} = z_0, Z_{-i} = z_1\right) \\
        =& \mathbb{P}\left(V_{i} \leq h_i(z_0, z_1) \mid Z_{i} = z_0, Z_{-i} = z_1\right) \\
        =& \mathbb{P}\left(V_{i} \leq h_i(z_0, z_1)\right) \\
        =& h_i(z_0, z_1),
    \end{aligned}
\end{equation}
given $z_0 \in \text{Supp}(Z_i) = \mathbb{R}^{k_i}, z_1 \in \text{Supp}(Z_{-i}) = \mathbb{R}^{k_{-i}}$. The third equality follows directly from Assumption \ref{as:RA}. Furthermore, under Assumption \ref{as:Vdist}, we normalize $ V_i $ to follow a uniform distribution $ \mathcal{U}(0, 1)$, which justifies the final equality. Equation \eqref{eq:p_score1} shows that the threshold function $h_i(Z_i, Z_{-i})$ in the treatment mechanism is identified by the propensity score function $P_i$ on its support $\mathcal{P}_i$.

Once the propensity scores $(P_i, P_{-i}) \in \mathcal{P}$ of all group members are identified, the copula $C_{V_i, V_{-i}}(p_0, p_1)$, which characterizes the dependence structure between the unobserved heterogeneities within the group, can be identified as
\begin{equation} \label{eq:p_score2}
    \begin{aligned}
        &\mathbb{P}\big(D_{i} = 1, D_{- i} = 1 \mid P_{i}\big(Z_{i}, Z_{-i}\big) = p_0, P_{-i}\big(Z_{-i}, Z_{i}\big) = p_1\big) \\
        =& \mathbb{P}\big(V_{i} \leq h_i(Z_{i}, Z_{-i}), V_{-i} \leq h_{-i}(Z_{-i}, Z_{i}) \mid P_{i}\big(Z_{i}, Z_{-i}\big) = p_0, P_{-i}\big(Z_{-i}, Z_{i}\big) = p_1\big) \\
        =& \mathbb{P}\big(V_{i} \leq p_0, V_{-i} \leq p_1 \mid P_{i}\big(Z_{i}, Z_{-i}\big) = p_0, P_{-i}\big(Z_{-i}, Z_{i}\big) = p_1\big) \\
        =& \mathbb{P}\big(V_{i} \leq p_0, V_{-i} \leq p_1\big) \\
        =& C_{V_i, V_{-i}}(p_0, p_1)
    \end{aligned}
\end{equation}
for $(p_0, p_1) \in \mathcal{P}$, where the second equality follows from the identification of the threshold function $h_i$ by the propensity score $P_i$, and the last equality holds under Assumption \ref{as:RA}. 

If Assumption \ref{as:cts_iv} holds and $\mathbb{E}[D_i D_{-i} \mid P_i, P_{-i}]$ is twice differentiable at $(p_0, p_1) \in \mathcal{P}$, then the copula density can be identified by taking second-order derivatives,  
    \begin{equation} \label{eq:joint_density}
        \begin{aligned}
            &\frac{\partial^2 \mathbb{E}\big[D_i D_{-i} \mid P_i=p_0, P_{-i}=p_1\big]}{\partial p_0 \partial p_1} \\
            =& \frac{\partial ^2 \mathbb{P}\big(V_{i} \leq p_0, V_{-i} \leq p_1\big)}{\partial p_0 \partial p_1} = c_{V_i, V_{-i}}(p_0, p_1).
        \end{aligned} 
    \end{equation}

\subsection{Identifying the marginal treatment response functions} \label{appendix:id_mtr}

Given the values of propensity scores $P_{i}(Z_{i}, Z_{-i}) = p_0$, $P_{-i}(Z_{-i}, Z_{i}) = p_1$, and any Borel set $A \subset \mathcal{Y}$, we have
%Should we have $(Z_{ig} = z, Z_{-i, g} = z')$ and $(Z_{-i, g} = z', Z_{ig} = z)$?
\begin{equation*} 
    \begin{aligned}
        &\mathbb{E}\big[\mathbbm{1}\{Y_{i} \in A\} D_{i} D_{-i} \mid P_{i}(Z_{i}, Z_{-i}) = p_0, P_{-i}(Z_{-i}, Z_{i}) = p_1\big] \\
        =& \mathbb{E}\big[\mathbbm{1}\{Y_{i} (1, 1) \in A\} \cdot \mathbbm{1}\{V_{i} \leq h_i(Z_{i}, Z_{-i})\} \cdot \mathbbm{1}\big\{V_{-i} \leq h_{-i}(Z_{-i}, Z_{i})\big\} \mid P_{i}(Z_{i}, Z_{-i}) = p_0, P_{-i}(Z_{-i}, Z_{i}) = p_1\big] \\
        =& \mathbb{E}\big[\mathbbm{1}\{Y_{i} (1, 1) \in A\} \cdot \mathbbm{1}\{V_{i} \leq p_0\} \cdot \mathbbm{1}\big\{V_{-i} \leq p_1\big\} \mid P_{i}(Z_{i}, Z_{-i}) = p_0, P_{-i}(Z_{-i}, Z_{i}) = p_1\big] \\
        =& \mathbb{E}\big[\mathbbm{1}\{Y_{i} (1, 1) \in A\} \cdot \mathbbm{1}\{V_{i} \leq p_0\} \cdot \mathbbm{1}\big\{V_{-i} \leq p_1\big\}\big] \\
        =& \int_{0}^{p_1} \int_{0}^{p_0} \mathbb{P}\big(Y_{i}(1, 1) \in A \mid V_{i} = v_0, V_{-i} = v_1\big) c_{V_i, V_{-i}}(v_0, v_1)dv_0 dv_1,
    \end{aligned}
\end{equation*} 
where the second equality follows from Equation \eqref{eq:p_score1}, and the third equality holds under Assumption \ref{as:RA}. If the function $\mathbb{E}\big[\mathbbm{1}\{Y_i(1, 1) \in A\} D_i D_{-i} \mid \cdot, \cdot\big]$ is twice differentiable and $m_{i}^{\big(1, 1\big)}(\cdot, \cdot)$ is continunous at $(p_0, p_1)$, by the Leibniz integral rule,
\begin{equation} \label{eq:cond_po1}
    \begin{aligned}
        &\frac{\partial^2}{\partial p_1 \partial p_0} \mathbb{E}\big[\mathbbm{1}\{Y_i(1, 1) \in A\} D_{i} D_{-i} \mid P_{i}(Z_{i}, Z_{-i}) = p_0, P_{-i}(Z_{-i}, Z_{i}) = p_1\big] \\
        =& \mathbb{P}\big(Y_{i}(1, 1) \in A \mid V_{i} = p_0, V_{-i} = p_1\big) \cdot c_{V_i, V_{-i}}(p_0, p_1).
    \end{aligned}
\end{equation}
Since the copula density of $(V_{i}, V_{-i})$, $c(\cdot, \cdot)$, is identified from Corollary \ref{corr:id_copula_density}, we can identify $\mathbb{P}\big(Y_{i}(1, 1) \in A \mid V_{i} = p_0, V_{-i} = p_1\big)$ from Equation \eqref{eq:cond_po1}. This, in turn, implies that the marginal treatment response function $m_i^{(1, 1)}(p_0, p_1)$ is identified.

We can apply the same procedure and obtain
%Add the derivations in the appendix
\begin{equation} \label{eq:cond_po2}
    \begin{aligned}
        &-\frac{\partial^2}{\partial p_1 \partial p_0} \mathbb{E}\big[\mathbbm{1}\{Y_i \in A\} D_{i} (1 - D_{-i}) \mid P_{i}(Z_{i}, Z_{-i}) = p_0, P_{-i}(Z_{-i}, Z_{i}) = p_1\big] \\
        =& \mathbb{P}\big(Y_{i}(1, 0) \in A \mid V_{i} = p_0, V_{-i} = p_1\big) \cdot c_{V_i, V_{-i}}(p_0, p_1), \\
        &-\frac{\partial^2}{\partial p_1 \partial p_0} \mathbb{E}\big[\mathbbm{1}\{Y_i \in A\} (1 - D_{i}) D_{-i} \mid P_{i}(Z_{i}, Z_{-i}) = p_0, P_{-i}(Z_{-i}, Z_{i}) = p_1\big] \\
        =& \mathbb{P}\big(Y_{i}(0, 1) \in A \mid V_{i} = p_0, V_{-i} = p_1\big) \cdot c_{V_i, V_{-i}}(p_0, p_1), \\
        &\frac{\partial^2}{\partial p_1 \partial p_0} \mathbb{E}\big[\mathbbm{1}\{Y_i \in A\} (1 - D_{i}) (1 - D_{-i}) \mid P_{i}(Z_{i}, Z_{-i}) = p_0, P_{-i}(Z_{-i}, Z_{i}) = p_1\big] \\
        =& \mathbb{P}\big(Y_{i}(0, 0) \in A \mid V_{i} = p_0, V_{-i} = p_1\big) \cdot c_{V_i, V_{-i}}(p_0, p_1).
    \end{aligned}
\end{equation}
We can then identify the remaining marginal treatment response functions from Equation \eqref{eq:cond_po2}.

Based on results in Equations \eqref{eq:cond_po1}-\eqref{eq:cond_po2}, $\operatorname{MCSE}_{i}^{(d)}(p_0, p_1)$ and $\operatorname{MCDE}_{i}^{(d)}(p_0, p_1)$ is identified for any $d \in \{0, 1\}$ and $p_0, p_1 \in \mathcal{P}$.

\subsection{Identification With Exogenous Covariates} \label{appendix:id_covariates}
The identification results can be extended to settings with exogenous covariates. Let $X_i \in \mathbb{R}^{d_i}$ denote a vector of covariates that affect both the outcome and the treatment assignment for unit $i$. For instance, unit $i$'s earnings and education choices may depend on the family characteristics of both herself and her best friend in our leading example.  Given $(X_{i}, X_{-i}) = \mathbf{x}$, $D_i = d$, and $D_{-i} = d'$, we model the potential outcome as
\begin{equation} \label{eq:model_covariates}
    Y_{i} (\mathbf{x}, d, d') = \mu_{d d'}\big(\mathbf{x}, U_{i}(d, d')\big),
\end{equation}
where the functions $\mu_{dd{\prime}}(\cdot, \cdot)$ are known and specified by the researcher, while $U_i(d, d')$ captures unobserved factors affecting unit $i$'s potential outcome under own treatment status $D_i = d$ and the peer $-i$'s treatment $D_{-i} = d'$. A common specification for $\mu_{d d'}(\cdot, \cdot)$ assumes additive separability and linearity in covariates: $\mu_{dd'}\big(\mathbf{x}, U_{i}(d, d')\big) = \mathbf{x} \beta_{d d'} + U_i(d, d').$

We next introduce a potential outcome model that incorporates exogenous covariates. 
\begin{equation} \label{eq:model_disc_trt_exo_covariates}
    \left\{\begin{array}{l}
        \begin{aligned}
            Y_{i} = & \big[Y_{i} (X_{i}, X_{-i}, 1, 1) D_{-i} + Y_{i} (X_{i}, X_{-i}, 1, 0)(1 - D_{-i})\big] D_{i} \\
            & + \big[Y_{i} (X_{i}, X_{-i}, 0, 1) D_{-i} + Y_{i}(X_{i}, X_{-i}, 0, 0)(1 - D_{-i})\big] (1 - D_{i}),
            & \quad \\
        D_{i} = & \mathbbm{1}\big\{V_{i} \leq h(W_{i}, W_{-i})\big\},
        \end{aligned} \\
        \end{array}\right.
\end{equation}
where $W_{i} \equiv (X_{i}, Z_{i}) \in \mathbb{R}^{d_i} \times \mathbb{R}^{k_i}$.

Under Equation \eqref{eq:model_disc_trt_exo_covariates}, we replace Assumption \ref{as:RA} with Assumption \ref{as:exo_covariates}, which imposes random assignment of both covariates and instruments.

\begin{assumption}(Exogenous covariates and random assignment) \label{as:exo_covariates}
    The covariates $X_{i}$ and the instruments $Z_i$ satisfy
    \begin{equation*}
        \big(X_{i}, X_{-i}, Z_{i}, Z_{-i}\big) \indep \Big\{\big(V_{i}, V_{-i}, U_{i}(d, d'), U_{-i}(d, d'),\big) \Big\}_{d \in \{0, 1\}, d' \in \{0, 1\}}.
    \end{equation*}
\end{assumption}

Under Assumptions \ref{as:er}, \ref{as:Vdist}, and \ref{as:exo_covariates}, the propensity score with exogenous covariates, defined as $P_i(W_i, W_{-i}) \equiv \mathbb{P}(D_i = 1 \mid W_i, W_{-i})$, can be expressed as
\begin{equation} \label{eq:p_score_covariates1}
    \begin{aligned}
        &\mathbb{P}\left(D_{i} = 1 \mid W_{i} = w_0, W_{-i} = w_1\right) \\
        =& \mathbb{P}\left(V_{i} \leq h_i(W_{i}, W_{-i}) \mid W_{i} = w_0, W_{-i} = w_1\right) \\
        =& \mathbb{P}\left(V_{i} \leq h_i(w, w') \mid W_{i} = w_0, W_{-i} = w_1\right) \\
        =& \mathbb{P}\left(V_{i} \leq h_i(w_0, w_1)\right) \\
        =& h_i(w_0, w_1)
    \end{aligned}
\end{equation}
given $W_{i} = w_0, W_{-i} = w_1$. Equation \eqref{eq:p_score_covariates1} demonstrates that, in the presence of exogenous covariates, the propensity score $P_i(W_i, W_{-i})$ continues to identify the threshold function $h_i$ over its support $\mathcal{P}_i$.

Similar to Equation \eqref{eq:p_score2}, we can identify the copula function $c_{V_i, V_{-i}}(p_0, p_1)$ as
\begin{equation} \label{eq:p_score_covariates2}
    \begin{aligned}
        &\mathbb{P}\big(D_{i} = 1, D_{- i} = 1 \mid P_i(W_{i}, W_{-i}) = p_0, P_{-i}(W_{-i}, W_{ig}) = p_1\big) \\
        =& \mathbb{P}\big(V_{i} \leq h(W_{i}, W_{-i}), V_{-i} \leq h(W_{-i}, W_{i}) \mid P_i(W_{i}, W_{-i}) = p_0, P_{-i}(W_{-i}, W_{i}) = p_1\big) \\
        =& \mathbb{P}\big(V_{i} \leq p_0, V_{-i} \leq p_1 \mid P_i(W_{i}, W_{-i}) = p_0, P_{-i}(W_{-i}, W_{i}) = p_1\big) \\
        =& \mathbb{P}\left(V_{i} \leq p_0, V_{-i} \leq p_1\right),
    \end{aligned}
\end{equation}
where the last equality holds under Assumption \ref{as:RA_cts}. Then, the copula density of $(V_{i}, V_{-i})$, $c_{V_i, V_{-i}}(\cdot, \cdot)$, is identifiable provided that the copula $C_{V_i, V_{-i}}(\cdot, \cdot)$ is twice differentiable,
\begin{equation} \label{eq:joint_density2}
        \begin{aligned}
            &\frac{\partial ^2}{\partial p_1 \partial p_0} \mathbb{P}\big(D_{i}=1, D_{-i}=1 \mid P_i(W_{i}, W_{-i})=p_0, P_{-i}(W_{-i}, W_{i})=p_1\big) \\
            =& \frac{\partial ^2}{\partial p_1 \partial p_0} \mathbb{P}\big(V_{i} \leq p_0, V_{-i} \leq p_1\big) = c_{V_i, V_{-i}}(p_0, p_1).
        \end{aligned} 
\end{equation}

The last step is to identify the marginal treatment response functions, defined as $m_i^{(\mathbf{x}, d, d')}(p_0, p_1) \equiv \mathbb{E}\big[Y_i(\mathbf{x}, d, d') \mid V_i=p_0, V_{-i}=p_1\big]$ with covariates. Given the covariates and propensity scores of both units $i$ and her peer $-i$, we can express the following conditional expectation as 
\begin{equation*} 
    \begin{aligned}
        &\mathbb{E}\big[Y_{i} D_{i} D_{-i} \mid (X_{i}, X_{-i}) = \mathbf{x}, P(W_{i}, W_{-i}) = p_0, P(W_{-i}, W_{i}) = p_1\big] \\
        =& \mathbb{E}\big[\mu_{11}(\mathbf{x}, U_i(1, 1)) \cdot \mathbbm{1}\{V_{i} \leq h_i(W_{i}, W_{-i})\} \cdot \mathbbm{1}\left\{V_{-i} \leq h_{-i}(W_{-i}, W_{i})\right\} \mid \\
        & \quad (X_{i}, X_{-i}) = \mathbf{x}, P_i(W_{i}, W_{-i}) = p_0, P_{-i}(W_{-i}, W_{i}) = p_1\big] \\
        =& \mathbb{E}\big[\mu_{11}(\mathbf{x}, U_{i}(1, 1)) \cdot \mathbbm{1}\{V_{i} \leq p_0\} \cdot \mathbbm{1}\left\{V_{-i} \leq p_1\right\} \mid \\
        & \quad (X_{i}, X_{-i}) = \mathbf{x}, P_i(W_{i}, W_{-i}) = p_0, P_{-i}(W_{-i}, W_{i}) = p_1\big] \\
        =& \mathbb{E}\big[\mu_{11}(\mathbf{x}, U_{i}(1, 1)) \cdot \mathbbm{1}\{V_{i} \leq p_0\} \cdot \mathbbm{1}\left\{V_{-i} \leq p_1\right\}\big] \\
        =& \int_{0}^{p_1} \int_{0}^{p_0} \mathbb{E}\big[\mu_{11}(\mathbf{x}, U_{i}(1, 1)) \mid V_{i} = v_0, V_{-i} = v_1\big] c_{V_i, V_{-i}}(v_0, v_1)dv_0 dv_1,
    \end{aligned}
\end{equation*} 
where the second equality follows from Equation \eqref{eq:p_score_covariates1}, and the third equality holds under Assumption \ref{as:exo_covariates}.  If the conditional mean $\mathbb{E}[Y_{i} D_{i} D_{-i} \mid (X_{i}, X_{-i}) = \mathbf{x}, \cdot, \cdot]$ is twice differentiable, and the marginal treatment response functions $\mathbb{E}[\mu_{dd'}(\mathbf{x}, U_{i}(d, d')) \mid \cdot, \cdot]$ are continuous at $(p_0, p_1)$, then the marginal treatment response (MTR) function $\mathbb{E}[\mu_{11}(\mathbf{x}, U_{i}(1, 1)) \mid V_i = p_0, V_{-i} = p_1]$ is identified by taking the cross-derivative as shown below:
\begin{equation*}
    \begin{aligned}
        &\frac{\partial^2}{\partial p_1 \partial p_0} \mathbb{E}\left[Y_{i} D_{i} D_{-i} \mid (X_{i}, X_{-i}) = \mathbf{x}, P_i(W_{i}, W_{-i}) = p_0, P_{-i}(W_{-i}, W_{i}) = p_1\right] \\
        =& \mathbb{E}\left[\mu_{11}(\mathbf{x}, U_{i}(1, 1)) \mid V_{i} = p_0, V_{-i} = p_1\right] c_{V_i, V_{-i}}(p_0, p_1).
    \end{aligned}
\end{equation*}
Therefore, the MTR function $m_i^{(\mathbf{x}, 1, 1)}(p_0, p_1)$ is identified, given that the copula density is identified as in Equation \eqref{eq:joint_density2}. By analogous reasoning, the remaining MTR functions $m_i^{(\mathbf{x}, d, d')}(p_0, p_1)$ are also identified for all $d, d' \in \{0, 1\}$ and $\mathbf{x} \in \mathbb{R}^{d_i}$. The marginal controlled spillover and direct effects are obtained by taking differences between the marginal treatment response (MTR) functions $m_i^{(\mathbf{x}, d, d')}(p_0, p_1)$ for $d, d' \in {0,1}$.

\section{Deriving Policy Relevant Treatment Effects with MCSEs and MCDEs} \label{app:prte}

In this section, we identify the PRTEs under three types of common policy interventions with identified MCSEs and MCDEs.

\textbf{Case 1: Absolute increase by an exogenous value.} Suppose there exists an alternative policy $a' \in \mathcal{A}$ that exogenously increases the propensity score of all units by a constant $\varepsilon > 0$, such that $P^{a'}_i = P^a_i + \varepsilon$ and $P^a_i, P^{a'}_i \in [0, 1]$, for all $i$ in every group. By taking the difference between the expected outcomes under the two policies, $\mathbb{E}[Y^a_i]$ and $\mathbb{E}[Y^{a'}_i]$, we can express this difference as weighted average of MCDEs and MCSEs as follows, 
\begin{equation*}
    \begin{aligned}
      & \mathbb{E}\left[Y_i^{a'}-Y_i^a\right] =\int_0^1 \int_0^1\bigg\{\operatorname{MCDE_i}(0; p_0, p_1) \mathbb{P}\left(p_0-\varepsilon \leq P_i^a \leq p_0, P_{-i}^a \leq p_1-\varepsilon\right) \\
      & +\operatorname{MCSE}_i(0; p_0, p_1) \mathbb{P}\left(P_i^a \leq p_0-\varepsilon, p_1-\varepsilon \leq P_{-i}^a<p_1\right) \\
      & +\operatorname{MCDE}_i(1;p_0,p_1) \mathbb{P}\left(p_0-\varepsilon \leq P_i^a \leq p_0, p_1 \leq P_{-i}^a\right) \\
      & +\operatorname{MCSE}_i(1;p_0,p_1) \mathbb{P}\left(p_0 \leq P_i^a, p_1-\varepsilon \leq P_{-i}^a<p_1\right) \\
      & +\left(\operatorname{MCDE}_i(1;p_0, p_1) + \operatorname{MCSE}_i(0;p_0,p_1)\right) \\
      & \quad \mathbb{P}\left(p_0 - \varepsilon \leq P_i^a < p_0, p_1 - \varepsilon \leq P_{-i}^a < p_1\right) \bigg\} c_{V_i, V_{-i}}(p_0, p_1) dp_0 dp_1
      \end{aligned}
  \end{equation*}

%{\color{blue} Explain the meaning of the above expressions in details: The policy change induces five different changes in group members' treatment selection, so the above difference can be derived by integrating the five different treatment effetcs over the corresponding area. Better if explaining using graphs.} 

Once we can identify MCDEs and MCSEs, as well as the joint distributions of propensity scores $(P_i^a, P_{-i}^a)$ and unobservables $(V_i, V_{-i})$ over the full support $[0, 1] \times [0, 1]$, the policy relevant treatment effect can be point identified as 
\begin{equation*}
    \mathbb{E}\left[Y^{a'}_i - Y^a_i\right] / \Delta P, 
\end{equation*}
where $\Delta P$ denotes the proportion of groups in which at least one member changes treatment status as a result of the policy shift from $a$ to $a'$. This proportion is identified as  
\begin{equation*}
    \begin{aligned}
        \Delta P=&\int_{0}^{1} \int_{0}^{1} \bigg\{\mathbb{P}\left(p_0 - \varepsilon \leq P^a_i \leq p_0, P_{-i}^a \leq p_1 -\varepsilon\right) + \mathbb{P}\left(P^a_i \leq p_0 - \varepsilon, p_1 - \varepsilon \leq P_{-i}^a < p_1\right) \\
        &+ \mathbb{P}\left(p_0 - \varepsilon \leq P^a_i \leq p_0, p_1 \leq P_{-i}^a\right) + \mathbb{P}\left(p_0 \leq P^a_i, p_1 - \varepsilon \leq P_{-i}^a < p_1\right) \\
        &+ \mathbb{P}\left(p_0 - \varepsilon \leq P_i^a < p_0, p_1 - \varepsilon \leq P_{-i}^a < p_1\right) \bigg\} c_{V_i, V_{-i}}(p_0, p_1) dp_0 dp_1.
    \end{aligned} 
\end{equation*}

We can also identify the PRTEs for cases where $\varepsilon<0$, or where the policy shift affects group members in opposite directions—for instance, $\varepsilon_i > 0$ and $\varepsilon_{-i} < 0$—by applying analogous derivations.

\textbf{Case 2: Proportional increase by an exogenous value.} Consider an alternative policy $a' \in \mathcal{A}$ that exogenously increases the propensity score of all individuals proportionally, such that $P^{a'}_i = P^a_i + \varepsilon (1-P^a_i)$ for all individuals $i$, where $0< \varepsilon < 1$ and $P^a_i, P^{a'}_i \in [0,1]$. Under this policy shift, we can identify the PRTE as $\mathbb{E}[Y^{a'}_i - Y^a_i] / \Delta P$, where
\begin{equation*}
    \begin{aligned}
      & \mathbb{E}\left[Y_i^{a^{\prime}}-Y_i^a\right] =\int_0^1 \int_0^1\bigg\{\operatorname{MCDE_i}(0; p_0, p_1) \mathbb{P}\left(\frac{p_0 - \varepsilon}{1-\varepsilon} \leq P_i^a \leq p_0, P_{-i}^a \leq \frac{p_1-\varepsilon}{1-\varepsilon}\right) \\
      & +\operatorname{MCSE}_i(0; p_0, p_1) \mathbb{P}\left(P_i^a \leq \frac{p_0 - \varepsilon}{1-\varepsilon}, \frac{p_1-\varepsilon}{1-\varepsilon} \leq P_{-i}^a<p_1\right) \\
      & +\operatorname{MCDE}_i(1;p_0,p_1) \mathbb{P}\left(\frac{p_0 - \varepsilon}{1-\varepsilon} \leq P_i^a \leq p_0, p_1 \leq P_{-i}^a\right) \\
      & +\operatorname{MCSE}_i(1;p_0,p_1) \mathbb{P}\left(p_0 \leq P_i^a, \frac{p_1-\varepsilon}{1-\varepsilon} \leq P_{-i}^a<p_1\right) \\
      & +\left(\operatorname{MCDE}_i(1;p_0, p_1) + \operatorname{MCSE}_i(0;p_0,p_1)\right) \\
      & \quad \mathbb{P}\left(\frac{p_0 - \varepsilon}{1-\varepsilon} \leq P_i^a < p_0, \frac{p_1-\varepsilon}{1-\varepsilon} \leq P_{-i}^a < p_1\right) \bigg\} c_{V_i, V_{-i}}(p_0, p_1) dp_0 dp_1, \\
      & \Delta P =\int_0^1 \int_0^1\bigg\{\mathbb{P}\left(\frac{p_0 - \varepsilon}{1-\varepsilon} \leq P_i^a \leq p_0, P_{-i}^a \leq \frac{p_1-\varepsilon}{1-\varepsilon}\right) \\
      & + \mathbb{P}\left(P_i^a \leq \frac{p_0 - \varepsilon}{1-\varepsilon}, \frac{p_1-\varepsilon}{1-\varepsilon} \leq P_{-i}^a<p_1\right) \\
      & +\mathbb{P}\left(\frac{p_0 - \varepsilon}{1-\varepsilon} \leq P_i^a \leq p_0, p_1 \leq P_{-i}^a\right) \\
      & +\mathbb{P}\left(p_0 \leq P_i^a, \frac{p_1-\varepsilon}{1-\varepsilon} \leq P_{-i}^a<p_1\right) \\
      & +\mathbb{P}\left(\frac{p_0 - \varepsilon}{1-\varepsilon} \leq P_i^a < p_0, \frac{p_1-\varepsilon}{1-\varepsilon} \leq P_{-i}^a < p_1\right) \bigg\} c_{V_i, V_{-i}}(p_0, p_1) dp_0 dp_1.
      \end{aligned}
\end{equation*}

\textbf{Case 3: Increase the instrument value.} The third type of policy intervention involves shifting the value of certain instruments. For example, consider a policy change where the $j$-th component of the instrument is increased by $\varepsilon$, such that $P_i^{a'} = P_i^a(Z + \varepsilon e_j)$, where $e_j$ denotes the unit vector in the $j$-th coordinate. In the previous policy changes, the direction of the shift in propensity scores was known for all individuals, allowing us to determine the corresponding changes in treatment responses across the entire range of unobserved characteristics $(V_i, V_{-i})$. However, when we change the instruments, the effect on propensity scores is not necessarily uniform—some individuals may experience an increase in their propensity scores, while others may see a decrease. The heterogeneous shifts in propensity scores introduce variation in group members' treatment responses, making the analysis more complicated. To address this problem, we decompose the expected outcome difference, $\mathbb{E}[Y_i^{a'} - Y_i^a]$, as 
\begin{equation*}
    \begin{aligned}
        \mathbb{E}\big[Y_i^{a'} - Y_i^a \big] =& \mathbb{E}\big[(Y_i^{a'} - Y_i^a) \mathbbm{1} \{P_i^{a'} \geq P_i^a, P_{-i}^{a'} \geq P_{-i}^a\} \big] \\
        &+ \mathbb{E}\big[(Y_i^{a'} - Y_i^a) \mathbbm{1} \{P_i^{a'} \geq P_i^a, P_{-i}^{a'} < P_{-i}^a\} \big] \\
        &+ \mathbb{E}\big[(Y_i^{a'} - Y_i^a) \mathbbm{1} \{P_i^{a'} < P_i^a, P_{-i}^{a'} \geq P_{-i}^a\} \big] \\
        &+ \mathbb{E}\big[(Y_i^{a'} - Y_i^a) \mathbbm{1} \{P_i^{a'} < P_i^a, P_{-i}^{a'} < P_{-i}^a\} \big],
    \end{aligned}
\end{equation*}
with applying the law of total probability. Given that the distribution of $P_i^a(\cdot)$ is identified and $P_i^{a'}(Z) = P_i^a (Z +\varepsilon e_j)$, we can also identify the joint distribution of $(P_i^a, P_{-i}^a, P_i^{a'}, P_{-i}^{a'})$.

We can solve each component in the above equation as 
\begin{equation*}
    \begin{aligned}
        &\mathbb{E}\big[(Y_i^{a'} - Y_i^a) \mathbbm{1} \{P_i^{a'} \geq P_i^a, P_{-i}^{a'} \geq P_{-i}^a\} \big] \\
        =& \int_{0}^{1} \int_{0}^{1} \bigg\{ \operatorname{MCDE_i}(0; p_0, p_1) \mathbb{P}\big(P_i^a < p_0 \leq P_i^{a'}, p_1 > P_{-i}^{a'} \geq P_{-i}^{a}\big) \\
        &+ \operatorname{MCSE}_i(0; p_0, p_1) \mathbb{P}\big(p_0 > P_i^{a'} \geq P_i^a, P_{-i}^{a} < p_1 \leq P_{-i}^{a'}\big) \\
        &+ \operatorname{MCDE}_i(1; p_0, p_1) \mathbb{P}\big(P_i^a < p_0 \leq P_i^{a'}, p_1 \leq P_{-i}^{a} \leq P_{-i}^{a'}\big) \\
        &+ \operatorname{MCSE}_i(1; p_0, p_1) \mathbb{P}\big(p_0 \leq P_i^{a} \leq P_i^{a'}, P_{-i}^{a} < p_1 \leq P_{-i}^{a'}\big) \\
        &+ \big(\operatorname{MCDE}_i(1;p_0, p_1) + \operatorname{MCSE}_i(0;p_0,p_1) \big) \mathbb{P}\big(P_i^a < p_0 \leq P_i^{a'}, P_{-i}^{a} < p_1 \leq P_{-i}^{a'}\big) \bigg\} c_{V_i, V_{-i}}(p_0, p_1) d p_0 d p_1, \\
        &\mathbb{E}\big[(Y_i^{a'} - Y_i^a) \mathbbm{1} \{P_i^{a'} \geq P_i^a, P_{-i}^{a'} < P_{-i}^a\} \big] \\
        =& \int_{0}^{1} \int_{0}^{1} \bigg\{ \operatorname{MCDE_i}(0; p_0, p_1) \mathbb{P}\big(P_i^a < p_0 \leq P_i^{a'}, p_1 > P_{-i}^{a} > P_{-i}^{a'}\big) \\
        &- \operatorname{MCSE}_i(0; p_0, p_1) \mathbb{P}\big(p_0 > P_i^{a'} \geq P_i^a, P_{-i}^{a'} < p_1 \leq P_{-i}^{a}\big) \\
        &+ \operatorname{MCDE}_i(1; p_0, p_1) \mathbb{P}\big(P_i^a < p_0 \leq P_i^{a'}, p_1 \leq P_{-i}^{a'} < P_{-i}^{a}\big) \\
        &- \operatorname{MCSE}_i(1; p_0, p_1) \mathbb{P}\big(p_0 \leq P_i^{a} \leq P_i^{a'}, P_{-i}^{a'} < p_1 \leq P_{-i}^{a}\big) \\
        &+ \big(\operatorname{MCDE}_i(0;p_0, p_1) - \operatorname{MCSE}_i(0;p_0,p_1) \big) \mathbb{P}\big(P_i^a < p_0 \leq P_i^{a'}, P_{-i}^{a'} < p_1 \leq P_{-i}^{a}\big) \bigg\} c_{V_i, V_{-i}}(p_0, p_1) d p_0 d p_1, \\
        &\mathbb{E}\big[(Y_i^{a'} - Y_i^a) \mathbbm{1} \{P_i^{a'} < P_i^a, P_{-i}^{a'} \geq P_{-i}^a\} \big] \\
        =& \int_{0}^{1} \int_{0}^{1} \bigg\{ -\operatorname{MCDE_i}(0; p_0, p_1) \mathbb{P}\big(P_i^{a'} < p_0 \leq P_i^{a}, p_1 > P_{-i}^{a'} \geq P_{-i}^{a}\big) \\
        &+ \operatorname{MCSE}_i(0; p_0, p_1) \mathbb{P}\big(p_0 > P_i^{a} > P_i^{a'}, P_{-i}^{a} < p_1 \leq P_{-i}^{a'}\big) \\
        &- \operatorname{MCDE}_i(1; p_0, p_1) \mathbb{P}\big(P_i^{a'} < p_0 \leq P_i^{a}, p_1 \leq P_{-i}^{a} \leq P_{-i}^{a'}\big) \\
        &+ \operatorname{MCSE}_i(1; p_0, p_1) \mathbb{P}\big(p_0 \leq P_i^{a'} < P_i^{a}, P_{-i}^{a} < p_1 \leq P_{-i}^{a'}\big) \\
        &+ \big(-\operatorname{MCDE}_i(0;p_0, p_1) + \operatorname{MCSE}_i(0;p_0,p_1) \big) \mathbb{P}\big(P_i^{a'} < p_0 \leq P_i^{a}, P_{-i}^{a} < p_1 \leq P_{-i}^{a'}\big) \bigg\} c_{V_i, V_{-i}}(p_0, p_1) d p_0 d p_1, \\
        &\mathbb{E}\big[(Y_i^{a'} - Y_i^a) \mathbbm{1} \{P_i^{a'} < P_i^a, P_{-i}^{a'} < P_{-i}^a\} \big] \\
        =& \int_{0}^{1} \int_{0}^{1} \bigg\{ -\operatorname{MCDE_i}(0; p_0, p_1) \mathbb{P}\big(P_i^{a'} < p_0 \leq P_i^{a}, p_1 > P_{-i}^{a} > P_{-i}^{a'}\big) \\
        &- \operatorname{MCSE}_i(0; p_0, p_1) \mathbb{P}\big(p_0 > P_i^{a} > P_i^{a'}, P_{-i}^{a'} < p_1 \leq P_{-i}^{a}\big) \\
        &- \operatorname{MCDE}_i(1; p_0, p_1) \mathbb{P}\big(P_i^{a'} < p_0 \leq P_i^{a}, p_1 \leq P_{-i}^{a'} < P_{-i}^{a}\big) \\
        &- \operatorname{MCSE}_i(1; p_0, p_1) \mathbb{P}\big(p_0 \leq P_i^{a'} < P_i^{a}, P_{-i}^{a'} < p_1 \leq P_{-i}^{a}\big) \\
        &- \big(\operatorname{MCDE}_i(1;p_0, p_1) + \operatorname{MCSE}_i(0;p_0,p_1) \big) \mathbb{P}\big(P_i^{a'} < p_0 \leq P_i^{a}, P_{-i}^{a'} < p_1 \leq P_{-i}^{a}\big) \bigg\} c_{V_i, V_{-i}}(p_0, p_1) d p_0 d p_1,
    \end{aligned}
\end{equation*}
which can be identified once we identify the MCDEs, MCSEs, and the copula density of $(V_i, V_{-i})$.

Finally, we can identify the PRTE in this case as $\mathbb{E}[Y_i^{a'}-Y_i^a] / \Delta P$, where 
\begin{equation*}
    \begin{aligned}
        & \Delta P = \Delta P_1 + \Delta P_2 + \Delta P_3 + \Delta P_4, \\
        &\Delta P_1 = \int_{0}^{1} \int_{0}^{1} \bigg\{ \mathbb{P}\big(P_i^a < p_0 \leq P_i^{a'}, p_1 > P_{-i}^{a'} \geq P_{-i}^{a}\big) + \mathbb{P}\big(p_0 > P_i^{a'} \geq P_i^a, P_{-i}^{a} < p_1 \leq P_{-i}^{a'}\big) \\
        &+ \mathbb{P}\big(P_i^a < p_0 \leq P_i^{a'}, p_1 \leq P_{-i}^{a} \leq P_{-i}^{a'}\big) + \mathbb{P}\big(p_0 \leq P_i^{a} \leq P_i^{a'}, P_{-i}^{a} < p_1 \leq P_{-i}^{a'}\big) \\
        &+ \mathbb{P}\big(P_i^a < p_0 \leq P_i^{a'}, P_{-i}^{a} < p_1 \leq P_{-i}^{a'}\big) \bigg\} c_{V_i, V_{-i}}(p_0, p_1) d p_0 d p_1, \\
        &\Delta P_2 = \int_{0}^{1} \int_{0}^{1} \bigg\{ \mathbb{P}\big(P_i^a < p_0 \leq P_i^{a'}, p_1 > P_{-i}^{a} > P_{-i}^{a'}\big) + \mathbb{P}\big(p_0 > P_i^{a'} \geq P_i^a, P_{-i}^{a'} < p_1 \leq P_{-i}^{a}\big) \\
        &+ \mathbb{P}\big(P_i^a < p_0 \leq P_i^{a'}, p_1 \leq P_{-i}^{a'} < P_{-i}^{a}\big) + \mathbb{P}\big(p_0 \leq P_i^{a} \leq P_i^{a'}, P_{-i}^{a'} < p_1 \leq P_{-i}^{a}\big) \\
        &+ \mathbb{P}\big(P_i^a < p_0 \leq P_i^{a'}, P_{-i}^{a'} < p_1 \leq P_{-i}^{a}\big) \bigg\} c_{V_i, V_{-i}}(p_0, p_1) d p_0 d p_1, \\
        &\Delta P_3 = \int_{0}^{1} \int_{0}^{1} \bigg\{ \mathbb{P}\big(P_i^{a'} < p_0 \leq P_i^{a}, p_1 > P_{-i}^{a'} \geq P_{-i}^{a}\big) + \mathbb{P}\big(p_0 > P_i^{a} > P_i^{a'}, P_{-i}^{a} < p_1 \leq P_{-i}^{a'}\big) \\
        &+\mathbb{P}\big(P_i^{a'} < p_0 \leq P_i^{a}, p_1 \leq P_{-i}^{a} \leq P_{-i}^{a'}\big) + \mathbb{P}\big(p_0 \leq P_i^{a'} < P_i^{a}, P_{-i}^{a} < p_1 \leq P_{-i}^{a'}\big) \\
        &+ \mathbb{P}\big(P_i^{a'} < p_0 \leq P_i^{a}, P_{-i}^{a} < p_1 \leq P_{-i}^{a'}\big) \bigg\} c_{V_i, V_{-i}}(p_0, p_1) d p_0 d p_1, \\
        &\Delta P_4 = \int_{0}^{1} \int_{0}^{1} \bigg\{ \mathbb{P}\big(P_i^{a'} < p_0 \leq P_i^{a}, p_1 > P_{-i}^{a} > P_{-i}^{a'}\big) + \mathbb{P}\big(p_0 > P_i^{a} > P_i^{a'}, P_{-i}^{a'} < p_1 \leq P_{-i}^{a}\big) \\
        &+\mathbb{P}\big(P_i^{a'} < p_0 \leq P_i^{a}, p_1 \leq P_{-i}^{a'} < P_{-i}^{a}\big) + \mathbb{P}\big(p_0 \leq P_i^{a'} < P_i^{a}, P_{-i}^{a'} < p_1 \leq P_{-i}^{a}\big) \\
        &+ \mathbb{P}\big(P_i^{a'} < p_0 \leq P_i^{a}, P_{-i}^{a'} < p_1 \leq P_{-i}^{a}\big) \bigg\} c_{V_i, V_{-i}}(p_0, p_1) d p_0 d p_1.
    \end{aligned}
\end{equation*}

\section{Comparison With Relevant Literature}

\subsection{Breakdown of MTE causal validity: Proof} \label{app:mte_breakdown}

We can express $\mathbb{E}\left[Y_i D_i \mid p_i\left(Z_i\right)=p_0\right]$ as 
\begin{equation*}
    \begin{aligned}
    & \mathbb{E}\left[Y_i D_i \mid P_i(Z_i) = p_0\right] \\
    =& \mathbb{E}\left[\mathbb{E}\left[Y_i D_i \mid P_i(Z_i) = p_0, P_{-i}(Z_{-i}) = p_1\right] \mid P_i(Z_i) = p_0\right]
    \end{aligned}
\end{equation*}
by applying the law of iterated expectations. The inner conditional expectation can be further expressed as 
\begin{equation*}
    \begin{aligned}
        &\mathbb{E}\left[Y_i D_i D_{-i} \mid P_i(Z_i) = p_0, P_{-i}(Z_{-i}) = p_1\right] \\
        &+ \mathbb{E}\left[Y_i D_i (1 - D_{-i}) \mid P_i(Z_i) = p_0, P_{-i}(Z_{-i}) = p_1\right] \\
        =&\mathbb{E}\left[Y_i(1, 1) \mathbbm{1}\{V_i \leq h_i(Z_i, Z_{-i})\} \mathbbm{1}\{V_{-i} \leq h_{-i}(Z_{-i}, Z_i)\} \mid h_i(Z_i, Z_{-i}) = p_0, h_{-i}(Z_{-i}, Z_i) = p_1\right] \\
        &+ \mathbb{E}\left[Y_i(1, 0) \mathbbm{1}\{V_i \leq h_i(Z_i, Z_{-i})\} \mathbbm{1}\{V_{-i} > h_{-i}(Z_{-i}, Z_i)\} \mid P_i(Z_i, Z_{-i}) = p_0, P_{-i}(Z_{-i}, Z_i) = p_1\right] \\
        =&\mathbb{E}\left[Y_i(1, 1) \mathbbm{1}\{V_i \leq p_0\} \mathbbm{1}\{V_{-i} \leq p_1\}\right] + \mathbb{E}\left[Y_i(1, 0) \mathbbm{1}\{V_i \leq p_0\} \mathbbm{1}\{V_{-i} > p_1\}\right] \\
        =& \int_{0}^{p_1} \int_{0}^{p_0} \mathbb{E}\left[Y_i(1, 1) \mid V_i = v_0, V_{-i} = v_1\right] c_{V_i, V_{-i}}(v_0, v_1) dv_0 dv_1 \\
        &+ \int_{p_1}^{1} \int_{0}^{p_0} \mathbb{E}\left[Y_i(1, 0) \mid V_i = v_0, V_{-i} = v_1\right] c_{V_i, V_{-i}}(v_0, v_1) dv_0 dv_1.
    \end{aligned}
\end{equation*}

Therefore, 
\begin{equation*}
    \begin{aligned}
        &\mathbb{E}\left[Y_i D_i \mid P_i(Z_i) = p_0\right] \\
        =& \int_{0}^{1} \int_{0}^{p_1} \int_{0}^{p_0} \mathbb{E}\left[Y_i(1, 1) \mid V_i = v_0, V_{-i} = v_1\right] c_{V_i, V_{-i}}(v_0, v_1) f_{P_{-i} \mid P_i = p_0} (p_1) dv_0 dv_1 dp_1 \\
        &+ \int_{0}^{1} \int_{p_1}^{1} \int_{0}^{p_0} \mathbb{E}\left[Y_i(1, 0) \mid V_i = v_0, V_{-i} = v_1\right] c_{V_i, V_{-i}}(v_0, v_1) f_{P_{-i} \mid P_i = p_0} (p_1) dv_0 dv_1 dp_1,
    \end{aligned}
\end{equation*}
where $f_{P_{-i} \mid P_i = p_0}(\cdot)$ denotes the conditional density of propensity score function $P_{-i}$ given $P_i = p_0$. %$f_{P_{-i} \mid P_i = p_0}(\cdot)$ does not equal to 0, since $P_i$ and $P_{-i}$ should be dependent as they are both functions of $(Z_i, Z_{-i})$. Therefore, $f_{P_{-i} \mid P_i = p_0}(\cdot)$ should be nondegenerate functions of $p_0$.
If $Y_i$ is bounded, i.e., $|Y_i| < \infty$, then by Fubini's theorem, we can interchange the order of integration in the expression above, yielding the following result.
\begin{equation*}
    \begin{aligned}
        &\mathbb{E}\left[Y_i D_i \mid P_i(Z_i) = p_0\right] \\
        =& \int_{0}^{p_0} \int_{0}^{1} \int_{0}^{p_1} \mathbb{E}\left[Y_i(1, 1) \mid V_i = v_0, V_{-i} = v_1\right] c_{V_i, V_{-i}}(v_0, v_1) f_{P_{-i} \mid P_i = p_0} (p_1) dv_1 dp_1 dv_0 \\
        &+ \int_{0}^{p_0} \int_{0}^{1} \int_{p_1}^{1} \mathbb{E}\left[Y_i(1, 0) \mid V_i = v_0, V_{-i} = v_1\right] c_{V_i, V_{-i}}(v_0, v_1) f_{P_{-i} \mid P_i = p_0} (p_1) dv_1 dp_1 dv_0.
    \end{aligned}
\end{equation*}
Suppose the above function is continuously differentiable with respect to $p_0$. In that case, we can apply the Leibniz integral rule to differentiate and obtain the following equalities,
\begin{equation*}
    \begin{aligned}
        & \frac{\partial}{\partial p_0} \int_{0}^{p_0} \int_{0}^{1} \int_{0}^{p_1} \mathbb{E}\left[Y_i(1, 1) \mid V_i = v_0, V_{-i} = v_1\right] c_{V_i, V_{-i}}(v_0, v_1) f_{P_{-i} \mid P_i = p_0} (p_1) dv_1 dp_1 dv_0 \\
        =& \int_{0}^{1} \int_{0}^{p_1} \mathbb{E}\left[Y_i(1, 1) \mid V_i = p_0, V_{-i} = v_1\right] c_{V_i, V_{-i}}(p_0, v_1) f_{P_{-i} \mid P_i = p_0} (p_1) dv_1 dp_1 \\
        &+ \int_{0}^{p_0} \int_{0}^{1} \int_{0}^{p_1} \mathbb{E}\left[Y_i(1, 1) \mid V_i = v_0, V_{-i} = v_1\right] 
        c_{V_i, V_{-i}}(v_0, v_1) \frac{\partial}{\partial p_0} f_{P_{-i} \mid P_i = p_0} (p_1) dv_1 dp_1 dv_0 \\
    \equiv& \int_{0}^{1} \int_{0}^{p_1} \mathbb{E}\left[Y_i(1, 1) \mid V_i = p_0, V_{-i} = v_1\right] c_{V_i, V_{-i}}(p_0, v_1) f_{P_{-i} \mid P_i = p_0} (p_1) dv_1 dp_1 + \mathcal{R}_{11}, \\
    & \frac{\partial}{\partial p_0} \int_{0}^{p_0} \int_{0}^{1} \int_{p_1}^{1} \mathbb{E}\left[Y_i(1, 0) \mid V_i = v_0, V_{-i} = v_1\right] c_{V_i, V_{-i}}(v_0, v_1) f_{P_{-i} \mid P_i = p_0} (p_1) dv_1 dp_1 dv_0 \\
        =& \int_{0}^{1} \int_{p_1}^{1} \mathbb{E}\left[Y_i(1, 0) \mid V_i = p_0, V_{-i} = v_1\right] c_{V_i, V_{-i}}(p_0, v_1) f_{P_{-i} \mid P_i = p_0} (p_1) dv_1 dp_1 \\
        &+ \int_{0}^{p_0} \int_{0}^{1} \int_{p_1}^{1} \mathbb{E}\left[Y_i(1, 0) \mid V_i = v_0, V_{-i} = v_1\right] c_{V_i, V_{-i}}(v_0, v_1) \frac{\partial}{\partial p_0} f_{P_{-i} \mid P_i = p_0} (p_1) dv_1 dp_1 dv_0 \\
    \equiv& \int_{0}^{1} \int_{p_1}^{1} \mathbb{E}\left[Y_i(1, 0) \mid V_i = p_0, V_{-i} = v_1\right] c_{V_i, V_{-i}}(p_0, v_1) f_{P_{-i} \mid P_i = p_0} (p_1) dv_1 dp_1 + \mathcal{R}_{10}.
    \end{aligned}
\end{equation*}
The terms $\mathcal{R}_{11}$ and $\mathcal{R}_{10}$ are nonzero because $c_{V_i, V_{-i}}(v_0, v_1) \neq 0$, and $\partial f_{P_{-i} \mid P_i = p_0} (p_1) / \partial p_0 \neq 0$ given that $P_i$ and $P_{-i}$ are dependent as they are both functions of $(Z_i, Z_{-i})$.

Similarly, under the assumption that $Y_i$ is bounded and that $\mathbb{E}\left[Y_i (1 - D_i) \mid P_i(Z_i) = p_0 \right]$ is continuously differentiable with respect to $p_0$, we obtain the following equalities,
\begin{equation*}
    \begin{aligned}
        &\frac{\partial}{\partial p_0} \mathbb{E}\left[Y_i(1 - D_i) \mid P_i(Z_i) = p_0\right] \\
        =& -\int_{0}^{1} \int_{0}^{p_1} \mathbb{E}\left[Y_i(0, 1) \mid V_i = p_0, V_{-i} = v_1\right] c_{V_i, V_{-i}}(p_0, v_1) f_{P_{-i} \mid P_i = p_0} (p_1) dv_1 dp_1 + \mathcal{R}_{01} \\
        &- \int_{0}^{1} \int_{p_1}^{1} \mathbb{E}\left[Y_i(0, 0) \mid V_i = p_0, V_{-i}= v_1\right] c_{V_i, V_{-i}}(p_0, v_1) f_{P_{-i} \mid P_i = p_0} (p_1) dv_1 dp_1 + \mathcal{R}_{00}, \\
    & \mathcal{R}_{01} = \int_{p_0}^{1} \int_{0}^{1} \int_{0}^{p_1} \mathbb{E}\left[Y_i(0,1) \mid V_i = v_0, V_{-i} = v_1\right] c_{V_i, V_{-i}}(v_0, v_1) \frac{\partial}{\partial p_0} f_{P_{-i} \mid P_i = p_0} (p_1) dv_1 dp_1 dv_0, \\
    & \mathcal{R}_{00} = \int_{p_0}^{1} \int_{0}^{1} \int_{p_1}^{1} \mathbb{E}\left[Y_i(0, 0) \mid V_i = v_0, V_{-i} = v_1\right] c_{V_i, V_{-i}}(v_0, v_1) \frac{\partial}{\partial p_0} f_{P_{-i} \mid P_i = p_0} (p_1) dv_1 dp_1 dv_0.
    \end{aligned}
\end{equation*}

By taking the difference between $\partial \mathbb{E}[Y_i D_i \mid 
P_i(Z_i) = p_0] / \partial p_0$ and $-\partial \mathbb{E}[Y_i (1 - D_i) \mid 
P_i(Z_i) = p_0] / \partial p_0$, we would identify 
\begin{equation*}
    \begin{aligned}
        &\int_{0}^{1} \int_{0}^{p_1} \mathbb{E}\left[Y_i(1, 1) - Y_i(0,1) \mid V_i = p_0, V_{-i} = v_1\right] c_{V_i, V_{-i}}(p_0, v_1) f_{P_{-i} \mid P_i = p_0} (p_1) dv_1 dp_1 \\
        +&\int_{0}^{1} \int_{p_1}^{1} \mathbb{E}\left[Y_i(1, 0) - Y_i(0, 0) \mid V_i = p_0, V_{-i} = v_1\right] c_{V_i, V_{-i}}(p_0, v_1) f_{P_{-i} \mid P_i = p_0} (p_1) dv_1 dp_1 \\
        &+ \mathcal{R},
    \end{aligned}
\end{equation*}
where the bias term $\mathcal{R}$ is given by the sum of its components $\mathcal{R}_{dd'}$, where $d, d' \in \{0,1\}$, as defined in the preceding derivations.

The first two terms correspond to the averages of the marginal controlled direct effects for unit $i$, $\operatorname{MCDE}i^{(1)}(p_0, v_1)$ and $\operatorname{MCDE}i^{(0)}(p_0, v_1)$, weighted by the peer's propensity scores conditional on $P_i = p_0$ and the copula density of $(V_i, V{-i})$. Since the bias terms, $\mathcal{R}$, involves the derivative
\begin{equation*}
    \frac{\partial}{\partial p_0} f_{P_{-i} \mid P_i = p_0}(p_1),
\end{equation*}
it vanishes whenever this derivative equals zero for all $p_1$, i.e., when $f_{P_{-i} \mid P_i = p_0}(p_1)$ does not depend on $p_0$. This holds when the propensity scores $P_i$ and $P_{-i}$ are independent, which occurs under two sufficient conditions:
\begin{enumerate}
    \item The instruments $Z_i$ and $Z_{-i}$ are independent within each group.
    \item The threshold function for unit $i$ depends solely on its own instrument, $h_i(Z_i, Z_{-i}) = h_i(Z_i)$.
\end{enumerate}

Under these conditions, since the propensity score identifies the threshold function, it follows that $P_i(Z_i, Z_{-i}) = h_i(Z_i)$ and $P_{-i}(Z_{-i}, Z_i) = h_{-i}(Z_{-i})$ almost surely. Since $Z_i$ is assumed to be independent of $Z_{-i}$, $P_i$ and $P_{-i}$ are independent, implying that the bias term $\mathcal{R}$ equals zero.

\subsection{Marginal Controlled Effects in the Absence of Spillovers} \label{app:mce_sutva}

If spillover effects are absent in both income and treatment selection, which means that $V_i \indep V_{-i}$, $Y_{i}(D_{i}, d)=Y_{i}(D_{i}, d') \equiv Y_{i}(D_{i})$, and $h_{i}(Z_{i}, z)=h_{i}(Z_{i}, z') \equiv h_{i}(Z_{i})$, then our identification results reduce to the standard MTE framework. 

In this case, the propensity score identifies
\begin{equation*}
    \begin{aligned}
        & \mathbb{P}\big(D_i=1 \mid Z_i=z_0, Z_{-i}=z_1\big) \\
        =& \mathbb{P}\big(V_i \leq h_i(z_0) \mid Z_i=z_0, Z_{-i}=z_1\big) \\
        =& \mathbb{P}\big(V_i \leq h_i(z_0)\big) = h_i(z_0),
    \end{aligned}
\end{equation*}
which is the threshold function from the standard MTE setting, where peer instruments do not influence individual treatment decisions.

Taking the cross-partial derivative of $\mathbb{P}(D_i = 1, D_{-i} = 1 \mid P_i = p_0, P_{-i} = p_1)$, the copula density evaluated at the point $(p_0, p_1)$, $c_{V_i, V_{-i}}(p_0, p_1)$, is identified as
\begin{equation*}
    \begin{aligned}
        & \frac{\partial^2}{\partial p_1 \partial p_0} \mathbb{P}(D_i = 1, D_{-i} = 1 \mid P_i = p_0, P_{-i} = p_1) \\
        =& \frac{\partial^2}{\partial p_1 \partial p_0} \mathbb{P}(V_i \leq p_0, V_{-i} \leq p_1 \mid P_i = p_0, P_{-i} = p_1) \\
        =& \frac{\partial^2}{\partial p_1 \partial p_0} \mathbb{P}\big(V_i \leq p_0, V_{-i} \leq p_1\big) \\
        =& \frac{\partial}{\partial p_0}\mathbb{P}\big(V_i \leq p_0\big) \cdot \frac{\partial}{\partial p_1}\mathbb{P}\big(V_{-i} \leq p_1\big) = 1.
    \end{aligned}
\end{equation*}
The copula density equals one since $V_i$ and $V_{-i}$ are independent. This aligns with the standard MTE framework, where the individuals' unobserved heterogeneities are independent.

Finally, by taking the cross-partial derivative of $\mathbb{E}[Y_{i} D_{i} D_{-i} \mid P_{i}=p_0, P_{-i}=p_1]$, 
\begin{equation*}
    \begin{aligned}
        &\frac{\partial^2}{\partial p_0 \partial p_1} \mathbb{E}\big[Y_{i} D_{i} D_{-i} \mid P_{i}=p_0, P_{-i}=p_1\big] \\
        =& \frac{\partial^2}{\partial p_0 \partial p_1} \mathbb{E}\big[Y_{i} \mathbbm{1} \{V_i \leq p_0\} \mathbbm{1} \{V_{-i} \leq p_1\} \mid P_{i}=p_0, P_{-i}=p_1\big] \\
        =& \frac{\partial^2}{\partial p_0 \partial p_1} \mathbb{E}\big[Y_{i} \mathbbm{1} \{V_i \leq p_0\} \mathbbm{1} \{V_{-i} \leq p_1\}\big] \\
        =& \frac{\partial^2}{\partial p_0 \partial p_1} \int_{0}^{p_1} \int_{0}^{p_0} \mathbb{E}\big[Y_i(1) \mid V_i = v_0, V_{-i} = v_1\big] c_{V_i, V_{-i}}(v_0, v_1) d v_0 d v_1 \\
        =& \frac{\partial^2}{\partial p_0 \partial p_1} \int_{0}^{p_1} \int_{0}^{p_0} \mathbb{E}\big[Y_i(1) \mid V_i = v_0\big] d v_0 d v_1 \\
        =& \mathbb{E}\big[Y_i(1) \mid V_i = p_0\big],
    \end{aligned}
\end{equation*}
where the third line follows from Assumption \ref{as:RA}, the fifth line holds because $V_{-i}$ is independent with $(Y_i(d), V_i)$, and the copula density equals to one. By analogous reasoning, taking the cross-partial derivatives of $\mathbb{E}[Y_{i} \mathbbm{1}\{D_{i} = d, D_{i} = d\} \mid P_{i}=p_0, P_{-i}=p_1]$ identifies 
\begin{equation*}
    \begin{aligned}
        &\frac{\partial^2}{\partial p_0 \partial p_1} \mathbb{E}\big[Y_{i} \mathbbm{1}\{D_{i} = d, D_{i} = d\} \mid P_{i}=p_0, P_{-i}=p_1\big] \\
        =& \mathbb{E}\big[Y_i(d) \mid V_i = p_0\big],
    \end{aligned}
\end{equation*}
and the cross-partial derivatives of $\mathbb{E}[Y_{i} \mathbbm{1}\{D_{i} = d, D_{-i} = 1-d\} \mid P_{i}=p_0, P_{-i}=p_1]$ identifies 
\begin{equation*}
    \begin{aligned}
        &\frac{\partial^2}{\partial p_0 \partial p_1} \mathbb{E}\big[Y_{i} \mathbbm{1}\{D_{i} = d, D_{-i} = 1-d\} \mid P_{i}=p_0, P_{-i}=p_1\big] \\
        =& -\mathbb{E}\big[Y_i(d) \mid V_i = p_0\big],
    \end{aligned}
\end{equation*}
for $d \in \{0, 1\}$ and $(p_0, p_1)$ being an interior point of $\mathcal{P}$. Then, it follows that the $\operatorname{MCSE}_i^{(d)}(p_0, p_1)$ is identified as 
\begin{equation*}
    \operatorname{sgn}(2 d-1) \cdot \frac{\partial^2 \mathbb{E}\left[Y_i \mathbbm{1}\{D_i=d\} \mid P_i=p_0, P_{-i}=p_1\right]}{\partial p_0 \partial p_1} = 0,
\end{equation*}
and the $\operatorname{MCDE}_i^{(d)}(p_0, p_1)$ is identified as 
\begin{equation*}
    \operatorname{sgn}(2 d-1) \cdot \frac{\partial^2 \mathbb{E}\left[Y_i \mathbbm{1}\{D_{-i}=d\} \mid P_i=p_0, P_{-i}=p_1\right]}{\partial p_0 \partial p_1} = \mathbb{E}\big[Y_i(1) - Y_i(0) \mid V_i = p_0\big].
\end{equation*}

One may interpret the model in Equation \eqref{eq:basic_model_disc_trt} as a standard MTE framework, treating $(D_{-i}, Z_{-i})$ as covariates. This raises a natural question: can the marginal treatment response (MTR) functions be identified using conventional MTE methods by conditioning on the peer's treatment status $D_{-i}$? The answer is affirmative, provided that $D_{-i}$ is exogenous with respect to unit $i$'s potential outcomes and unobserved heterogeneity. Specifically, this requires the independence assumption 
\begin{equation} \label{eq:indep_peer_trt}
    (Z_i, Z_{-i}, D_{-i}) \indep \{(V_i, Y_{i}(d, d'))\}_{d \in \{0,1\}, d' \in \{0,1\}}.
\end{equation}
When this condition holds, the marginal treatment response functions conditional on the individual's own unobserved heterogeneity, $\mathbb{E}[Y_i(d, d') \mid V_i ]$, can be identified using the standard MTE identification strategy, extended to include $D_{-i}$ as an additional covariate. For instance,
\begin{equation*}
    \begin{aligned}
         & \frac{\partial}{\partial p} \mathbb{E}\left[Y_i \cdot D_i \mid P_i(Z_i, Z_{-i}) = p, D_{-i} = d'\right] \\
         =& \frac{\partial}{\partial p} \mathbb{E}\left[Y_i(1, d') \cdot \mathbbm{1}\{V_i \leq p\} \mid P_i(Z_i, Z_{-i}) = p, D_{-i} = d'\right] \\
         =& \frac{\partial}{\partial p} \int_{0}^{p} \mathbb{E}\left[Y_i(1, d') \mid V_i = v\right] dv \\
         =& \mathbb{E}\left[Y_i(1, d') \mid V_i = p\right].
    \end{aligned}
\end{equation*}
It is important to note that the second equality in the above derivation holds only if $D_{-i}$ satisfies the independence condition in Equation \eqref{eq:indep_peer_trt}.

However, the assumption of exogeneity for $D_{-i}$ is not realistic within our spillover framework. As illustrated in Figure \ref{fig:causal_relation}, the peer's unobserved heterogeneity $V_{-i}$ affects her treatment $D_{-i}$ and may also be correlated with unit $i$'s unobserved confounder $V_i$, thereby threatening the exogeneity of $D_{-i}$. In the marginal spillover setting, where $V_i$ and $V_{-i}$ are allowed to be arbitrarily correlated, applying the standard MTE identification strategy by conditioning on $D_{-i}$ would lead to
\begin{equation*}
    \begin{aligned}
        & \frac{\partial}{\partial p} \mathbb{E}\left[Y_i \cdot D_i \mid P_i(Z_i, Z_{-i}) = p, D_{-i} = d'\right] \\
         =& \frac{\partial}{\partial p} \mathbb{E}\left[Y_i(1, d') \cdot \mathbbm{1}\{V_i \leq p\} \mid P_i(Z_i, Z_{-i}) = p, D_{-i} = d'\right] \\
         =& \frac{\partial}{\partial p} \int_{0}^{p} \mathbb{E}\left[Y_i(1, d') \mid V_i = v, D_{-i} = 1\right] f_{V_i \mid D_{-i} = d}(v) dv \\
         =& \mathbb{E}\left[Y_i(1, d') \mid V_i = p, D_{-i} = d'\right],
    \end{aligned}
\end{equation*}
with the distribution of $V_i \mid D_{-i} = d$ being normalized to follow a uniform distribution on $[0, 1]$. However, $\mathbb{E}\left[Y_i(1, d') \mid V_i = p, D_{-i} = d'\right]$ is not equal to $\mathbb{E}\left[Y_i(1, d') \mid V_i = p\right]$, because $D_{-i}$ is determined by $V_{-i}$, which is dependent on $Y_i(d, d')$ even conditioning on $V_i$, as illustrated in Figure \ref{fig:causal_relation}.

\subsection{Comparing With Multivalued Treatments Literature} \label{app:multi_trt}

In this section, we compare the methods for identifying marginal spillover effects with the framework discussed in \cite{lee2018identifying}. We focus on the group level and consider the treatment vector $ \boldsymbol{D}_g \equiv (D_{0g}, D_{1g}) $, assigned to each group $ g $. The treatment vector $ \boldsymbol{D}_g $ takes values from the set $ \{(1,1), (1,0), (0,1), (0,0)\} $, consisting of four elements. Therefore, $ \boldsymbol{D}_g $ can be regarded as multivalued treatments assigned at the group level. To align with the notation in \cite{lee2018identifying}, we relabel the treatment vectors as follows: $ (0,0) \equiv 0 $, $ (0,1) \equiv 1 $, $ (1,0) \equiv 2 $, and $ (1,1) \equiv 3 $. Consequently, the treatment $ \boldsymbol{D}_g $ takes values from $ \{0, 1, 2, 3\} \equiv \mathcal{D} $. 

Each group is randomly assigned a continuous instrumental variable, with the instrument vector for group $ g $ denoted as $ \boldsymbol{Z}_g \equiv (Z_{0g}, Z_{1g}) $. Let $ V_{ig} \in \mathbb{R} $ represent the unobserved characteristics of individual $ i $ in group $ g $, and let $ \boldsymbol{V}_g \equiv (V_{0g}, V_{1g}) $ denote the vector of unobserved heterogeneity for both individuals in group $ g $. For simplicity, we omit the group subscript $ g $ from the notation. The parameters of interest in \cite{lee2018identifying}, $ E(Y_k \mid \boldsymbol{V} = \boldsymbol{v}) - E(Y_{k'} \mid \boldsymbol{V} = \boldsymbol{v}) $, where $ k \neq k' $ and $ k, k' \in \{0, 1, 2, 3\} $, can be interpreted as marginal controlled spillover effects and marginal controlled direct effects within the spillover framework.

According to the model in Equation \eqref{eq:basic_model_disc_trt}, we have 
\begin{enumerate}
    \item $\boldsymbol{D} = 0$ if and only if $V_0 > h_0(\boldsymbol{Z})$ and $V_1 > h_1(\boldsymbol{Z})$.
    \item $\boldsymbol{D} = 1$ if and only if $V_0 > h_0(\boldsymbol{Z})$ and $V_1 \leq h_1(\boldsymbol{Z})$.
    \item $\boldsymbol{D} = 2$ if and only if $V_0 \leq h_0(\boldsymbol{Z})$ and $V_1 > h_1(\boldsymbol{Z})$.
    \item $\boldsymbol{D} = 3$ if and only if $V_0 \leq h_0(\boldsymbol{Z})$ and $V_1 \leq h_1(\boldsymbol{Z})$.
\end{enumerate}

It is straightforward to see that the treatment $ \boldsymbol{D} $ is measurable with respect to the $ \sigma $-field generated by the events $ \left\{V_i < Q_i(\boldsymbol{Z})\right\} $ for $ i \in \{0,1\} $, which aligns with the selection mechanism described in Assumption 2.1 of \cite{lee2018identifying}. Furthermore, Theorem 3.1 in \cite{lee2009training} is similar to our approach in identifying the joint density of unobserved heterogeneity $ \boldsymbol{V} $ and the marginal treatment response functions, once the threshold functions $ h_i(\boldsymbol{Z}) $ (denoted as $ Q_i(\boldsymbol{Z}) $ in \cite{lee2018identifying}) are identified.

However, in our setting, fewer assumptions are needed to point identify the thresholds. Specifically, we only require that the instruments $ \boldsymbol{Z} $ are randomly assigned at the group level and do not directly influence the outcomes, without relying on the additional exclusion restrictions on instruments imposed in Assumption 4.1 of \cite{lee2018identifying}. This is because we have more information on the observed treatments $ D $, allowing us to identify the marginal distributions of $ V_0 $ and $ V_1 $ from the proportions of observed treatments:
\begin{equation*}
    \begin{aligned}
        &\mathbb{P}(\boldsymbol{D} = 3 \mid \boldsymbol{Z} = z) + \mathbb{P}(\boldsymbol{D} = 2 \mid \boldsymbol{Z} = z) = \mathbb{P}(V_0 \leq h_0(z)), \\
        &\mathbb{P}(\boldsymbol{D} = 3 \mid \boldsymbol{Z} = z) + \mathbb{P}(\boldsymbol{D} = 1 \mid \boldsymbol{Z} = z) = \mathbb{P}(V_1 \leq h_1(z)).
    \end{aligned}
\end{equation*}

\begin{remark}
    (Monotonicity for each unit) When we focus on each individual unit $i$ within a group, the monotonicity condition is satisfied. Specifically, consider any two vectors of instruments, denoted as $(z_0, z_1)$ and $(\tilde{z}_0, \tilde{z}_1)$, where $P_i(z_0, z_1) \leq P_i(\tilde{z}_0, \tilde{z}_1)$. Under the monotonicity assumption, this ordering of the propensity scores implies that the corresponding potential treatments satisfy $D_i(z_0, z_1) \leq D_i(\tilde{z}_0, \tilde{z}_1)$. However, if we treat the entire group as a single decision-making unit and reformulate the setting into a multivalued treatment model, the monotonicity assumption may no longer hold. For instance, in the two-way flow model discussed in \cite{lee2018identifying}, when the proportion of $D = 2$ changes, it is unclear whether the shift is driven by changes in $h_0$ or $h_1$. In other words, shifts in either  $h_0$ or $h_1$ can induce changes in the proportion of receiving a given treatment, making it impossible to distinguish between the two effects. As a result, the monotonicity condition is violated, and the marginal distributions of $V_0$ and $V_1$ cannot be identified.
%In a multi-valued treatment framework, treatment decisions for different individuals within the group may interact in complex ways that violate the monotonicity condition. For example, the treatment assignment for one individual might depend not only on their own instruments but also on the instruments or decisions of their peers within the group, leading to possible reversals in treatment ordering. Such interdependencies complicate the modeling of treatment effects and pose challenges for identification. As a result, the monotonicity assumption, which is valid when applied to individual units, may fail when extended to group-level treatment decisions.This potential failure of monotonicity in the group-level context introduces additional complexities to identify treatment effects. 
\end{remark}

In the multivalued treatment setting discussed in \cite{lee2018identifying}, if we have enough information on the observed treatment $D$ that allows us to identify the threshold $h_{j}(\boldsymbol{Z})$ for each $j \in \{1, \cdots, J\}$, then we can point identify the joint density of $\boldsymbol{V}$ and the marginal treatment response functions in Theorem 3.1 of \cite{lee2018identifying}. Specifically, for each dimension $j \in \{1, \cdots, J\}$ of the unobservable $\boldsymbol{V}$, we need a subset of the support of the treatments, $\mathcal{K}_j \subseteq \mathcal{K}$, $\mathcal{K} = \{0, \cdots, K-1\}$, such that
\begin{equation*}
    \sum_{k \in \mathcal{K}_j} \mathbb{P}\left(D = k \mid \boldsymbol{Z}\right) = \mathbb{P}\left(V_{j} \leq h_j(\boldsymbol{Z}) \mid \boldsymbol{Z}\right) = \mathbb{P}\left(V_j \leq h_j(\boldsymbol{Z})\right) =h_j(\boldsymbol{Z}).
\end{equation*}

If $\mathcal{K}_j$ that satisfies the above conditions does not exist for some $j$, we can still partially identify the threshold $h_j(\boldsymbol{Z})$. $h_j(\boldsymbol{Z})$ can be partially identified as 
\begin{equation*}
    \sum_{k \in \underline{\mathcal{K}}_j} \mathbb{P}\left(D = k \mid \boldsymbol{Z}\right) \leq \mathbb{P}\left(V_{j} \leq h_j(\boldsymbol{Z}) \mid \boldsymbol{Z}\right) = \mathbb{P}\left(V_j \leq h_j(\boldsymbol{Z})\right) = h_j(\boldsymbol{Z}) \leq \sum_{k \in \bar{\mathcal{K}_j}} \mathbb{P}\left(D = k \mid \boldsymbol{Z}\right),
\end{equation*}
where $\underline{\mathcal{K}}_j$ is the largest subset $\widetilde{\underline{\mathcal{K}}}_j$ of $\mathcal{K}$ such that 
\begin{equation*}
    \cup_{k \in \widetilde{\underline{\mathcal{K}}}_j} d_k^{-1}\{D = k\} \subseteq \{V_{j} \leq h_j(\boldsymbol{Z})\},
\end{equation*}
and $\bar{\mathcal{K}_j}$ is the smallest subset $\widetilde{\bar{\mathcal{K}_j}}$ of $\mathcal{K}$ such that 
\begin{equation*}
    \{V_{j} \leq h_j(\boldsymbol{Z})\} \subseteq \cup_{k \in \widetilde{\bar{\mathcal{K}_j}}} d_k^{-1}\{D = k\}.
\end{equation*}

For example, in the two-way flow model discussed in \cite{lee2018identifying}, we can partially identify $h_1(\boldsymbol{Z})$ and $h_2(\boldsymbol{Z})$ as 
\begin{equation*}
    \begin{aligned}
        & \mathbb{P}\left(D = 0 \mid \boldsymbol{Z}\right) \leq h_1(\boldsymbol{Z}) \leq \mathbb{P}\left(D = 0 \mid \boldsymbol{Z}\right) + \mathbb{P}\left(D = 2\mid \boldsymbol{Z}\right), \\
        & \mathbb{P}\left(D = 0 \mid \boldsymbol{Z}\right) \leq h_2(\boldsymbol{Z}) \leq \mathbb{P}\left(D = 0 \mid \boldsymbol{Z}\right) + \mathbb{P}\left(D = 2\mid \boldsymbol{Z}\right).
    \end{aligned}
\end{equation*}

\section{Derivations of Testable Implications} \label{app:testable_implication}

Corollary~\ref{corr:id_copula_density} establishes identification of the copula density. Because any valid copula density must be nonnegative, this result directly implies a testable inequality restriction:
\begin{equation*}
    \frac{\partial^2 \mathbb{E}\big[D_i D_{-i} \mid P_i=p_0, P_{-i}=p_1\big]}{\partial p_0 \partial p_1} \geq 0
\end{equation*}
if the cross-derivative correctly identifies the copula density. This constitutes one of the testable implications.

In addition, the proof of Theorem~\ref{thm:id_mtr_disc_trt} identifies the conditional joint distributions of potential outcomes, $\mathbb{P}(Y_i(d, d') \in A_1, Y_{-i}(d, d') \in A_2 \mid V_i = p_0, V_{-i} = p_1)$, weighted by the copula density $c_{V_i, V_{-i}}(p_0, p_1)$. By the nonnegativity of both probabilities and copula densities, this quantity must be nonnegative. Together with the inequality derived from the copula density, this yields a set of testable inequality restrictions that must hold for any Borel sets $A_1, A_2 \subseteq \mathcal{Y}$ and all $d \in \{0,1\}$:
\begin{equation*}
        \begin{aligned}
            &\frac{\partial^2}{\partial p_1 \partial p_0} \mathbb{E}\Big[\mathbbm{1} \big\{Y_i \in A_1, Y_{-i} \in A_2 \big\} \mathbbm{1}\big\{D_{i} = d, D_{-i} = d \big\} \mid P_{i} = p_0, P_{-i}= p_1 \Big] \\ 
            =& \mathbb{P}\big(Y_i(d, d) \in A_1, Y_{-i}(d, d) \in A_2 \mid V_i = p_0, V_{-i} = p_1\big) c_{V_i, V_{-i}}(p_0, p_1) \geq 0, \\
            &-\frac{\partial^2}{\partial p_1 \partial p_0} \mathbb{E}\Big[\mathbbm{1} \big\{Y_i \in A_1, Y_{-i} \in A_2 \big\} \mathbbm{1}\big\{D_{i} = d, D_{-i} = 1-d \big\} \mid P_{i} = p_0, P_{-i} = p_1\Big] \geq 0 \\
            =& \mathbb{P}\big(Y_i(d, 1-d) \in A_1, Y_{-i}(d, 1-d) \in A_2 \mid V_i = p_0, V_{-i} = p_1\big) c_{V_i, V_{-i}}(p_0, p_1) \geq 0.
        \end{aligned}
\end{equation*}

It is worth noting that both the copula function and the marginal treatment response functions are functions of the propensity scores rather than the instrument values themselves. This feature implies an additional set of testable implications. In particular, if two distinct pairs of instrument values, $(z_0, z_1) \neq (\tilde{z}_0, \tilde{z}_1)$, yield the same propensity scores, $P_i(z_0, z_1) = P_i(\tilde{z}_0, \tilde{z}_1) = p_0$ and $P_{-i}(z_0, z_1) = P_{-i}(\tilde{z}_0, \tilde{z}_1) = p_1$, then any identified quantities that depend only on $(p_0, p_1)$, such as the copula density or the marginal treatment response functions, must be equal across these instrument values. This leads to the following testable equality
\begin{equation*}
    \begin{aligned}
        & \mathbb{P}\big(D_i = 1, D_{-i} = 1 \mid Z_i = z_0, Z_{-i} = z_1\big) \\
        =& \mathbb{P}\big(D_i = 1, D_{-i} = 1 \mid P_i(z_0, z_1) = p_0, P_{-i}(z_1, z_0) = p_1\big) \\
        =& \mathbb{P}\big(D_i = 1, D_{-i} = 1 \mid P_i(\tilde{z}_0, \tilde{z}_1) = p_0, P_{-i}(\tilde{z}_1, \tilde{z}_0) = p_1\big) \\
        =& \mathbb{P}\big(D_i = 1, D_{-i} = 1 \mid Z_i = \tilde{z}_0, Z_{-i} = \tilde{z}_1\big),
    \end{aligned}
\end{equation*}
since both conditional probabilities $\mathbb{P}(D_i = 1, D_{-i} = 1 \mid Z_i = z_0, Z_{-i} = z_1)$ and $\mathbb{P}(D_i = 1, D_{-i} = 1 \mid Z_i = \tilde{z}_0, Z_{-i} = \tilde{z}_1
)$ identify the copula function $C_{V_i, V_{-i}}(p_0, p_1)$. Similarly, 
    \begin{equation*}
        \begin{aligned}
            & \mathbb{E}\big[\mathbbm{1}\{Y_i \in A_1, Y_{-i} \in A_2\} \mathbbm{1}\{D_i = d_0, D_{-i} = d_1\} \mid Z_i = z_0, Z_{-i} = z_1\big] \\
            =& \mathbb{E}\big[\mathbbm{1}\{Y_i \in A_1, Y_{-i} \in A_2\} \mathbbm{1}\{D_i = d_0, D_{-i} = d_1\} \mid P_i(z_0, z_1)=p_0, P_{-i}(z_1, z_0)=p_1\big] \\
            =& \mathbb{E}\big[\mathbbm{1}\{Y_i \in A_1, Y_{-i} \in A_2\} \mathbbm{1}\{D_i = d_0, D_{-i} = d_1\} \mid P_i(\tilde{z}_0, \tilde{z}_1)=p_0, P_{-i}(\tilde{z}_1, \tilde{z}_0)=p_1\big] \\
            =& \mathbb{E}\big[\mathbbm{1}\{Y_i \in A_1, Y_{-i} \in A_2\} \mathbbm{1}\{D_i = d_0, D_{-i} = d_1\} \mid Z_i = \tilde{z}_0, Z_{-i} = \tilde{z}_1\big],
        \end{aligned}
    \end{equation*}
since both conditional means $\mathbb{E}[\mathbbm{1}\{Y_i \in A_1, Y_{-i} \in A_2\} \mathbbm{1}\{D_i = d_0, D_{-i} = d_1\} \mid Z_i = z_0, Z_{-i} = z_1]$ and $\mathbb{E}[\mathbbm{1}\{Y_i \in A_1, Y_{-i} \in A_2\} \mathbbm{1}\{D_i = d_0, D_{-i} = d_1\} \mid Z_i = \tilde{z}_0, Z_{-i} = \tilde{z}_1]$ identify the same integral of $\mathbb{P}(Y_i(d_0, d_1) \in A_1, Y_{-i}(d_0, d_1) \in A_2 \mid V_i=v_0, V_{-i}=v_1)$ weighted by the copula density $c_{V_i, V_{-i}}(v_0, v_1)$, according to the proof of Theorem \ref{thm:id_mtr_disc_trt}.

To conclude, these insights yield the sets of testable implications stated in Corollary~\ref{corr:test_implication}.

\section{Proof of asymptotic results}

\subsection{Convergence rate of nonparametric cross-derivative estimators} \label{app:cvg_rate_cross_deriv}

To simplify the notation, we introduce the following matrix definitions,
\begin{equation*}
    \begin{aligned}
        &\widehat{\mathbf{X}}_P \equiv \left[\begin{array}{cccccc}
            1 & \big(\hat{P}_{01}-p_0\big) & \cdots & \big(\hat{P}_{01}-p_0\big) \big(\hat{P}_{11}-p_1\big) & \cdots & \big(\hat{P}_{11}-p_1\big)^3\\
            \vdots & \vdots & \vdots & \vdots & \vdots & \vdots \\
            1 & \big(\hat{P}_{0G}-p_0\big) & \cdots & \big(\hat{P}_{0G}-p_0\big) \big(\hat{P}_{1G}-p_1\big) & \cdots & \big(\hat{P}_{1G}-p_1\big)^3,
            \end{array}\right] \\
        & \widehat{\mathbf{W}}_h \equiv \operatorname{diag} \Big(K_{h}\big(\hat{P}_1-p\big), \cdots, K_{h}\big(\hat{P}_G-p\big)\Big) \\
        & \mathbf{D} \equiv \big[D_{01}D_{11}, \cdots, D_{0G}D_{1G}\big]', \\
        & \widehat{\mathbf{U}}_{idd'} = \big[\hat{U}_{i d d' 1}, \cdots, \hat{U}_{i d d' G}\big]',
    \end{aligned}
\end{equation*}
where $\widehat{\mathbf{X}}_P$ is a $G \times 10$ matrix of regressors used in the local polynomial regression, $\widehat{\mathbf{W}}_h$ is a $G \times G$ diagonal matrix consisting of the kernel functions, $\mathbf{D}$ is a $G \times 1$ vector entries $D_{0g}D_{1g}$, and $\widehat{\mathbf{U}}_{idd'}$ is a $G \times 1$ vector residuals $\hat{U}_{i d d' g}$ as entries. Using this notation, we can express the estimators $\hat{b}_4(p_0, p_1)$ and $\hat{c}_4(d, d'; p_0, p_1)$ derived in Section \ref{subsec:estimation} as 
\begin{equation*}
    \begin{aligned}
        &\hat{b}_4(p_0, p_1) = e'_5 \big(\widehat{\mathbf{X}}'_P \widehat{\mathbf{W}}_{h_{G1}} \widehat{\mathbf{X}}_P\big)^{-1} \widehat{\mathbf{X}}'_P \widehat{\mathbf{W}}_{h_{G1}} \mathbf{D}, \\
        & \hat{c}_4(d, d'; p_0, p_1) = e'_5 \big(\widehat{\mathbf{X}}'_P \widehat{\mathbf{W}}_{h_{G1}} \widehat{\mathbf{X}}_P\big)^{-1} \widehat{\mathbf{X}}'_P \widehat{\mathbf{W}}_{h_{G1}} \widehat{\mathbf{U}}_{idd'},
    \end{aligned}
\end{equation*}
where $e_5$ is the $10 \times 1$ standard basis vector with a one in the fifth entry and zeros elsewhere.

We aim to characterize the asymptotic behavior of 
\begin{equation} \label{eq:asymp_denom}
    e'_5 \big(\widehat{\mathbf{X}}'_P \widehat{\mathbf{W}}_{h_{G1}} \widehat{\mathbf{X}}_P\big)^{-1} \widehat{\mathbf{X}}'_P \widehat{\mathbf{W}}_{h_{G1}} \mathbf{D} - \frac{\partial^2}{\partial p_0 \partial p_1} \mathbb{E}\big[D_{0g}D_{1g} \mid P_{0g} = p_0, P_{1g} = p_1\big].
\end{equation}
We rewrite Equation \eqref{eq:asymp_denom} as 
\begin{equation*}
    \begin{aligned}
        &e'_5 \big(\widehat{\mathbf{X}}'_P \widehat{\mathbf{W}}_{h_{G1}} \widehat{\mathbf{X}}_P\big)^{-1} \widehat{\mathbf{X}}'_P \widehat{\mathbf{W}}_{h_{G1}} \big[\mathbf{D} - \mathbf{D}^*\big] + e'_5 \big(\widehat{\mathbf{X}}'_P \widehat{\mathbf{W}}_{h_{G1}} \widehat{\mathbf{X}}_P\big)^{-1} \widehat{\mathbf{X}}'_P \widehat{\mathbf{W}}_{h_{G1}} \mathbf{D}^* \\
        &- e'_5 \big(\widehat{\mathbf{X}}'_P \widehat{\mathbf{W}}_{h_{G1}} \widehat{\mathbf{X}}_P\big)^{-1} \widehat{\mathbf{X}}'_P \widehat{\mathbf{W}}_{h_{G1}} \widehat{\mathbf{X}}_P \mathbf{D}_* \\
        =& e'_5 \big(\widehat{\mathbf{X}}'_P \widehat{\mathbf{W}}_{h_{G1}} \widehat{\mathbf{X}}_P\big)^{-1} \widehat{\mathbf{X}}'_P \widehat{\mathbf{W}}_{h_{G1}} \big[\mathbf{D} - \mathbf{D}^*\big] + e'_5 \big(\widehat{\mathbf{X}}'_P \widehat{\mathbf{W}}_{h_{G1}} \widehat{\mathbf{X}}_P\big)^{-1} \widehat{\mathbf{X}}'_P \widehat{\mathbf{W}}_{h_{G1}} \big[\mathbf{D}^* - \widehat{\mathbf{X}}_P \mathbf{D}_*\big],
    \end{aligned} 
\end{equation*}
and define $\mathbf{D}^*, \mathbf{D}_*$ as
\begin{equation*}
    \begin{aligned}
        \mathbf{D}^* \equiv & \big[\mathbb{E}[D_{01} D_{11} \mid P_{01}, P_{11}], \cdots, \mathbb{E}[D_{0G} D_{1G} \mid P_{0G}, P_{1G}]\big]',\\
         \mathbf{D}_* \equiv & \Big[\mathbb{E}\big[D_{0g}D_{1g} \mid P_{0g} = p_0, P_{1g} = p_1\big], \cdots, \frac{\partial^2}{\partial p_0 \partial p_1} \mathbb{E}\big[D_{0g}D_{1g} \mid P_{0g} = p_0, P_{1g} = p_1\big], \\
        & \cdots, \frac{\partial^3}{\partial p_1^3} \mathbb{E}\big[D_{0g}D_{1g} \mid P_{0g} = p_0, P_{1g} = p_1\big]\Big]',
    \end{aligned}
\end{equation*}
where $\mathbf{D}^*$ is a $G \times 1$ vector consisting of the conditional means $\mathbb{E}[D_{0g} D_{1g} \mid P_{0g}, P_{1g}], g = 1, \cdots, G$, and $\mathbf{D}_*$ is a $10 \times 1$ vector comprising the partial derivatives of $\mathbb{E}[D_{0g}D_{1g} \mid P_{0g} = p_0, P_{1g} = p_1]$ up to third order.

We apply Taylor series expansion to expand the conditional mean $\mathbb{E}[D_{0g} D_{1g} \mid P_{0g}, P_{1g}]$ around $(p_0, p_1)$ and express it as
\begin{equation*}
    \begin{aligned}
        & \mathbb{E}\big[D_{0g} D_{1g} \mid P_{0g} = p_0, P_{1g} = p_1\big] + \frac{\partial}{\partial p_0} \mathbb{E}\big[D_{0g} D_{1g} \mid P_{0g} = p_0, P_{1g} = p_1\big] (P_{0g} - p_0) \\
        &+ \frac{\partial}{\partial p_1} \mathbb{E}\big[D_{0g} D_{1g} \mid P_{0g} = p_0, P_{1g} = p_1\big] (P_{1g} - p_1) + \cdots \\
        &+ \frac{\partial^3}{6 \partial p_1^3} \mathbb{E}\big[D_{0g} D_{1g} \mid P_{0g} = p_0, P_{1g} = p_1\big] (P_{1g} - p_1)^3 + R_P(p_0, p_1) \\
        =& \mathbb{E}\big[D_{0g} D_{1g} \mid P_{0g} = p_0, P_{1g} = p_1\big] + \frac{\partial}{\partial p_0} \mathbb{E}\big[D_{0g} D_{1g} \mid P_{0g} = p_0, P_{1g} = p_1\big] \big[(P_{0g} - \hat{P}_{0g}) + (\hat{P}_{0g} - p_0)\big] \\
        &+ \frac{\partial}{\partial p_1} \mathbb{E}\big[D_{0g} D_{1g} \mid P_{0g} = p_0, P_{1g} = p_1\big] \big[(P_{1g} - \hat{P}_{1g}) + (\hat{P}_{1g} - p_1)\big] + \cdots \\
        &+ \frac{\partial^3}{6 \partial p_1^3} \mathbb{E}\big[D_{0g} D_{1g} \mid P_{0g} = p_0, P_{1g} = p_1\big] \big[(P_{1g} - \hat{P}_{1g}) + (\hat{P}_{1g} - p_1)\big]^3 + R_P(p_0, p_1) 
    \end{aligned}
\end{equation*}
where $R_P(p_0, p_1)$ represents the remainder terms from the Taylor expansion, and the last step is to decompose $(P_{ig} - p_i)$ as $[(P_{ig} - \hat{P}_{ig})+(\hat{P}_{ig} - p_i)]$, $i \in \{0,1\}$. Noting that the $g$th entry of $[\mathbf{D}^* - \widehat{\mathbf{X}}_P \mathbf{D}_*]$ equals to
\begin{equation*}
    \begin{aligned}
        \big[\mathbf{D}^* - \widehat{\mathbf{X}}_P \mathbf{D}_* \big]_g =& \mathbb{E}\big[D_{0g} D_{1g} \mid P_{0g}, P_{1g}\big] - \mathbb{E}\big[D_{0g} D_{1g} \mid P_{0g} = p_0, P_{1g} = p_1\big] \\
        &- \frac{\partial}{\partial p_0} \mathbb{E}\big[D_{0g} D_{1g} \mid P_{0g} = p_0, P_{1g} = p_1\big] (\hat{P}_{0g} - p_0) - \cdots \\
        &- \frac{\partial^3}{\partial p_1^3} \mathbb{E}\big[D_{0g} D_{1g} \mid P_{0g} = p_0, P_{1g} = p_1\big] (\hat{P}_{1g} - p_1)^3,
    \end{aligned}
\end{equation*}
we can leverage the expansion results to show
\begin{equation*}
    \begin{aligned}
        \big[\mathbf{D}^* - \widehat{\mathbf{X}}_P \mathbf{D}_* \big]_g = O_P\Big[\max _{g: 1 \leq g \leq G}|\hat{P}_{0g}-P_{0g}| + \max _{g: 1 \leq g \leq G}|\hat{P}_{1g}-P_{1g}| + h_{G1}^4\Big].
    \end{aligned}
\end{equation*}
The results in Section \ref{subsec:ps_consistent}, combined with the boundedness of the kernel assumed in Assumption \ref{as:local_polynomial}, imply that
% \begin{equation*}
%     \max _{g: 1 \leq g \leq G}\big|\hat{P}_{ig}-P_{ig} \big| = o_p(1), \max _{g: 1 \leq g \leq G}\big|K_{h_{G 1}}(\hat{P}_g-p)-K_{h_{G 1}}(P_g-p)\big| = o_p(1), 
% \end{equation*}
% which imply that 
\begin{equation*}
    \begin{aligned}
        \big\|\widehat{\mathbf{X}}_P - \mathbf{X}_P\big\|\ = o_P(1), \big\|\widehat{\mathbf{W}}_{h_{G 1}}-\mathbf{W}_{h_{G 1}}\big\| = o_P(1),
    \end{aligned}
\end{equation*}
where we define $\mathbf{X}_P$ and $\mathbf{W}_{h_{G 1}}$ as 
\begin{equation*}
    \begin{aligned}
        &\mathbf{X}_P \equiv \left[\begin{array}{cccccc}
            1 & \big(P_{01}-p_0\big) & \cdots & \big(P_{01}-p_0\big) \big(P_{11}-p_1\big) & \cdots & \big(P_{11}-p_1\big)^3\\
            \vdots & \vdots & \vdots & \vdots & \vdots & \vdots \\
            1 & \big(P_{0G}-p_0\big) & \cdots & \big(P_{0G}-p_0\big) \big(P_{1G}-p_1\big) & \cdots & \big(P_{1G}-p_1\big)^3,
            \end{array}\right] \\
        & \mathbf{W}_h \equiv \operatorname{diag} \Big(K_{h}\big(P_1-p\big), \cdots, K_{h}\big(P_G-p\big)\Big).
    \end{aligned}
\end{equation*}
Since $\big(\mathbf{X}_P^{\prime} \mathbf{W}_{h_{G 1}} \mathbf{X}_P\big)^{-1} \mathbf{X}_P^{\prime} \mathbf{W}_{h_{G 1}} = O_P(1)$ by assumptions, it follows that
\begin{equation*}
    \begin{aligned}
        & \big(\widehat{\mathbf{X}}_P^{\prime} \widehat{\mathbf{W}}_{h_{G 1}} \widehat{\mathbf{X}}_P\big)^{-1} \widehat{\mathbf{X}}_P^{\prime} \widehat{\mathbf{W}}_{h_{G 1}} = O_P(1) \\
        \implies& e_5^{\prime}\big(\widehat{\mathbf{X}}_P^{\prime} \widehat{\mathbf{W}}_{h_{G 1}} \widehat{\mathbf{X}}_P\big)^{-1} \widehat{\mathbf{X}}_P^{\prime} \widehat{\mathbf{W}}_{h_{G 1}}\big[\mathbf{D}^*-\widehat{\mathbf{X}}_P \mathbf{D}_*\big] \\
        & = O_P\Big[\max _{g: 1 \leq g \leq G}|\hat{P}_{0g}-P_{0g}| + \max _{g: 1 \leq g \leq G}|\hat{P}_{1g}-P_{1g}| + h_{G1}^4\Big].
    \end{aligned}
\end{equation*}
Additionally, 
\begin{equation*}
    e'_5 \big(\widehat{\mathbf{X}}'_P \widehat{\mathbf{W}}_{h_{G1}} \widehat{\mathbf{X}}_P\big)^{-1} \widehat{\mathbf{X}}'_P \widehat{\mathbf{W}}_{h_{G1}} \big[\mathbf{D} - \mathbf{D}^*\big] = e'_5 \big(\mathbf{X}_P^{\prime} \mathbf{W}_{h_{G 1}} \mathbf{X}_P\big)^{-1} \mathbf{X}_P^{\prime} \mathbf{W}_{h_{G 1}} \big[\mathbf{D} - \mathbf{D}^*\big]\big[1 + o_P(1)\big].
\end{equation*} 

By applying the results from \cite{masry1996multivariate}, the term $e'_5 \big(\mathbf{X}_P^{\prime} \mathbf{W}_{h_{G 1}} \mathbf{X}_P\big)^{-1} \mathbf{X}_P^{\prime} \mathbf{W}_{h_{G 1}} \big[\mathbf{D} - \mathbf{D}^*\big]$ converges at the rate $O_P[(G h_{G1}^6)^{-1/2}]$. the convergence rate of Equation \eqref{eq:asymp_denom} is given by 
\begin{equation*}
    \begin{aligned}
        &e'_5 \big(\widehat{\mathbf{X}}'_P \widehat{\mathbf{W}}_{h_{G1}} \widehat{\mathbf{X}}_P\big)^{-1} \widehat{\mathbf{X}}'_P \widehat{\mathbf{W}}_{h_{G1}} \mathbf{D} - \frac{\partial^2}{\partial p_0 \partial p_1} \mathbb{E}\big[D_{0g}D_{1g} \mid P_{0g} = p_0, P_{1g} = p_1\big] \\
        =& O_P\Big[(G h_{G1}^6)^{-1/2} + \max _{g: 1 \leq g \leq G}|\hat{P}_{0g}-P_{0g}| + \max _{g: 1 \leq g \leq G}|\hat{P}_{1g}-P_{1g}| + h_{G1}^4\Big]
    \end{aligned}
\end{equation*}

Furthermore, we can show that
\begin{equation*}
    \begin{aligned}
        & e'_5 \big(\widehat{\mathbf{X}}'_P \widehat{\mathbf{W}}_{h_{G1}} \widehat{\mathbf{X}}_P\big)^{-1} \widehat{\mathbf{X}}'_P \widehat{\mathbf{W}}_{h_{G1}} \widehat{\mathbf{U}}_{idd'} - \frac{\partial^2}{\partial p_0 \partial p_1} \mathbb{E}\big[U_{i d d' g} \mid P_{0g} = p_0, P_{1g} = p_1\big] \\
        =& e'_5 \big(\widehat{\mathbf{X}}'_P \widehat{\mathbf{W}}_{h_{G1}} \widehat{\mathbf{X}}_P\big)^{-1} \widehat{\mathbf{X}}'_P \widehat{\mathbf{W}}_{h_{G1}} \big[\mathbf{U}_{idd'} - X_i(\hat{\beta}_{dd'} - \beta_{dd'})\big] \\
        &- \frac{\partial^2}{\partial p_0 \partial p_1} \mathbb{E}\big[U_{i d d' g}  \mid P_{0g} = p_0, P_{1g} = p_1\big] \\
        =& e'_5 \big(\widehat{\mathbf{X}}'_P \widehat{\mathbf{W}}_{h_{G1}} \widehat{\mathbf{X}}_P\big)^{-1} \widehat{\mathbf{X}}'_P \widehat{\mathbf{W}}_{h_{G1}} \mathbf{U}_{idd'} - \frac{\partial^2}{\partial p_0 \partial p_1} \mathbb{E}\big[U_{i d d' g}  \mid P_{0g} = p_0, P_{1g} = p_1\big] + O_P(G^{-1/2}),
    \end{aligned}
\end{equation*}
where the last equation holds by $\big(\widehat{\mathbf{X}}'_P \widehat{\mathbf{W}}_{h_{G1}} \widehat{\mathbf{X}}_P\big)^{-1} \widehat{\mathbf{X}}'_P \widehat{\mathbf{W}}_{h_{G1}} = O_P(1)$, $\hat{\beta}_{dd'} - \beta_{dd'} = O_P(G^{-1/2})$ according to Theorem 3 in \cite{carneiro2009estimating}, and $\widehat{\mathbf{U}}_{idd'}, \mathbf{U}_{idd'}$ are defined as
\begin{equation*}
    \widehat{\mathbf{U}}_{idd'} = \big[\hat{U}_{i d d' 1}, \cdots, \hat{U}_{i d d' G}\big]', \mathbf{U}_{idd'} = \big[U_{i d d' 1}, \cdots, U_{i d d' G}\big]'.
\end{equation*}
Then, the convergence rate of $e'_5 \big(\widehat{\mathbf{X}}'_P \widehat{\mathbf{W}}_{h_{G1}} \widehat{\mathbf{X}}_P\big)^{-1} \widehat{\mathbf{X}}'_P \widehat{\mathbf{W}}_{h_{G1}} \widehat{\mathbf{U}}_{idd'}$ can be proven as 
\begin{equation*}
    \begin{aligned}
        &e'_5 \big(\widehat{\mathbf{X}}'_P \widehat{\mathbf{W}}_{h_{G1}} \widehat{\mathbf{X}}_P\big)^{-1} \widehat{\mathbf{X}}'_P \widehat{\mathbf{W}}_{h_{G1}} \widehat{\mathbf{U}}_{idd'} - \frac{\partial^2}{\partial p_0 \partial p_1} \mathbb{E}\big[U_{i d d' g}  \mid P_{0g} = p_0, P_{1g} = p_1\big] \\
        =& O_P\Big[(G h_{G2}^6)^{-1/2} + \max _{g: 1 \leq g \leq G}|\hat{P}_{0g}-P_{0g}| + \max _{g: 1 \leq g \leq G}|\hat{P}_{1g}-P_{1g}| + h_{G2}^4\Big]
    \end{aligned}
\end{equation*}
using an argument analogous to that used in the preceding analysis.

\subsection{Asymptotic distribution of nonparametric marginal treatment response estimators} \label{app:asymp_dist_ratio}

Under Assumption \ref{as:asym_dist_ratio}, combined with conclusions in \cite{masry1996multivariate}, we have 
\begin{equation*}
    \begin{aligned}
        &\big(G h_{G2}^6\big)^{1/2}\Big\{\hat{c}_4(d, d'; p_0, p_1) - c_4(d, d'; p_0, p_1) \Big\} \xrightarrow{d} N\Big(0, \frac{\sigma^2(d, d'; p_0, p_1)}{f(p_0, p_1)} \big(M^{-1}\Gamma M^{-1}\big)_{5,5}\Big), \\
        & c_4(d, d'; p_0, p_1) \equiv \frac{\partial^2}{\partial p_0 \partial p_1} \mathbb{E}\big[U_{i d d' g} \mid P_{0g} = p_0, P_{1g} = p_1\big],
    \end{aligned}
\end{equation*}
where $\sigma^2(d, d'; p_0, p_1) = \text{Var}(U_{i d d' g} \mid P_{0g} = p_0, P_{1g} = p_1)$, and $f(p_0, p_1)$ denotes the density of $(P_{0g}, P_{1g})$ evaluated at the point $(p_0, p_1)$. The matrices $M$ and $\Gamma$ are $10 \times 10$-matrices composed of multivariate moments of the kernel functions $K$ and $K^2$, and are defined as 
\begin{equation*}
    \begin{aligned}
        &M=\left[\begin{array}{lllll}
            \int u_0^0 u_1^0 K(u)d(u) & \int u_0^1 u_1^0 K(u)d(u) & \cdots & \int u_0^1 u_1^2 K(u)d(u) & \int u_0^0 u_1^3 K(u)d(u) \\
            \int u_0^1 u_1^0 K(u)d(u) & \int u_0^2 u_1^0 K(u)d(u) & \cdots & \int u_0^2 u_1^2 K(u)d(u) & \int u_0^1 u_1^3 K(u)d(u) \\
            \vdots & \vdots & \cdots & \vdots & \vdots \\
            \int u_0^0 u_1^3 K(u)d(u) & \int u_0^1 u_1^2 K(u)d(u) & \cdots & \int u_0^1 u_1^5 K(u)d(u) & \int u_0^0 u_1^6 K(u)d(u) \\
            \end{array} \right], \\
        &\Gamma=\left[\begin{array}{lllll}
                \int u_0^0 u_1^0 K^2(u)d(u) & \int u_0^1 u_1^0 K^2(u)d(u) & \cdots & \int u_0^1 u_1^2 K^2(u)d(u) & \int u_0^0 u_1^3 K^2(u)d(u) \\
                \int u_0^1 u_1^0 K^2(u)d(u) & \int u_0^2 u_1^0 K^2(u)d(u) & \cdots & \int u_0^2 u_1^2 K^2(u)d(u) & \int u_0^1 u_1^3 K^2(u)d(u) \\
                \vdots & \vdots & \cdots & \vdots & \vdots \\
                \int u_0^0 u_1^3 K^2(u)d(u) & \int u_0^1 u_1^2 K^2(u)d(u) & \cdots & \int u_0^1 u_1^5 K^2(u)d(u) & \int u_0^0 u_1^6 K^2(u)d(u) \\
                \end{array} \right].
    \end{aligned}
\end{equation*}
Since $\partial^2 \mathbb{E}[D_{0 g} D_{1 g} \mid P_{0 g}=p_0, P_{1 g}=p_1] / \partial p_0 \partial p_1$ is bounded from above and away from zero by Assumption \ref{as:local_polynomial} and $\hat{b}_4(p_0,p_1) \xrightarrow{p} \partial^2 \mathbb{E}[D_{0 g} D_{1 g} \mid P_{0 g}=p_0, P_{1 g}=p_1] / \partial p_0 \partial p_1 \equiv b_4(p_0, p_1)$, it follows that 
\begin{equation*}
    \begin{aligned}
        &\big(G h_{G2}^6\big)^{1/2}\Bigg\{\hat{c}_4(d, d'; p_0, p_1) - c_4(d, d'; p_0, p_1) \Bigg\} \frac{1}{\hat{b}_4(p_0,p_1)} \\
        \xrightarrow{d}& N\Bigg(0, \frac{\sigma^2(d, d'; p_0, p_1)}{\big(b_4(p_0, p_1)\big)^2 f(p_0, p_1)} \big(M^{-1}\Gamma M^{-1}\big)_{5,5}\Bigg)
    \end{aligned}
\end{equation*} 
We can rewrite $\big(G h_{G2}^6\big)^{1/2}\Big\{\hat{c}_4(d, d'; p_0, p_1) - c_4(d, d'; p_0, p_1) \Big\} \big/ \hat{b}_4(p_0,p_1)$ as 
\begin{equation*}
    \begin{aligned}
        &\big(G h_{G2}^6\big)^{1/2}\Bigg\{\hat{c}_4(d, d'; p_0, p_1) - c_4(d, d'; p_0, p_1) \Bigg\} \frac{1}{\hat{b}_4(p_0,p_1)} \\
        =& \big(G h_{G2}^6\big)^{1/2}\Bigg\{\frac{\hat{c}_4(d, d'; p_0, p_1)}{\hat{b}_4(p_0,p_1)} - \frac{c_4(d, d'; p_0, p_1)}{b_4(p_0,p_1)} + \frac{c_4(d, d'; p_0, p_1)}{b_4(p_0,p_1)} - \frac{c_4(d, d'; p_0, p_1)}{\hat{b}_4(p_0,p_1)} \Bigg\} \\
        =& \big(G h_{G2}^6\big)^{1/2}\Bigg\{\frac{\hat{c}_4(d, d'; p_0, p_1)}{\hat{b}_4(p_0,p_1)} - \frac{c_4(d, d'; p_0, p_1)}{b_4(p_0,p_1)} \Bigg\} \\
        &+ \big(G h_{G2}^6\big)^{1/2}\Bigg\{\frac{c_4(d, d'; p_0, p_1)}{b_4(p_0,p_1)} - \frac{c_4(d, d'; p_0, p_1)}{\hat{b}_4(p_0,p_1)} \Bigg\}.
    \end{aligned}
\end{equation*}
By applying Assumption \ref{as:asym_dist_ratio} along with the results in Section \ref{subsec:asymp_cross_deriv}, we obtain
\begin{equation*}
    \frac{1}{\hat{b}_4(p_0,p_1)} - \frac{1}{b_4(p_0,p_1)} = O_P \big[(G h_{G 1}^6)^{-1 / 2}\big], 
\end{equation*}
which implies that
\begin{equation*}
    \big(G h_{G2}^6\big)^{1/2}\Bigg\{\frac{c_4(d, d'; p_0, p_1)}{b_4(p_0,p_1)} - \frac{c_4(d, d'; p_0, p_1)}{\hat{b}_4(p_0,p_1)} \Bigg\} = o_P(1)
\end{equation*}
under the condition $h_{G2} = o(h_{G1})$.
Therefore, 
\begin{equation*}
    \begin{aligned}
        &\big(G h_{G2}^6\big)^{1/2}\Bigg\{\frac{\hat{c}_4(d, d'; p_0, p_1)}{\hat{b}_4(p_0,p_1)} - \frac{c_4(d, d'; p_0, p_1)}{b_4(p_0,p_1)} \Bigg\} \\
        =& \big(G h_{G2}^6\big)^{1/2}\Bigg\{\hat{c}_4(d, d'; p_0, p_1) - c_4(d, d'; p_0, p_1) \Bigg\} \frac{1}{\hat{b}_4(p_0,p_1)} + o_P(1),
    \end{aligned}
\end{equation*}
and the asymptotic distribution of estimated marginal treatment response function without the covariate effect can be characterized as
\begin{equation*}
    \begin{aligned}
        & \big(G h_{G2}^6\big)^{1/2}\Bigg\{\frac{\hat{c}_4(d, d'; p_0, p_1)}{\hat{b}_4(p_0,p_1)} - \frac{c_4(d, d'; p_0, p_1)}{b_4(p_0,p_1)} \Bigg\} \\
        \xrightarrow{d}& N\Bigg(0, \frac{\sigma^2(d, d'; p_0, p_1)}{\big(b_4(p_0, p_1)\big)^2 f(p_0, p_1)} \big(M^{-1}\Gamma M^{-1}\big)_{5,5}\Bigg).
    \end{aligned}
\end{equation*}

Finally, under the assumptions that $\hat{c}_4(d, d'; p_0, p_1) / \hat{b}_4(p_0,p_1) - c_4(d, d'; p_0, p_1) / b_4(p_0,p_1)$ are asymptotically independent across different values of $d, d' \in \{0,1\}$, we can derive the asymptotic distributions of $\widehat{\text{MCSE}}(\mathbf{x}, d; p_0, p_1)$ as 
\begin{equation*}
    \begin{aligned}
        &\big(G h_{G2}^6\big)^{1/2}\Bigg\{\widehat{\text{MCSE}}(\mathbf{x}, d; p_0, p_1) - \text{MCSE}(\mathbf{x}, d; p_0, p_1)\Bigg\} \\
        =& \big(G h_{G2}^6\big)^{1/2} \Big[\mathbf{x}'\big(\hat{\beta}_{d1} - \beta_{d1}\big) +  \mathbf{x}'\big(\hat{\beta}_{d0} - \beta_{d0}\big)\Big] \\
        &+ \big(G h_{G2}^6\big)^{1/2} \Bigg[\Big(\frac{\hat{c}_4(d, 1; p_0, p_1)}{\hat{b}_4(p_0,p_1)} - \frac{c_4(d, 1; p_0, p_1)}{b_4(p_0,p_1)}\Big) + \Big(\frac{\hat{c}_4(d, 0; p_0, p_1)}{\hat{b}_4(p_0,p_1)} - \frac{c_4(d, 0; p_0, p_1)}{b_4(p_0,p_1)}\Big)\Bigg] \\
        =& o_P(1) + \big(G h_{G2}^6\big)^{1/2} \Bigg[\Big(\frac{\hat{c}_4(d, 1; p_0, p_1)}{\hat{b}_4(p_0,p_1)} - \frac{c_4(d, 1; p_0, p_1)}{b_4(p_0,p_1)}\Big) + \Big(\frac{\hat{c}_4(d, 0; p_0, p_1)}{\hat{b}_4(p_0,p_1)} - \frac{c_4(d, 0; p_0, p_1)}{b_4(p_0,p_1)}\Big)\Bigg] \\
        &\xrightarrow{d} N\Bigg(0, \frac{\sigma^2(d, 1; p_0, p_1) + \sigma^2(d, 0; p_0, p_1)}{\big(b_4(p_0, p_1)\big)^2 f(p_0, p_1)} \big(M^{-1}\Gamma M^{-1}\big)_{5,5}\Bigg),
    \end{aligned}
\end{equation*}
where the second equality holds because $\hat{\beta}_{dd'} - \beta_{dd'} = O_P(G^{-1/2})$ applying Theorem 3 in \cite{carneiro2009estimating}. Similarly, the asymptotic distribution of $\widehat{\text{MCDE}}(\mathbf{x}, d; p_0, p_1)$ can be derived as 
\begin{equation*}
    \begin{aligned}
        &\big(G h_{G2}^6\big)^{1/2}\Bigg\{\widehat{\text{MCSE}}(\mathbf{x}, d; p_0, p_1) - \text{MCSE}(\mathbf{x}, d; p_0, p_1)\Bigg\} \\
        &\xrightarrow{d} N\Bigg(0, \frac{\sigma^2(1, d; p_0, p_1) + \sigma^2(0, d; p_0, p_1)}{\big(b_4(p_0, p_1)\big)^2 f(p_0, p_1)} \big(M^{-1}\Gamma M^{-1}\big)_{5,5}\Bigg).
    \end{aligned}
\end{equation*}

\subsection{Consistency of the parametric first-stage estimator} \label{app:proof_parametric_first}
For each $i \in \{0, 1\}$, the function $l(\theta_i; d, \mathbf{w})$ is continuous in $\theta_i$ for all $d \in \{0, 1\}$ and $\mathbf{w} \in \mathcal{W}$. Additionally, the parameter space $\Theta_i$ is compact and $\mathbb{E}[\sup_{\theta_i \in \Theta_i}|l(\theta_i; D_{ig}, W_g)|] < \infty$ given Assumption \ref{as:para_consist1}, applying the uniform law of large numbers, we have 
    \begin{equation*}
        \sup_{\theta_i \in \Theta_i}\Bigg|\frac{1}{G} \sum_{g = 1}^{G} l(\theta_i; D_{ig}, W_g) - \mathbb{E}\big[l(\theta_i; D_{ig}, W_g)\big] \Bigg| \asto 0,
    \end{equation*}
    as $G \rightarrow \infty$. Define $Q(\theta_i) = \mathbb{E}\big[l(\theta_i; D_{ig}, W_g)\big]$ and $Q_G(\theta_i) = \sum_{g = 1}^{G} l(\theta_i; D_{ig}, W_g) / G$. We can derive 
    \begin{equation*}
        \begin{aligned}
0 \leq Q(\theta_{i0})-Q(\hat{\theta}_i) & =Q_G(\hat{\theta}_i)-Q(\hat{\theta}_i)+Q(\theta_{i0})-Q_G(\hat{\theta}_i) \\
& \leq \sup _{\theta_i \in \Theta_i}\big|Q_G(\theta_i)-Q(\theta_i)\big|+Q(\theta_{i0})-Q_G(\theta_{i0}) \\
& \leq 2 \sup _{\theta_i \in \Theta_i}\left|Q_G(\theta_i)-Q(\theta_i)\right| \\
& \asto 0,
\end{aligned}
    \end{equation*}
    as $G \rightarrow \infty$, where the second line holds because $\hat{\theta}_i$ maximizes the function $Q_G(\theta_i)$. Since $\theta_i$ is the unique maximizer of $Q(\theta_i)$ and $\Theta_i$ is compact based on Assumption \ref{as:para_consist1}, $Q_G(\hat{\theta}_i) \asto Q(\theta_{i0})$ implies that $\hat{\theta}_i \asto \theta_{i0}$ as $G \rightarrow \infty$.

\subsection{Consistency of the parametric second-stage estimator} \label{app:proof_parametric_second}
We use $Q_G(\rho)$ and $\widehat{Q}_G(\rho)$ to define 
    \begin{equation*}
        \begin{aligned}
            &Q_G(\rho) = \frac{1}{G} \sum_{g = 1}^{G} \tilde{l}\big(\rho; D_g, P_g \big), \\
            &\widehat{Q}_G(\rho) = \frac{1}{G} \sum_{g = 1}^{G} \tilde{l}\big(\rho; D_g, \widehat{P}_g\big),
        \end{aligned}
    \end{equation*}
    where $\widehat{P}_g = (\widehat{P}_{0g}, \widehat{P}_{1g})$ is the vector of propensity scores estimated in the first stage. Then, we can write 
    \begin{equation*}
        \begin{aligned}
            \bigg|\widehat{Q}_G(\rho) - \mathbb{E}\big[\tilde{l}(\rho; D_g, P_g)\big] \bigg| \leq& \bigg|\widehat{Q}_G(\rho) - Q_G(\rho)\bigg| \\
            +& \bigg|Q_G(\rho) - \mathbb{E}\big[\tilde{l}(\rho; D_g, P_g)\big]\bigg|.
        \end{aligned}
    \end{equation*}
    Since $\tilde{l}(\rho; d, p)$ is continuous in $\rho$ for all $d \in \{0, 1\}^2$ and $p \in (0, 1)^2$, $\rho$ lies in a compact interval, and $\mathbb{E}\big[\sup_{\rho \in [-\varepsilon, \varepsilon]} |l(\rho, \theta; D_{g}, W_{g})|\big] < \infty$ under Assumption \ref{as:para_consist2}, the law of large numbers implies that
    \begin{equation*}
        \sup_{\rho \in [-\varepsilon, \varepsilon]}\bigg|Q_G(\rho) - \mathbb{E}\big[\tilde{l}(\rho; D_g, P_g)\big]\bigg| \asto 0.
    \end{equation*}

    Assuption \ref{as:para_consist2} also assumes that there exists a function $L(\cdot)$ such that for all $\rho \in [-\varepsilon, \varepsilon]$, 
    \begin{equation*}
        \big|\tilde{l}(\rho; D_g, \widehat{P}_g) - \tilde{l}(\rho; D_g, P_g)\big| \leq L(D_g) \big|\big|\widehat{P}_g - P_g \big| \big|.
    \end{equation*}
    Since $|L(D_g)| < \infty$ almost surely and $\big|\big|\widehat{P}_g - P_g \big| \big| \asto 0$ by Lemma \ref{lemma:para_consist1}, it follows that $\sup_{\rho \in [\varepsilon, \varepsilon]} \big|\tilde{l}(\rho; D_g, \widehat{P}_g) - \tilde{l}(\rho; D_g, P_g)\big| \asto 0$, which further implies 
    \begin{equation*}
        \sup_{\rho \in [\varepsilon, \varepsilon]} \bigg|\widehat{Q}_G(\rho) - Q_G(\rho)\bigg| \asto 0.
    \end{equation*}
    
    Therefore, we have 
    \begin{equation*}
        \sup_{\rho \in [\varepsilon, \varepsilon]} \bigg|\widehat{Q}_G(\rho) - \mathbb{E}\big[\tilde{l}(\rho; D_g, P_g)\big] \bigg| \asto 0.
    \end{equation*}
    Since $\rho_0$ is the unique maximizer of $\mathbb{E}\big[\tilde{l}(\rho; D_g, P_g)\big]$ and lies within a compact interval, by the similar arguments in the proof of Lemma \ref{lemma:para_consist1}, we can derive $\hat{\rho} \asto \rho_0$ as $G \rightarrow \infty$.

\subsection{Consistency of the parametric MTR coefficient estimates} \label{app:proof_parametric_third}
Based on the identification results the third specification in Assumption \ref{as:para_assump}, for each $i \in \{0, 1\}$ and $g \in \{1, \cdots, G\}$,
    \begin{equation*}
        Y_{ig} \mathbbm{1} \big\{D_{0g} = d \big\} \mathbbm{1} \big\{D_{1g} = d' \big\} = X_{Pdd'_g} \big(\alpha'_{idd'}, \beta'_{idd'} \big)' + \varepsilon_{idd'g},
    \end{equation*}
    where the error term $\varepsilon_{idd'g}$ satisfies $\mathbb{E}[\varepsilon_{idd'g} \mid X_{Pdd'_g}] = 0$. The vector of coefficients $(\alpha'_{idd'}, \beta'_{idd'})'$ is estimated by 
    \begin{equation*}
        \begin{aligned}
            \big(\hat{\alpha}'_{idd'}, \hat{\beta}'_{idd'} \big)' = \big(\widehat{X}'_{Pdd'} \widehat{X}_{Pdd'}\big)^{-1} \widehat{X}'_{Pdd'} \widetilde{Y}_{idd'},
        \end{aligned}
    \end{equation*}
    where $\widetilde{Y}_{idd'}$ is defined as a $G \times 1$ vector with the $g$-th element as $Y_{ig} \mathbbm{1} \big\{D_{0g} = d \big\} \mathbbm{1} \big\{D_{1g} = d' \big\}$, and $\widehat{X}_{Pdd'}$ is obtained by substituting $\widehat{P}_{0g}$, $\widehat{P}_{1g}$, and $\hat{\rho}$ for the true values into $X_{Pdd'}$. Then, we can write the estimated coefficients as
    \begin{equation*}
        \begin{aligned}
             \big(\hat{\alpha}'_{idd'}, \hat{\beta}'_{idd'} \big)' =& \big(\widehat{X}'_{Pdd'} \widehat{X}_{Pdd'}\big)^{-1} \widehat{X}'_{Pdd'} \big(X_{Pdd'} \big(\alpha'_{idd'}, \beta'_{idd'} \big)' + \varepsilon_{idd'}\big) \\
            =& \big(\widehat{X}'_{Pdd'} \widehat{X}_{Pdd'}\big)^{-1} \widehat{X}'_{Pdd'} X_{Pdd'} \big(\alpha'_{idd'}, \beta'_{idd'} \big)' + \big(\widehat{X}'_{Pdd'} \widehat{X}_{Pdd'}\big)^{-1} \widehat{X}'_{Pdd'} \varepsilon_{idd'}.
        \end{aligned}
    \end{equation*}

    Let $\widehat{X}_{Pdd'} = X_{Pdd'} + \Delta_G$, where $\Delta_G$ is defined as a $G \times K$ matrix such that $\Delta_G = \widehat{X}_{Pdd'} - X_{Pdd'}$. Then, we have 
    \begin{equation*}
        \begin{aligned}
            \frac{1}{G} \widehat{X}'_{Pdd'} \widehat{X}_{Pdd'} =& \frac{1}{G} \big(X_{Pdd'} + \Delta_G\big)'\big(X_{Pdd'} + \Delta_G\big) \\
            =& \frac{1}{G} X'_{Pdd'} X_{Pdd'} + \frac{1}{G} X'_{Pdd'} \Delta_G + \frac{1}{G} \Delta_G' X_{Pdd'} + \frac{1}{G} \Delta_G' \Delta_G, \\
            \frac{1}{G} \widehat{X}'_{Pdd'} X_{Pdd'} =& \frac{1}{G} \big(X_{Pdd'} + \Delta_G\big)'X_{Pdd'} \\
            =& \frac{1}{G} X'_{Pdd'} X_{Pdd'} + \frac{1}{G} \Delta_G' X_{Pdd'} 
        \end{aligned}
    \end{equation*}
    Applying the Cauchy-Schwarz inequalities, we obtain
    \begin{equation*}
        \frac{1}{G}\big\|X_{Pdd'} \Delta_G\big\|_F \leq \sqrt{\frac{\|X_{Pdd'}\|_F^2}{G}} \cdot \sqrt{\frac{\|\Delta_G\|_F^2}{G}} \asto 0,
    \end{equation*}
    since $\|\Delta_G \|_F^2 / G \asto 0$, and $\|X_{Pdd'}\|_F^2/G$ is bounded almost surely by $\mathbb{E}[X_{P d d'}', X_{P d d'}]$ is nonsingular. Therefore, we should have 
    \begin{equation*}
        \begin{aligned}
            & \frac{1}{G} \widehat{X}'_{Pdd'} \widehat{X}_{Pdd'} =  \frac{1}{G} X'_{Pdd'} X_{Pdd'} + o_{a.s.}(1) \asto \mathbb{E}\big[X'_{Pdd'g} X_{Pdd'g}\big], \\
            & \frac{1}{G} \widehat{X}'_{Pdd'} X_{Pdd'} =  \frac{1}{G} X'_{Pdd'} X_{Pdd'} + o_{a.s.}(1) \asto \mathbb{E}\big[X'_{Pdd'g} X_{Pdd'g}\big],
        \end{aligned}
    \end{equation*}
    and then 
    \begin{equation*}
        \bigg(\frac{1}{G} \widehat{X}'_{Pdd'} \widehat{X}_{Pdd'} \bigg)^{-1} \asto \bigg(\mathbb{E}\big[X'_{Pdd'g} X_{Pdd'g}\big] \bigg)^{-1}
    \end{equation*}
    by the continuous mapping theorem and the nonsingularity condition.

    We can also express the term $\widehat{X}_{P d d'}' \varepsilon_{i d d'} / G$ as
    \begin{equation*}
        \frac{\widehat{X}_{P d d'}' \varepsilon_{i d d'}}{G} = \frac{1}{G}\big(X_{P d d'} + \Delta_G\big)' \varepsilon_{i d d'} = \frac{1}{G} X_{P d d'}' \varepsilon_{i d d'} + \frac{1}{G} \Delta_G' \varepsilon_{i d d'}.
    \end{equation*}
    The first term $X_{P d d'}' \varepsilon_{i d d'} / G \asto \mathbb{E}[X_{P d d'_g}' \varepsilon_{i d d' g}] = 0$ as $\mathbb{E}[\varepsilon_{idd'g} \mid X_{Pdd'_g}] = 0$. Applying the Cauchy-Schwarz, the second term becomes 
    \begin{equation*}
        \bigg\|\frac{1}{G} \Delta_G' \varepsilon_{idd'}\bigg\| \leq \sqrt{\frac{1}{G}\big\|\Delta_G\big\|_F^2} \cdot \sqrt{\frac{1}{G}\|\varepsilon_{idd'}\|^2} \asto 0,
    \end{equation*}
    since $\|\Delta_G \|_F^2 / G \asto 0$, and $\|\varepsilon_{idd'}\|^2/G$ is bounded almost surely by $\text{Var}(\varepsilon_{idd'g}) = \sigma_{idd'g} < \infty$. Thus,
    \begin{equation*}
        \frac{\widehat{X}_{P d d'}' \varepsilon_{i d d'}}{G} \asto 0.
    \end{equation*}

    Combining the above results, we have 
    \begin{equation*}
        \big(\hat{\alpha}'_{idd'}, \hat{\beta}'_{idd'} \big)' \asto \bigg(\mathbb{E}\big[X'_{Pdd'g} X_{Pdd'g}\big] \bigg)^{-1} \mathbb{E}\big[X'_{Pdd'g} X_{Pdd'g}\big] \big(\alpha'_{idd'}, \beta'_{idd'} \big)' = \big(\alpha'_{idd'}, \beta'_{idd'} \big)'.
    \end{equation*}

\section{Model With Continuous Treatment} \label{sec:cts_treatment}
\subsection{Setting}
An important extension of Section \ref{sec:disc_treatment} considers settings in which the treatment $D_{ig}$ is continuous. This extension broadens the applicability of our framework to a wide range of empirical environments. For example, in education, students' long-term outcomes may depend not only on their own time of schooling or study intensity but also on the continuous education investments of their peers, such as best friends or classmates. Similarly, in agriculture, a farmer's yield may be influenced by her own input choices, such as fertilizer use or irrigation intensity, as well as by neighboring farmers' continuous decisions that affect shared resources such as groundwater or pest control. In both cases, treatment choices are endogenous, depending on unobserved preferences, constraints, or abilities, while outcomes may be directly affected by peers' continuous treatments within the group.  

For expositional clarity, we focus on the case where each group contains two units, $i = 0, 1$, while the identification results extend straightforwardly to settings with any finite number of units per group. In this extension, the treatment for unit $i$ in group $g$, $D_{ig} \in \mathbb{R}$, is a continuous random variable. We formulate the following potential outcomes framework To accommodate such treatments. Again, we suppress the group subscript $g$ for notational convenience.

\begin{equation} \label{eq:basic_model_cts_trt}
    \left\{\begin{array}{l}
        Y_{i} = g_i(D_{i}, D_{- i}, U_{i}, U_{-i}) \\
        D_{i} = h_i(Z_{i}, Z_{-i}, V_{i}) 
        \end{array}\right.
\end{equation}

In Equation \eqref{eq:basic_model_cts_trt}, the random variable $V_{i} \in \mathbb{R}$ represents private unobserved characteristics, such as ability in education decisions or resource constraints in agricultural production, that may simultaneously influence both the treatment $D_i$ and the outcome $Y_i$, thereby generating endogeneity. To address this endogeneity, we introduce continuous instrumental variables. Specifically, each unit $i$ is assigned a vector of instruments $Z_{i} \in \mathbb{R}^{k_i}$. We allow the endogenous treatment $D_i$ to depend not only on the individual's own instrument $Z_i$ but also on her peer's instrument $Z_{-i}$. This formulation accommodates the possibility of spillovers in treatment assignment, where one individual's instruments may affect both her own treatment decision and those of her peers.

In the outcome equation, we allow the individual's outcome $Y_i$ to depend not only on her own continuous treatment $D_i$ but also on the vector of continuous treatments $D_{-i}$ chosen by her peers within the group. The random vector $U_i \in \mathbb{R}^{l_i}$ captures the individual-specific unobserved characteristics that affect outcomes, and we impose no restrictions on its dimensionality. Moreover, we explicitly permit the outcome $Y_i$ to depend on the peers' unobservables $U_{-i}$. This formulation captures a rich set of spillover channels, as outcomes may be influenced both by peers' observed treatment decisions and by their latent characteristics. For example, in education, a student's earnings or academic performance may depend on both her own study effort and her peers' corresponding education investments. At the same time, peers' unobserved abilities or motivation may also affect individual's outcomes through collaboration, competition, or shared environments.

For the unobserved random variables $(U_i, V_i)$, we impose no restrictions on their joint dependence structure with $(U_{-i}, V_{-i})$ within each group, similar to the discrete treatment case. This allows for arbitrary correlation in both outcome- and treatment-related unobservables across group members.

In the continuous treatment setting, we maintain Assumption \ref{as:Vdist} and introduce two additional conditions, Assumptions \ref{as:RA_cts} and \ref{as:mon_h}, to establish identification of the marginal controlled spillover and direct effects.

\begin{assumption}(Random assignment: continuous treatment) \label{as:RA_cts}
    We assume that the instruments assigned to each group, $(Z_{i}, Z_{- i})$, satisfy
    \begin{equation*}
        \big(Z_{i}, Z_{- i}\big) \indep \big(V_{i}, V_{-i}, U_{i}, U_{-i}\big).
    \end{equation*}
\end{assumption}

\begin{assumption}(Monotonicity of $h_i$) \label{as:mon_h}
    Given $\mathbf{z}, \mathbf{z'} \in \mathbb{R}^k$, the treatment function $h_i(\mathbf{z}, \mathbf{z'}, v)$ is continuous and strictly monotonic in $v$.
\end{assumption}

In Assumption \ref{as:RA_cts}, we require that the instruments are randomly assigned at the group level. Assumption \ref{as:mon_h} further imposes a monotonicity condition on the function $h_i(\cdot)$, ensuring that the treatment $D_i$ is uniquely mapped to the unobservable $V_i$ conditional on the instruments $(Z_i, Z_{-i})$. Formally, given $Z_i = \mathbf{z}$ and $Z_{-i} = \mathbf{z'}$, the equation $D_{i}\mid(Z_{i} = \mathbf{z}, Z_{-i} = \mathbf{z'}) = h_i(\mathbf{z}, \mathbf{z'}, V_{i})$ can be inverted with respect to $V_{i}$, yielding $V_{i}\mid(Z_{i} = \mathbf{z}, Z_{-i} = \mathbf{z'}) = {h^{i}_{\mathbf{z}, \mathbf{z'}}}^{-1}(D_{i})$. Then, under Assumption \ref{as:RA_cts}, it follows that $V_{i} = {h^{i}_{Z_i, Z_{-i}}}^{-1}(D_{i})$. Without loss of generality, we specify that $h_i(\mathbf{z}, \mathbf{z'}, v)$ is strictly increasing in $v$ for all instrument values $(\mathbf{z}, \mathbf{z'})$. The monotonicity requirement on the treatment function $h_i$ is not an additional structural restriction but instead a natural implication of the existence of a well-defined conditional distribution of the treatment. As shown in \citet{goff2024testing}, if the conditional distribution $F_{D_i \mid Z_i, Z_{-i}}$ is strictly increasing and continuous, then the condition $(Z_i, Z_{-i}) \indep V_i$, together with Assumptions \ref{as:Vdist} and \ref{as:mon_h}, follows naturally under a reduced-form interpretation of the treatment selection equation. In particular, by defining $h_i(\mathbf{z}, \mathbf{z'}, v)=Q_{D_i \mid Z_i=\mathbf{z}, Z_{-i}=\mathbf{z'}}(v)$ and $V_i = F_{D_i \mid Z_i, Z_{-i}}(D_i)$, where $Q_{D_i \mid Z_i=\mathbf{z}, Z_{-i}=\mathbf{z'}}$ denotes the quantile function of distribution $D_i \mid Z_i=\mathbf{z}, Z_{-i}=\mathbf{z'}$, the treatment $D_i$ can be represented as a strictly increasing function of the latent variable $V_i$.

\subsection{Identification}

Analogous to the discrete-treatment case in Equation \eqref{eq:p_score1}, we define the propensity score function under continuous treatments as $\mathbb{P}\left(D_{i} \leq d \mid Z_{i}, Z_{- i}\right) \equiv P_i(Z_{i}, Z_{- i}, d)$ given $d \in \mathbb{R}$. Then, we have 
\begin{equation} \label{eq:p_score1_cts}
    \begin{aligned}
        P_i(\mathbf{z}, \mathbf{z'}, d) \equiv&\mathbb{P}\left(D_{i} \leq d \mid Z_{i} = \mathbf{z}, Z_{-i} = \mathbf{z'}\right) \\
        =& \mathbb{P}\left(h(\mathbf{z}, \mathbf{z'}, V_{i}) \leq d \mid Z_{i} = \mathbf{z}, Z_{-i} = \mathbf{z'}\right) \\
        =& \mathbb{P}\left(V_{i} \leq {h^{i}_{\mathbf{z}, \mathbf{z'}}}^{-1}(d) \mid Z_{i} = \mathbf{z}, Z_{-i} = \mathbf{z'}\right) \\
        =& \mathbb{P} \left(V_{i} \leq {h^{i}_{\mathbf{z}, \mathbf{z'}}}^{-1}(d)\right) \\
        =& {h^{i}_{\mathbf{z}, \mathbf{z'}}}^{-1}(d),
    \end{aligned}
\end{equation}
where the second equality in Equation \eqref{eq:p_score1_cts} is a direct implication of Assumption \ref{as:mon_h}, which guarantees invertibility of the treatment equation. The third equality holds under Assumption \ref{as:RA_cts}, requiring random assignment of instruments. Finally, the last equality is implied by Assumption \ref{as:Vdist}, since the unobservable $V_i$ can be normalized to follow a uniform distribution on $[0,1]$. Denote the identified support of the propensity score $P_i(Z_{i}, Z_{- i}, D_i)$ as $\mathcal{P}_i$.

After identifying the mapping ${h^{i}_{Z_{i}, Z_{-i}}}^{-1}(D_{i})$ for $i \in \{0, 1\}$ within each group $g$, we can identify the conditional distributions of potential outcomes by using the propensity score functions $P_i(Z_{i}, Z_{-i}, D_{i})$ as control functions:
\begin{equation} \label{eq:cond_po3}
    \begin{aligned}
        &\mathbb{P}\left(Y_{i} \in B \mid D_{i} = d, D_{-i} = d', P_i(Z_{i}, Z_{-i}, D_{i}) = p, P_{-i}(Z_{-i}, Z_{i}, D_{-i}) = p'\right)\\
        =& \mathbb{P}\left(g_i(d, d', U_{i}) \in B \mid h_i(Z_{i}, Z_{-i}, V_{i}) = d, h_{-i}(Z_{-i}, Z_{i}, V_{-i}) = d', \right.\\
        &\left. {h^{i}_{Z_{i}, Z_{-i}}}^{-1}(d) = p, {h^{-i}_{Z_{-i}, Z_{i}}}^{-1}(d') = p'\right) \\
        =& \mathbb{P}\left(g_i(d, d', U_{i}) \in B \mid V_{i} = {h^{i}_{Z_{i}, Z_{-i}}}^{-1}(d), V_{-i} = {h^{-i}_{Z_{-i}, Z_{i}}}^{-1}(d'), \right.\\
        &\left. {h^{i}_{Z_{i}, Z_{-i}}}^{-1}(d) = p, {h^{-i}_{Z_{-i}, Z_{i}}}^{-1}(d') = p'\right) \\ 
        =& \mathbb{P}\left(g(d, d', U_{i}) \in B \mid V_{i} = p, V_{-i} = p', {h^{i}_{Z_{i}, Z_{-i}}}^{-1}(d) = p, {h^{-i}_{Z_{-i}, Z_{i}}}^{-1}(d') = p'\right) \\
        =& \mathbb{P}\left(Y_{i}(d, d') \in B \mid V_{i} = p, V_{-i} = p'\right),
    \end{aligned}
\end{equation}
where $(d, d') \in \mathbb{R}^2$, $(p, p') \in \mathcal{P}_i \times \mathcal{P}_{-i}$, and $B$ denotes any Borel set in the sigma-field generated by $Y_{i}$. The first equality follows directly from Equation \eqref{eq:p_score1_cts}, which links the propensity score to the inverse of the treatment function. The second equality is implied by Assumption \ref{as:mon_h}, which ensures that the treatment equation is strictly monotone in the unobservable and hence invertible. Finally, the last equality holds under Assumption \ref{as:RA_cts}, since the instruments $(Z_i, Z_{-i})$ are independent of the unobserved component $U_i$.

Equation \eqref{eq:cond_po3} provides identification of the marginal treatment response (MTR) function, 
\begin{equation*}
    m_i^{(d, d')}(p, p') \equiv \mathbb{E}\big[Y_i\left(d, d'\right) \mid V_i=p, V_{-i}=p'\big], (d, d') \in \mathbb{R}^2, (p, p') \in \mathcal{P}_i \times \mathcal{P}_{-i},
\end{equation*}
by taking expectations over the identified potential outcome distributions, $\mathbb{P}(Y_{i}(d, d') \in B \mid V_{i} = p, V_{-i} = p')$. Building on this result, the marginal controlled spillover effect (MCSE) and the marginal controlled direct effect (MCDE) with continuous treatments are identified as differences across the MTR functions, corresponding respectively to changes in peers' treatments and in one's own treatment while conditioning on latent characteristics $(V_i, V_{-i})$.

\begin{theorem}(Identifying marginal treatment response) \label{thm:id_mtr_cts_trt}
    Consider the model in Equation \eqref{eq:basic_model_cts_trt}. Suppose that Assumptions \ref{as:Vdist}, \ref{as:RA_cts} and \ref{as:mon_h} hold. For any $d_0, d_1 \in \mathbb{R}^2$ and $(p_0, p_1) \in \mathcal{P}_i \times \mathcal{P}_{-i}$, $m_{i}^{(d_0, d_1)}(p_0, p_1)$ is identified as
    \begin{equation*}
        \mathbb{E}\big[Y_{i} \mid D_{i} = d_0, D_{-i} = d_1, P_i(Z_{i}, Z_{-i}, D_{i}) = p_0, P_{-i}(Z_{-i}, Z_{i}, D_{-i}) = p_1\big],
    \end{equation*}
    where $P_i(\mathbf{z}, \mathbf{z'}, d) \equiv \mathbb{P}(D_{i} \leq d \mid Z_{i} = \mathbf{z}, Z_{- i} = \mathbf{z'})$.
\end{theorem}

\section{Proofs for the Exposure Mapping Model} \label{app:proof_exposure}

The individual propensity score $P_{ig}(Z_g)$ can be used to identify her threshold function $h_i(\cdot)$:
\begin{equation*}
    \begin{aligned}
        & \mathbb{P}\left(D_{ig}=1 \mid Z_g=z\right) \\
        = & \mathbb{P}\left(V_{ig} \leq h_i\left(Z_g\right) \mid Z_g=z\right) \\
        = & \mathbb{P}\left(V_{ig} \leq h_i\left(z\right) \mid Z_g=z\right) \\
        = & \mathbb{P}\left(V_{ig} \leq h_i\left(z\right)\right) \\
        = & h_i\left(z\right).
    \end{aligned}
\end{equation*}
These equalities follow from two key ingredients. First, by normalizing the individual unobservable $V_{ig}$ to be uniformly distributed on $(0, 1)$. Second, because the group-level instrument $Z_g$ is randomly assigned, it is independent of the unobserved heterogeneity. Additionally, the group-level propensity score function identifies the inverse of the exposure function $m$:
\begin{equation*}
    \begin{aligned}
        & \mathbb{P}\left(H_g \leq h \mid Z_g=z\right) \\
        = & \mathbb{P}\left(m\left(z, \varepsilon_g\right) \leq h \mid Z_g=z\right) \\
        = & \mathbb{P}\left(\varepsilon_g \leq m^{-1}_z(h) \mid Z_g=z\right) \\
        = & \mathbb{P}\left(\varepsilon_g \leq m^{-1}_z(h)\right) \\
        = & m^{-1}_z(h).
        \end{aligned}
\end{equation*}
The second equality is implied by Assumption \ref{as:mon_m}, which guarantees that the exposure mapping is strictly monotone in the group-level unobservable and therefore invertible. The third equality follows from Assumption \ref{as:RA_mixed}, which ensures that the instruments are randomly assigned and thus independent of group-level unobserved heterogeneity. Finally, the last equality results from normalizing the group-level unobservable $\varepsilon_g$ to follow a uniform distribution on $(0, 1)$.

Next, we use the individual- and group-level propensity score functions as control functions to identify the conditional distribution of $V_{ig}$ given $\varepsilon_g$,
\begin{equation*}
    \begin{aligned}
        & \mathbb{P}\left(D_{ig} = 1 \mid H_g = h, P_{ig}(Z_g) = p_0, P_g(Z_g, H_g) = p_1\right) \\
        =& \mathbb{P}\left(V_{ig} \leq h_{i}(Z_g) \mid m(Z_g, \varepsilon_g) = h, h_{i}(Z_g) = p_0, m_{Z_g}^{-1}(h) = p_1\right) \\
        =& \mathbb{P}\left(V_{ig} \leq p_0 \mid \varepsilon_g = m_{Z_g}^{-1}(h), h_{i}(Z_g) = p_0, m_{Z_g}^{-1}(h) = p_1\right) \\
        =& \mathbb{P}\left(V_{ig} \leq p_0 \mid \varepsilon_g = p_1, h_{i}(Z_g) = p_0, m_{Z_g}^{-1}(h) = p_1\right) \\
        =& \mathbb{P}\left(V_{ig} \leq p_0 \mid \varepsilon_g = p_1\right),
    \end{aligned}
\end{equation*}
where the first equality follows from the identification of control functions, the second equality is implied by the monotonicity of function $m(\cdot)$, and the last equality holds under the random assignment of instruments $Z_g$. 

Given that the probability $\mathbb{P}(D_{ig} = 1 \mid H_g = h, P_{ig}(Z_g) = p_0, P_g(Z_g, H_g) = p_1)$ is differentiable with respect to $p_0$, we can take derivatives to recover the conditional density of the individual-level unobservable given the group-level unobservable, denoted $f_{V_{ig} \mid \varepsilon_g}(\cdot)$:
\begin{equation*}
    \begin{aligned}
        \frac{\partial}{\partial p_0} \mathbb{P}\left(D_{ig} = 1 \mid H_g = h, P_{ig}(Z_g) = p_0, P_g(Z_g, H_g) = p_1\right) = f_{V_{ig} \mid \varepsilon_g = p_1}(p_0).
    \end{aligned}
\end{equation*}

Then, the marginal treatment response function, $\mathbb{E}[Y_{ig}(1, h) \mid V_{ig} = p_0, \varepsilon_g = p_1]$, can be identified from 
\begin{equation*}
    \begin{aligned}
        &\mathbb{E}\left[Y_{ig} D_{ig} \mid H_g = h, P_{ig}\left(Z_g\right) = p_0, P_g\left(Z_g, H_g\right) = p_1\right] \\
        =& \mathbb{E}\left[Y_{ig} \mathbbm{1}\{V_{ig} \leq h_{i}(Z_g)\} \mid m(Z_g, \varepsilon_g) =h, h_i(Z_g) = p_0, m_{Z_g}^{-1}(h) = p_1\right] \\
        =& \mathbb{E}\left[Y_{ig} \mathbbm{1}\{V_{ig} \leq p_0\} \mid \varepsilon_g=m_{Z_g}^{-1}(h), h_i(Z_g) = p_0, m_{Z_g}^{-1}(h) = p_1\right] \\
        =& \mathbb{E}\left[Y_{ig} \mathbbm{1}\{V_{ig} \leq p_0\} \mid \varepsilon_g=p_1\right].
    \end{aligned}
\end{equation*}
Given that the function $\mathbb{E}[Y_{ig} D_{ig} \mid H_g = h, P_{ig}\left(Z_g\right) = p_0, P_g\left(Z_g, H_g\right) = p_1]$ is differentiable with repect to $p_0$ and that the marginal treatment response functions are continuous, we can differentiate this conditional expectation to obtain
\begin{equation*}
    \begin{aligned}
        & \frac{\partial}{\partial p_0} \mathbb{E}\left[Y_{ig} D_{ig} \mid H_g = h, P_{ig}\left(Z_g\right) = p_0, P_g\left(Z_g, H_g\right) = p_1\right] \\
        =& \mathbb{E}\left[Y_{ig}(1, h) \mid V_{ig} = p_0, \varepsilon_g = p_1\right] \cdot f_{V_{ig} \mid \varepsilon_g = p_1}(p_0).
    \end{aligned}
\end{equation*}
By dividing both sides of the previous expression by the conditional density $f_{V_{ig} \mid \varepsilon_g = p_1}(p_0)$, we can identify $\mathbb{E}[Y_{ig}(1, h) \mid V_{ig} = p_0, \varepsilon_g = p_1]$ for $h \in \mathbb{R}$ and $(p_0, p_1) \in \mathcal{P}$. Analogously, the marginal treatment response function $\mathbb{E}[Y_{ig}(0, h) \mid V_{ig} = p_0, \varepsilon_g = p_1]$ can be identified as 
\begin{equation*}
    \begin{aligned}
        & \frac{\partial}{\partial p_0} \mathbb{E}\left[Y_{ig} (1 - D_{ig}) \mid H_g = h, P_{ig}\left(Z_g\right) = p_0, P_g\left(Z_g, H_g\right) = p_1\right] \Big/ \\
        & \frac{\partial}{\partial p_0} \mathbb{P}\left(1 - D_{ig} \mid H_g = h, P_{ig}\left(Z_g\right) = p_0, P_g\left(Z_g, H_g\right) = p_1\right).
    \end{aligned}
\end{equation*}

\end{appendices}

\end{document}